\documentclass[11pt, a4paper]{article}
\usepackage{siunitx}
\usepackage{rotating}  %

\usepackage[a4paper,top=2cm,bottom=2cm,left=2.5cm,right=2.5cm,marginparwidth=1.75cm]{geometry}
\usepackage[linesnumbered,ruled,vlined]{algorithm2e}

\usepackage{amsmath}
\usepackage{array} %
\usepackage{graphicx}
\usepackage{xcolor} %
\usepackage{kotex}
\usepackage{placeins}
\usepackage[title]{appendix}
\usepackage{mathrsfs}
\usepackage{amsfonts}
\usepackage{amssymb}
\usepackage{amsthm}
\usepackage{booktabs} %
\usepackage{multirow}
\usepackage{caption}  %
\usepackage{threeparttable} %
\usepackage{listings}
\usepackage{enumitem}
\usepackage{chngcntr}
\usepackage{booktabs}
\usepackage{lipsum}
\usepackage{subcaption}
\usepackage{authblk}
\usepackage[T1]{fontenc}    %
\usepackage{csquotes}       %
\usepackage{diagbox}
\usepackage{natbib}

\usepackage{amsthm} %

\newtheorem{remark}{Remark}
\newtheorem{proposition}{Proposition}
\usepackage{setspace}
\usepackage{titlesec}
\titleformat{\section} %
  {\normalfont\Large\bfseries}{\thesection.}{1em}{}

\usepackage{float}   %
\usepackage{caption} %


\title{Online Change-Point Detection in Linear Regression with Autoregressive Errors via Streaming Variational Bayes}
\author[1]{\small Yujin Jeong}
\author[1]{Chaeeun Park}
\author[1]{Sihyeon Pyeon}
\author[1,2]{Seonghyun Jeong\thanks{Corresponding author: \texttt{sjeong@yonsei.ac.kr}}}
\author[1,2]{Taeyoung Park\thanks{Corresponding author: \texttt{tpark@yonsei.ac.kr}}}

\affil[1]{Department of Statistics and Data Science, Yonsei University}
\affil[2]{Department of Applied Statistics, Yonsei University}

\date{} %

\usepackage[colorlinks=true, allcolors=blue]{hyperref}
\usepackage{orcidlink}

\makeatletter
\let\SR@papertitle\@title
\let\SR@paperauthor\@author
\let\SR@titlelayout\@maketitle
\newcommand{\SRsupplementtitle}{%
  \par
  \begingroup
    \let\thanks\@gobble %
    \let\newpage\relax
    \def\@title{\SR@papertitle\\(Supplementary Material)}%
    \let\@author\SR@paperauthor
    \let\@date\@empty
    \renewenvironment{tabular}[2][]{\begin{center}}{\end{center}}%
    \SR@titlelayout
  \endgroup
  \thispagestyle{plain}%
}
\makeatother

\begin{document}
\maketitle

\begin{abstract}
Large-scale streaming time-series regression requires methods that can detect structural changes while accounting for serially correlated errors and quantifying uncertainty about the current regime. We propose a streaming variational Bayes method for linear regression with autoregressive errors that combines online changepoint detection with sequential posterior inference. The proposed method processes observations in chronological mini-batches and uses a finite-mixture variational approximation to represent separate retention and reset configurations for the regression coefficients and innovation variance, while the autoregressive coefficients are shared across regimes and learned sequentially. This construction yields block-specific changepoint probabilities and posterior uncertainty summaries. Closed-form variational updates and mixture pruning enable scalable sequential inference without repeated full-data refits. In simulations, the proposed method maintains low empirical false-alarm rates under stable regimes, with posterior coverage improving as the sample size increases. Under repeated structural changes, it provides accurate type-specific changepoint detection, small changepoint location errors, and short detection delays relative to frequentist sequential monitoring procedures, while posterior coverage approaches nominal levels after an initial adaptation period. An application to high-frequency Bitcoin volatility data illustrates distinct adaptation of the regression and innovation-variance components while accounting for residual autoregressive dependence.
\end{abstract}


\section{Introduction}

Large-scale time-series applications increasingly require regression
inference to be updated as observations arrive, while accounting for
serially correlated errors and possible structural changes. Linear regression
remains attractive because of its simplicity and interpretability, and
time-series regression with autoregressive (AR) errors provides a natural
framework by separating covariate effects from serial dependence remaining
in the regression residuals. When both the cumulative sample size and the
AR order are large, repeatedly refitting a full-data model can become
computationally and memory intensive. At the same time, the regression
coefficients or innovation variance may change abruptly as the data stream
moves between regimes. 
These considerations motivate scalable sequential methods that can perform online changepoint detection while providing uncertainty quantification (UQ) for the current regime.

Structural change creates a statistical difficulty beyond the computational
cost of processing an expanding data stream. Standard posterior recycling
is effective when the data-generating mechanism remains stable, because the
posterior from the preceding batch provides an increasingly informative
summary of past observations. After a regime change, however, a posterior
concentrated using a long history from the preceding regime can dominate the
contribution of incoming batches from the new regime.
{This posterior inertia can slow adaptation and distort
current-regime estimation and UQ.}
Moreover, changes in the regression coefficients and innovation variance
have distinct interpretations and need not occur simultaneously: a change
in the regression coefficients alters the conditional relationship between
the response and covariates, whereas a change in the innovation variance
alters the scale of the remaining unexplained variation. A useful sequential
procedure should therefore be able to weaken inherited information
selectively and represent uncertainty about which parameter block has changed.

Regression with AR errors has a long history in time-series analysis; see,
for example, \citet{boxjenkins1970}, \citet{brockwell1991}, and
\citet{hamilton1994}. Classical approaches include the Cochrane--Orcutt
and Prais--Winsten procedures and likelihood-based estimation
\citep{cochrane1949,prais1954,beach1978}. For Bayesian inference,
\citet{chib1993bayesAR} develops a Gibbs sampler based on a conditionally
transformed regression model. Sequential Monte Carlo and particle-learning
methods provide flexible sequential Bayesian inference
\citep{doucet2001sequential,carvalho2010particle}, although maintaining an
accurate particle approximation can be computationally demanding for
higher-dimensional static parameters. Variational Bayes (VB) provides
an optimization-based alternative
\citep{jordan1999,wainwright2008,blei2017}. Streaming variational Bayes
(SVB) propagates an approximate posterior across data batches
\citep{broderick2013streaming}, while updating variational Bayes provides
related recursive updates using newly arriving data
\citep{tomasetti2022updating}.
Such posterior-recycling schemes can inherit the
posterior inertia described above when the data-generating mechanism changes.

Several strands of methodology address structural change and sequential
nonstationarity. Discrete regime and changepoint models include
Markov-switching formulations and Bayesian multiple-changepoint methods
\citep{hamilton1989regime,chib1998changepoint,
pesaran2006breaks,fearnhead2006exact}. Stochastic variational inference
has also been developed for long sequences in hidden Markov and hidden
semi-Markov models \citep{johnson2014stochastic,foti2014stochastic}.
Online multiple-changepoint filtering and Bayesian online changepoint
detection recursively represent uncertainty about regime transitions
\citep{fearnhead2007online,adams2007bayesian}, with variational extensions
available when segment-level predictive distributions are intractable
\citep{turner2013online}. For dependent time series,
\citet{knoblauch2018spatiotemporal} combine online changepoint detection
with model selection over Bayesian autoregressive models. Frequentist alternatives include fluctuation- and cumulative-sum-based regression monitoring \citep{chu1996monitoring,zeileis2010testing}, including recent sequential monitoring for dynamic regression with autoregressive residuals \citep{niang2026online}, and functional-pruning changepoint detection \citep{romano2023focus}.

Another approach is to adapt by discounting inherited information.
Bayesian forgetting provides a general basis for recursive adaptation
\citep{kulhavy1993forgetting,west1997bayesian}, while power priors control
the influence of previous information through likelihood discounting
\citep{ibrahim2000power}.
{Within SVB, \citet{masegosa2017bayesian} introduce
hierarchical power priors, including an extension with multiple parameter-specific continuous forgetting factors.}
{\citet{moens2018hierarchical} also develops a
hierarchical variational filter for adaptive forgetting.}
Variational estimation of an unknown forgetting factor has been studied
for nonstationary AR models \citep{smidlquinn2004}. More closely related,
\citet{li2021irregular} combine binary change variables, posterior tempering,
and beam search over change histories for Bayesian online adaptation.

The remaining challenge is to combine block-specific online changepoint detection with sequential Bayesian inference and UQ in AR-error regression. Existing approaches
address different parts of this problem but do not directly provide this
combination. Changepoint formulations that restart the full parameter state
at a change can discard stable regression information when only the
innovation variance changes. The linear-regression formulation of
\citet{li2021irregular} tempers the inherited coefficient posterior through
a single change variable but treats the observation-noise variance as fixed.
Parameter-specific continuous forgetting can adapt to drift in individual
parameters, but it does not define a discrete change event, so changepoint
detection would rely on interpreting the inferred discount levels rather
than on probabilities of explicitly defined reset events. Frequentist sequential monitoring instead focuses on detecting parameter
instability, leaving post-change estimation to a separate step. Serial
dependence further complicates the problem, because ignoring
autocorrelation can produce spurious changepoints \citep{romano2022decafs}.

To address this gap, we develop separate-reset streaming variational Bayes (SR-SVB) for linear regression with AR errors. At each subsequent mini-batch, the method expands the retained mixture into candidate components corresponding to the four possible combinations of retaining or resetting the regression and innovation-variance blocks.
A reset discounts, rather than discards, inherited information, allowing Bayesian forgetting to operate separately across the two blocks. The AR coefficients, treated as common across
regimes, are updated with each batch to account for residual serial dependence.
The mixture weights determine the approximation propagated to the next batch and yield approximate marginal and joint reset probabilities for the two blocks. The resulting procedure
provides sequential posterior inference and block-specific changepoint
detection without repeated full-data refits, with per-batch computation and
memory that do not grow with the cumulative sample size.

Our methodological contributions are twofold. First, we develop a block-separable variational construction for online changepoint detection and posterior UQ in linear regression with AR errors. The propagated state uses componentwise independent Gaussian and inverse-gamma factors, which avoid coupling the regression factor to the variance scale and thereby allow inherited regression and innovation-variance information to be discounted separately. Second, we derive closed-form coordinate updates for the component-specific regression and innovation-variance factors, the shared AR factor, and the mixture weights, and use pruning to bound the number of propagated mixture components.
We evaluate SR-SVB under both stable and repeatedly changing regimes.
Under stability, we compare point estimation, posterior coverage, and
computational cost with full-batch references and nonadaptive SVB and
assess empirical false-alarm frequencies.
{Under repeated changes, we compare SR-SVB with two
frequentist sequential monitoring methods using type-specific detection
accuracy, changepoint location error, and detection delay, while also
examining posterior coverage and, through a joint-reset ablation, the
contribution of separate reset states.}
In the Bitcoin realized-volatility application, we illustrate separate
adaptation of the regression and innovation-variance blocks, including a
shift in the fitted U.S.-trading-hours contrast.

The remainder of this paper is organized as follows.
Section~\ref{sec:sec2} introduces the model and prewhitening construction.
Section~\ref{sec:sec3} develops the proposed SR-SVB method.
Section~\ref{sec:sec4} evaluates the proposed method through simulation
studies.
Section~\ref{sec:sec5} presents the real-data analysis.
Section~\ref{sec:sec6} concludes.
Additional derivations, implementation details, and ablation and
sensitivity experiments are provided in the \hyperref[sec:supplement]{supplementary material}.

\section{Model Specification}
\label{sec:sec2}

In this section,
we begin with a static-parameter AR-error regression model, which provides the starting point for the streaming updates.
The streaming inference procedure developed in
Section~\ref{sec:sec3} will additionally allow the regression and
innovation-variance states to adapt as new mini-batches arrive.

Time is indexed discretely by $t\in\{1,2,\ldots\}$, and at each time $t$ we observe $(y_t,\mathbf{x}_t)$, where $y_t\in\mathbb{R}$ denotes the response and $\mathbf{x}_t=(x_{t,1},\ldots,x_{t,k})^\top\in\mathbb{R}^k$ denotes the $k$-dimensional covariate vector. Mini-batches are indexed by $b\in\mathbb{N}$; batch $b$ starts at time $t_b$ and contains $n_b$ consecutive observations, with index set 
\begin{align*}
T_b:=\{t_b,t_b+1,\ldots,t_b+n_b-1\}.    
\end{align*}
We assume batches are consecutive and non-overlapping, so that $t_1=1$ and $t_{b+1}=t_b+n_b$ for $b\geq1$. After processing batch $B$, define $T_{(B)}:=\bigcup_{b=1}^B T_b$ and $n_{(B)}:=\sum_{b=1}^B n_b$ as the cumulative time index set and sample size, respectively. By construction, $T_{(B)}=\{1,2,\ldots,n_{(B)}\}$.

Under the static-parameter model, for
$t\in T_{(B)}$ with $t>p$, we model the response as
\begin{align}
\label{eq:linearreg}
y_t
=
\mathbf{x}_t^\top\boldsymbol{\beta}
+
\varepsilon_t,
\quad
\phi(L)\varepsilon_t=u_t,
\quad
u_t \overset{\mathrm{iid}}{\sim} \mathrm{N}(0,\sigma^2),
\end{align}
where $\boldsymbol{\beta}=(\beta_1,\ldots,\beta_k)^\top\in\mathbb{R}^k$ and $\boldsymbol{\phi}=(\phi_1,\ldots,\phi_p)^\top\in\mathbb{R}^p$ denote the regression and AR coefficients, respectively, $\phi(L):=1-\sum_{i=1}^p\phi_iL^i$ is the AR lag polynomial, and $L^i\varepsilon_t=\varepsilon_{t-i}$ is the backshift operator. In our specification, we do not constrain the AR coefficients $\boldsymbol{\phi}$ to satisfy weak stationarity. 

To facilitate standard linear regression updates and posterior computation, we transform the model to have independent errors. Let $T_{(B)}^\ast:=\{t\in T_{(B)}:t>p\}$ denote the subset of $T_{(B)}$ obtained by excluding the first $p$ observations used to initialize the lag structure. We adopt the conditional likelihood that conditions on these initial $p$ values; thus, the likelihood is computed only for observations with time indices in $T_{(B)}^\ast$, and we assume $n_1>p$. We then apply a prewhitening transformation based on the lag polynomial $\phi(L)$. For all $t\in T_{(B)}^\ast$,
\begin{align}
\label{eq:pre-whiten}
\begin{split}
y_t^\ast
&=
\phi(L)y_t
=
y_t-\sum_{j=1}^p\phi_jy_{t-j},\\
\mathbf{x}_t^\ast
&=
\phi(L)\mathbf{x}_t
=
\mathbf{x}_t-\sum_{j=1}^p\phi_j\mathbf{x}_{t-j},
\end{split}
\end{align}
which gives the conditionally independent regression representation $y_t^\ast=\mathbf{x}_t^{\ast\top}\boldsymbol{\beta}+u_t$, where $u_t\sim\mathrm{N}(0,\sigma^2)$.
For each batch $b\leq B$, define
$T_b^\ast:=\{t\in T_b:t>p\}$ and $n_b^\ast:=|T_b^\ast|$.
Since $t_1=1$ and $n_1>p$, we have $n_1^\ast=n_1-p$, whereas
$n_b^\ast=n_b$ for $b\geq2$.  Collect the transformed observations as
\begin{align*}
    \mathbf{y}_b^\ast & :=(y_t^\ast)_{t\in T_b^\ast}\in\mathbb{R}^{n_b^\ast},\\
    \mathbf{X}_b^\ast&:=[\mathbf{x}_t^{\ast\top}]_{t\in T_b^\ast}\in\mathbb{R}^{n_b^\ast\times k}.
\end{align*}
It follows that
$\mathbf{y}_b^\ast\mid\mathbf{X}_b^\ast,\boldsymbol{\beta},\sigma^2,\boldsymbol{\phi}
\sim
\mathrm{N}(\mathbf{X}_b^\ast\boldsymbol{\beta},
\sigma^2\mathbf{I}_{n_b^\ast})$.
Up to batch $B$, we write the cumulative raw response vector and design matrix as $\mathbf{y}_{(B)}:=(y_t)_{t\in T_{(B)}}\in\mathbb{R}^{n_{(B)}}$ and $\mathbf{X}_{(B)}:=[\mathbf{x}_t^\top]_{t\in T_{(B)}}\in\mathbb{R}^{n_{(B)}\times k}$. 
For any batch $b$, we use $\mathbf{y}_{b\mid T_b^\ast}:=(y_t)_{t\in T_b^\ast}$ and $\mathbf{X}_{b\mid T_b^\ast}:=[\mathbf{x}_t^\top]_{t\in T_b^\ast}$ to denote the corresponding unwhitened responses and covariates prior to prewhitening. Note that for $t\in T_b^\ast$ and any $j\in\{1,\ldots,p\}$, the lagged index $t-j$ in \eqref{eq:pre-whiten} may lie in an earlier batch. Hence, the streaming implementation retains the last $p$ raw observations from previous batches to evaluate $y_t^\ast$ and $\mathbf{x}_t^\ast$.

When the regression state changes across a mini-batch boundary, lagged residuals from the preceding batch need not correspond to the current regression state. For $b\geq2$, let $\boldsymbol{\beta}_b$ and $\sigma_b^2$ denote the regression and innovation-variance states in batch $b$, and let $\boldsymbol{\beta}_{b-1}$ denote the regression state associated with the immediately preceding batch. We assume $n_{b-1}\geq p$, so that any lag crossing the current boundary lies in batch $b-1$. The boundary representation underlying the two treatments considered below is
\begin{align}
\label{eq:boundary-working}
y_t^\ast
&=
\mathbf{x}_t^{\ast\top}\boldsymbol{\beta}_b
+
\sum_{\substack{1\leq j\leq p\\t-j<t_b}}
\phi_j\mathbf{x}_{t-j}^\top
\left(
\boldsymbol{\beta}_b-\boldsymbol{\beta}_{b-1}
\right)
+
u_t,
\quad
t\in T_b^\ast ,
\end{align}
where $u_t\sim\mathrm{N}(0,\sigma_b^2)$. Thus, the preceding regression state enters only through lags that cross the current batch boundary; once all $p$ lags lie within batch $b$, the additional term in \eqref{eq:boundary-working} vanishes. For the first batch, no preceding regression state is available, and the representation reduces to the ordinary prewhitened model in \eqref{eq:pre-whiten}.

For the main implementation, we replace the preceding regression state in
\eqref{eq:boundary-working} by
$\bar{\boldsymbol{\beta}}_{b-1}$, the mean regression state carried forward
from the preceding batch. This quantity is held fixed throughout the
batch-$b$ update. We refer to this plug-in construction as the
previous-mean (PM) boundary treatment.
For subsequent batchwise updates, we absorb the resulting boundary correction into the transformed response and design. For $b\geq2$ and $t\in T_b^\ast$, define
\begin{align*}
\tilde y_t^\ast
&:=
y_t^\ast
+
\sum_{\substack{1\leq j\leq p\\t-j<t_b}}
\phi_j\mathbf{x}_{t-j}^\top
\bar{\boldsymbol{\beta}}_{b-1},
\\
\tilde{\mathbf{x}}_t^\ast
&:=
\mathbf{x}_t^\ast
+
\sum_{\substack{1\leq j\leq p\\t-j<t_b}}
\phi_j\mathbf{x}_{t-j}.
\end{align*}
Then \eqref{eq:boundary-working}, under the PM treatment, is equivalently written as
\[
\tilde y_t^\ast
=
\tilde{\mathbf{x}}_t^{\ast\top}\boldsymbol{\beta}_b+u_t.
\]
For $b=1$, we set
$\tilde y_t^\ast=y_t^\ast$ and
$\tilde{\mathbf{x}}_t^\ast=\mathbf{x}_t^\ast$.
Collecting these quantities over $T_b^\ast$, define
\begin{align*}
\tilde{\mathbf{y}}_b^\ast
&:=
(\tilde y_t^\ast)_{t\in T_b^\ast}
\in\mathbb{R}^{n_b^\ast},
\\
\tilde{\mathbf{X}}_b^\ast
&:=
[\tilde{\mathbf{x}}_t^{\ast\top}]_{t\in T_b^\ast}
\in\mathbb{R}^{n_b^\ast\times k}.
\end{align*}
The PM boundary-adjusted working likelihood is therefore
\begin{align*}
\tilde{\mathbf{y}}_b^\ast
\mid
\tilde{\mathbf{X}}_b^\ast,
\boldsymbol{\beta}_b,
\sigma_b^2,
\boldsymbol{\phi}
&\sim
\mathrm{N}\!\left(
\tilde{\mathbf{X}}_b^\ast\boldsymbol{\beta}_b,
\sigma_b^2\mathbf{I}_{n_b^\ast}
\right).
\end{align*}

An alternative full-joint (FJ) treatment retains the preceding regression state as a local latent quantity rather than replacing it by a fixed posterior summary. 
For regression-reset candidates, the current and preceding regression states are updated jointly so that uncertainty in the preceding state and its dependence with the current state are reflected in the current-batch update.
We use the computationally faster PM treatment as the default implementation in the main text and provide the FJ treatment as an uncertainty-aware extension. Detailed derivations and simulation comparisons between PM and FJ are provided in Supplementary Sections~\ref{app:variational_fj}, \ref{app:elbo_fj}, and~\ref{sec:supp_pm_fj}.

To accommodate possible structural changes across mini-batches, we allow the information propagated for the regression coefficients and innovation variance to be selectively retained or discounted at batch boundaries. The corresponding latent reset mechanism and finite-mixture variational representation are introduced in Section~\ref{sec:sec3}.
\section{Separate-Reset SVB for Regression with AR(\texorpdfstring{$p$}{p}) Errors}
\label{sec:sec3}

In this section, we develop the SR-SVB method for the linear regression model with AR($p$) errors introduced in
Section~\ref{sec:sec2}. Section~\ref{subsec:recursive_state} introduces
the no-reset benchmark, baseline prior, and recursively propagated state.
Section~\ref{subsec:reset_prior_vf} constructs the separate-reset
pseudo-prior and corresponding finite-mixture variational family, and
Section~\ref{subsec:subsec3.2} derives the batchwise evidence lower bound
(ELBO), coordinate-ascent variational inference (CAVI) updates, pruning, and
posterior propagation.

\subsection{Recursive State Representation}
\label{subsec:recursive_state}
\label{subsec:variational_family}

We first recall the ordinary sequential Bayesian update under the
static-parameter model. Let
$\pi^{(0)}(\boldsymbol{\beta},\sigma^2,\boldsymbol{\phi})$
denote the baseline prior, specified below in
\eqref{eq:prior}. For $b\geq2$, the exact posterior after observing
batches $1,\ldots,b$ satisfies
\begin{equation}
\begin{aligned}
p(\boldsymbol{\beta},\sigma^2,\boldsymbol{\phi}
\mid \mathbf y_{(b)},\mathbf X_{(b)})
\propto
p(\mathbf y_b^\ast\mid
\mathbf X_b^\ast,
\boldsymbol{\beta},\sigma^2,\boldsymbol{\phi})
\,p(\boldsymbol{\beta},\sigma^2,\boldsymbol{\phi}
\mid \mathbf y_{(b-1)},\mathbf X_{(b-1)}).
\label{eq:sequential_posterior}
\end{aligned}
\end{equation}
This decomposition provides the no-reset benchmark for the streaming
recursion. Under this static-parameter recursion, past
likelihood contributions are summarized by the posterior from the
preceding batch, while the incoming batch contributes through its
conditional likelihood. The last $p$ raw observations are also retained
to construct the AR lag terms in this likelihood. In our
streaming variational construction, however, the preceding posterior is
replaced by a propagated variational approximation, and its inherited
parameter states are allowed to be selectively retained or reset.
Operationally, this construction requires two ingredients: a baseline
prior for initialization and a recursively propagated approximation
carrying forward inherited information. We define these ingredients below
and then use them to construct the adaptive batchwise target.

Before any data are observed, we initialize the recursion with the
independent baseline prior
\begin{align}
\pi^{(0)}(\boldsymbol{\beta},\sigma^2,\boldsymbol{\phi})
&=
\pi_\beta^{(0)}(\boldsymbol{\beta})
\pi_{\sigma^2}^{(0)}(\sigma^2)
\pi_\phi^{(0)}(\boldsymbol{\phi}),
\label{eq:prior}
\end{align}
where
\begin{align*}
\pi_\beta^{(0)}(\boldsymbol{\beta})
&=
\mathrm N_k\!\left(
\boldsymbol{\beta}\mid
\boldsymbol{\beta}_0^{(0)},
\left(\mathbf A_0^{(0)}\right)^{-1}
\right),\\
\pi_{\sigma^2}^{(0)}(\sigma^2)
&=
\mathrm{IG}\!\left(
\sigma^2\mid
\nu_0^{(0)}/2,
\delta_0^{(0)}/2
\right),\\
\pi_\phi^{(0)}(\boldsymbol{\phi})
&=
\mathrm N_p\!\left(
\boldsymbol{\phi}\mid
\boldsymbol{\phi}_0^{(0)},
\left(\boldsymbol{\Phi}_0^{(0)}\right)^{-1}
\right).
\end{align*}
Throughout, $\mathrm{IG}(a,d)$ denotes the inverse-gamma distribution
with shape $a$ and scale $d$, whose density is proportional to
$x^{-a-1}\exp(-d/x)$ for $x>0$. For the AR coefficients
$\boldsymbol{\phi}$, we adapt the Minnesota prior from Bayesian vector
autoregressions to the univariate AR($p$) error process
\citep{litterman1986forecasting,doan1984forecasting}. Specifically, we
set $\boldsymbol{\phi}_0^{(0)}=\mathbf 0$ and choose
$\boldsymbol{\Phi}_0^{(0)}
=\operatorname{diag}\left(\tau^{-2}1^\alpha,\ldots,\tau^{-2}p^\alpha\right)$,
where $\tau^2>0$ controls the overall scale and $\alpha\geq0$ controls
lag decay, with $\alpha=0$ corresponding to equal prior scaling across
lags. In our experiments, we use a diffuse baseline with $\tau^2=10^6$
and $\alpha=0$. For the remaining hyperparameters, we set
$\boldsymbol{\beta}_0^{(0)}=\mathbf 0$,
$\mathbf A_0^{(0)}=10^{-6}\mathbf I_k$, and
$(\nu_0^{(0)},\delta_0^{(0)})=(0.02,2)$.

To initialize the recursion at batch $b=1$, set
$K_0=1$ and $\omega_{0,1}=1$, and let
$q_{\beta,1}^{(0)}=\pi_\beta^{(0)}$,
$q_{\sigma^2,1}^{(0)}=\pi_{\sigma^2}^{(0)}$, and
$q_\phi^{(0)}=\pi_\phi^{(0)}$. Equivalently, $(\boldsymbol{\beta}_{0,1}^{(0)},\mathbf A_{0,1}^{(0)},
\nu_{0,1}^{(0)},\delta_{0,1}^{(0)})
=
(\boldsymbol{\beta}_0^{(0)},\mathbf A_0^{(0)},
\nu_0^{(0)},\delta_0^{(0)}).$
With this convention, $q^{(0)}(\boldsymbol{\beta},\sigma^2,\boldsymbol{\phi})
=
\pi^{(0)}(\boldsymbol{\beta},\sigma^2,\boldsymbol{\phi}).$
For a generic batch $b\geq1$, let $q^{(b-1)}$ denote the retained
approximation available at the start of batch $b$. For $b=1$, this is
the initialized baseline state $q^{(0)}$. For $b\geq2$, it is the
post-pruning output of batch $b-1$, after the retained candidate
children have been renormalized and relabeled as
$r=1,\ldots,K_{b-1}$. 
We assume inductively that this retained state
has the form
\begin{align}
q^{(b-1)}
(\boldsymbol{\beta},\sigma^2,\boldsymbol{\phi})
&=
\left(
\sum_{r=1}^{K_{b-1}}
\omega_{b-1,r}\,
q_{\beta,r}^{(b-1)}(\boldsymbol{\beta})
q_{\sigma^2,r}^{(b-1)}(\sigma^2)
\right)
q_{\phi}^{(b-1)}(\boldsymbol{\phi}),
\label{eq:retained_posterior}
\end{align}
where
\begin{align*}
q_{\beta,r}^{(b-1)}(\boldsymbol{\beta})
&=
\mathrm N_k\!\left(
\boldsymbol{\beta}\mid
\boldsymbol{\beta}_{0,r}^{(b-1)},
\left(\mathbf A_{0,r}^{(b-1)}\right)^{-1}
\right),\\
q_{\sigma^2,r}^{(b-1)}(\sigma^2)
&=
\mathrm{IG}\!\left(
\sigma^2\mid
\nu_{0,r}^{(b-1)}/2,
\delta_{0,r}^{(b-1)}/2
\right),\\
q_{\phi}^{(b-1)}(\boldsymbol{\phi})
&=
\mathrm N_p\!\left(
\boldsymbol{\phi}\mid
\boldsymbol{\phi}_0^{(b-1)},
\left(\boldsymbol{\Phi}_0^{(b-1)}\right)^{-1}
\right).
\end{align*}
Here, $K_{b-1}$ is the number of components in the retained state
available at the start of batch $b$, and
$r=1,\ldots,K_{b-1}$ indexes these retained parent components.
For $b\geq2$, they are precisely the candidate children retained after
pruning at batch $b-1$; for $b=1$, the single parent $r=1$ is the
initialized baseline component. The quantity $\omega_{b-1,r}$ is the
normalized mixture weight of retained parent $r$, with
$\omega_{b-1,r}\geq0$ and
$\sum_{r=1}^{K_{b-1}}\omega_{b-1,r}=1$. The factors for
$\boldsymbol{\beta}$ and $\sigma^2$ are component-specific, whereas
the AR factor $q_{\phi}^{(b-1)}$ is shared across components. The
subscript $0$ on the component parameters marks quantities stored for
use as prior hyperparameters in batch $b$, and the superscript records
the batch after which they were obtained. In particular,
$\mathbf A_{0,r}^{(b-1)}$ is the absolute precision of
$\boldsymbol{\beta}$ and is not scaled by $\sigma^2$.
This componentwise independence is deliberate rather than merely
computational. Under a conjugate normal--inverse-gamma component, the
conditional covariance of $\boldsymbol{\beta}$ is
$\sigma^2\mathbf A^{-1}$; hence, modifying the inverse-gamma state also
changes the marginal distribution of $\boldsymbol{\beta}$, even when
its conditional Gaussian parameters are retained. Using an absolute
Gaussian precision instead makes the propagated pseudo-prior reset
mechanism block-separable, so that the regression and
innovation-variance factors can be retained or discounted separately.
This choice makes $\boldsymbol{\beta}$ and $\sigma^2$ independent
within each component, while marginal dependence can still arise after
mixing over components. Equation~\eqref{eq:retained_posterior} therefore
provides the recursively propagated state.
For the PM boundary treatment, the corresponding retained regression
mean used at batch $b\geq2$ is
\begin{align}
\bar{\boldsymbol{\beta}}_{b-1}
&=
\mathbb E_{q^{(b-1)}}(\boldsymbol{\beta})
=
\sum_{r=1}^{K_{b-1}}
\omega_{b-1,r}\,
\boldsymbol{\beta}_{0,r}^{(b-1)}.
\label{eq:retained_beta_mean}
\end{align}

\subsection{Separate-Reset Pseudo-Prior and Finite-Mixture Variational Family}
\label{subsec:reset_prior_vf}

The block-separable retained state in
\eqref{eq:retained_posterior} allows the regression and
innovation-variance states to be retained or discounted independently.
We place reset coordinates on these two blocks only. This reflects the
inferential scope of the present formulation: structural changes are modeled
in the regression relationship and innovation scale, while the residual AR
structure is treated as a common dependence mechanism across regimes.
Allowing $\boldsymbol{\phi}$ to reset would make changes in residual
dependence an additional source of adaptation, complicating attribution of
the current-batch evidence to the regression and innovation-variance blocks.
The common factor $q_{\phi}^{(b-1)}$ is nevertheless updated after every
batch, so $\boldsymbol{\phi}$ continues to be learned sequentially rather
than being held fixed. Sharing this factor also allows the associated
prewhitening calculations to be reused across candidate children. Abrupt
changes in residual dependence are outside the scope of the present
formulation.

Accordingly, we treat each retained component as a parent and generate
candidate children under each retention and reset configuration. Define
\begin{align*}
\mathcal Z_1
&:=\{(0,0)\},\quad
\mathcal Z_b
:=\{0,1\}^2,
\quad b\geq2,
\end{align*}
and, for every $b\geq1$, let
\begin{align*}
\mathcal G_b
&=
\left\{
(r,z_\beta,z_{\sigma^2}):
r=1,\ldots,K_{b-1},\;
(z_\beta,z_{\sigma^2})\in\mathcal Z_b
\right\}.
\end{align*}
Thus, $\mathcal G_1=\{(1,0,0)\}$ is a singleton. Let $G_b$
denote the latent child index taking values in $\mathcal G_b$. For each
$g\in\mathcal G_b$, write
\[
g=(r_g,z_{\beta,g},z_{\sigma^2,g}),
\]
where $r_g$ denotes the parent index and
$z_{\beta,g},z_{\sigma^2,g}\in\{0,1\}$ denote the two reset indicators.
The value zero denotes retention, whereas the value one denotes a reset
implemented by discounting the inherited information. For a given parent
$r$, the four configurations $(r,0,0)$, $(r,1,0)$, $(r,0,1)$, and $(r,1,1)$
correspond to no reset, a regression reset, an innovation-variance reset,
and a joint reset, respectively.

For $b\geq2$, assign independent prior reset probabilities
$h_\beta,h_{\sigma^2}\in(0,1)$ to the two reset coordinates and define
\begin{align*}
\rho_\beta(z)
&=
h_\beta^z(1-h_\beta)^{1-z},
\quad
\rho_{\sigma^2}(z)
=
h_{\sigma^2}^z(1-h_{\sigma^2})^{1-z}.
\end{align*}
The prior weight assigned to child
$g\in\mathcal G_b$ is
\begin{align}
\tilde{\omega}_{b,g}
&=
\begin{cases}
1,
& b=1,\\
\omega_{b-1,r_g}\,
\rho_\beta(z_{\beta,g})\,
\rho_{\sigma^2}(z_{\sigma^2,g}),
& b\geq2.
\end{cases}
\label{eq:child_prior_weight}
\end{align}
For $b\geq2$, the weight in \eqref{eq:child_prior_weight} combines
the posterior weight of the parent with the prior probabilities of
its two reset coordinates. Since
$\sum_{r=1}^{K_{b-1}}\omega_{b-1,r}=1$ and
$\sum_{z=0}^1\rho_\beta(z)
=
\sum_{z=0}^1\rho_{\sigma^2}(z)=1$,
the child prior weights satisfy
$\sum_{g\in\mathcal G_b}\tilde{\omega}_{b,g}=1$. 

Throughout the batch-$b$ PM update below, we suppress the batch
subscript and write $\boldsymbol{\beta}$ for the current regression
state $\boldsymbol{\beta}_b$. For the regression state, let
$(\boldsymbol{\beta}_\star,\mathbf A_\star)$
denote a weak Gaussian baseline with
$\mathbf A_\star\succ\mathbf 0$, and let
$\eta_\beta\in(0,1]$ be the regression-discount factor.
For $g\in\mathcal G_b$, define
\begin{equation}
\begin{aligned}
\tilde{\mathbf A}_{b,g}
&=
\begin{cases}
\mathbf A_{0,r_g}^{(b-1)},
& z_{\beta,g}=0,\\
\mathbf A_\star+
\eta_\beta\mathbf A_{0,r_g}^{(b-1)},
& z_{\beta,g}=1,
\end{cases}
\\
\tilde{\boldsymbol{\beta}}_{b,g}
&=
\begin{cases}
\boldsymbol{\beta}_{0,r_g}^{(b-1)},
& z_{\beta,g}=0,\\
\tilde{\mathbf A}_{b,g}^{-1}
\left(
\mathbf A_\star\boldsymbol{\beta}_\star+
\eta_\beta
\mathbf A_{0,r_g}^{(b-1)}
\boldsymbol{\beta}_{0,r_g}^{(b-1)}
\right),
& z_{\beta,g}=1.
\end{cases}
\end{aligned}
\label{eq:child_regression_state}
\end{equation}
The $z_{\beta,g}=0$ case retains the inherited Gaussian information
as the current-batch regression pseudo-prior, whereas
$z_{\beta,g}=1$ discounts the inherited precision and
precision-weighted mean through $\eta_\beta$ and combines them with
the weak baseline. By default, we take
$(\boldsymbol{\beta}_\star,\mathbf A_\star)
=(\boldsymbol{\beta}_0^{(0)},\mathbf A_0^{(0)})$.

For the innovation-variance state, let $\eta_{\sigma^2}\in(0,1]$
be the variance-reset discount factor, with $\eta_{\sigma^2}=1$
corresponding to the no-discount limit. Define
\begin{equation}
\begin{aligned}
\tilde{\nu}_{b,g}
&=
2+
\eta_{\sigma^2}^{z_{\sigma^2,g}}
\left(
\nu_{0,r_g}^{(b-1)}-2
\right),\\
\tilde{\delta}_{b,g}
&=
\eta_{\sigma^2}^{z_{\sigma^2,g}}
\delta_{0,r_g}^{(b-1)}.
\end{aligned}
\label{eq:child_variance_state}
\end{equation}
The $z_{\sigma^2,g}=0$ case in
\eqref{eq:child_variance_state} retains the inherited inverse-gamma state. When
$z_{\sigma^2,g}=1$, the transformation discounts its effective information
by $\eta_{\sigma^2}$ while preserving the inherited mean whenever
$\nu_{0,r_g}^{(b-1)}>2$:
\begin{align*}
\frac{\tilde{\delta}_{b,g}}
     {\tilde{\nu}_{b,g}-2}
&=
\frac{
\eta_{\sigma^2}^{z_{\sigma^2,g}}\delta_{0,r_g}^{(b-1)}
}{
\eta_{\sigma^2}^{z_{\sigma^2,g}}(\nu_{0,r_g}^{(b-1)}-2)
}
=
\frac{\delta_{0,r_g}^{(b-1)}}
     {\nu_{0,r_g}^{(b-1)}-2}.
\end{align*}
This mean-preserving discount does not assume a particular direction
of variance change: it weakens the inherited information without
shifting its center, allowing the current-batch likelihood to determine
the new variance level.

The resulting child-specific pseudo-prior factors are
\begin{align*}
\tilde{\pi}_{\beta,g}^{(b)}(\boldsymbol{\beta})
&=
\mathrm N_k\!\left(
\boldsymbol{\beta}\mid
\tilde{\boldsymbol{\beta}}_{b,g},
\tilde{\mathbf A}_{b,g}^{-1}
\right),\\
\tilde{\pi}_{\sigma^2,g}^{(b)}(\sigma^2)
&=
\mathrm{IG}\!\left(
\sigma^2\mid
\tilde{\nu}_{b,g}/2,
\tilde{\delta}_{b,g}/2
\right).
\end{align*}
Combining these factors with the shared propagated AR factor gives
\begin{align*}
\tilde{\pi}_{g}^{(b)}
(\boldsymbol{\beta},\sigma^2,\boldsymbol{\phi})
&=
\tilde{\pi}_{\beta,g}^{(b)}(\boldsymbol{\beta})\,
\tilde{\pi}_{\sigma^2,g}^{(b)}(\sigma^2)\,
q_{\phi}^{(b-1)}(\boldsymbol{\phi}).
\end{align*}
Together with \eqref{eq:child_prior_weight}, this construction defines
the batch-$b$ prior hierarchically as
\begin{align*}
(\boldsymbol{\beta},\sigma^2,\boldsymbol{\phi})
\mid G_b=g
&\sim
\tilde{\pi}_{g}^{(b)},
\quad
g\in\mathcal G_b.
\end{align*}
For $b=1$, this construction reduces to the baseline prior
$\pi^{(0)}.$ 

For every $b\geq1$, under the PM boundary treatment, the hierarchical
construction induces an adaptive batchwise target. To distinguish it
from the exact posterior under the static-parameter model in
\eqref{eq:sequential_posterior}, let $\tilde p^{(b)}$ denote the
corresponding pseudo-posterior. Its density satisfies
\begin{align}
&\tilde p^{(b)}
\left(
G_b=g,
\boldsymbol{\beta},
\sigma^2,
\boldsymbol{\phi}
\mid
\mathbf{y}_{(b)},\mathbf{X}_{(b)}
\right)
\propto
\tilde{\omega}_{b,g}\,
p\!\left(
\tilde{\mathbf y}_b^\ast
\mid
\tilde{\mathbf X}_b^\ast,
\boldsymbol{\beta},
\sigma^2,
\boldsymbol{\phi}
\right)\,
\tilde{\pi}_{g}^{(b)}
(\boldsymbol{\beta},\sigma^2,\boldsymbol{\phi}),
\quad
g\in\mathcal G_b.
\label{eq:augmented_batch_target}
\end{align}
Here,
$\tilde{\mathbf y}_b^\ast$ and
$\tilde{\mathbf X}_b^\ast$
are the PM boundary-adjusted response and design matrix defined in
Section~\ref{sec:sec2}. For $b=1$, they reduce to the ordinary
prewhitened quantities $\mathbf y_1^\ast$ and $\mathbf X_1^\ast$.
We approximate the adaptive target in
\eqref{eq:augmented_batch_target} within the following finite-mixture
variational family. 
We use $\tilde q^{(b)}$ to denote the
candidate approximation on the full child set $\mathcal G_b$ before
pruning, and reserve $q^{(b)}$ for the retained approximation obtained
after pruning:
\begin{align}
&\tilde q^{(b)}
\left(
G_b=g,
\boldsymbol{\beta},
\sigma^2,
\boldsymbol{\phi}
\right)
=
\omega_{b,g}\,
q_{\beta,g}^{(b)}(\boldsymbol{\beta})\,
q_{\sigma^2,g}^{(b)}(\sigma^2)\,
q_{\phi}^{(b)}(\boldsymbol{\phi}),
\quad
g\in\mathcal G_b.
\label{eq:augmented_independent_vf}
\end{align}
Here, $\omega_{b,g}\geq0$ and
$\sum_{g\in\mathcal G_b}\omega_{b,g}=1$. The factors $q_{\beta,g}^{(b)}$, $q_{\sigma^2,g}^{(b)}$, and
$q_{\phi}^{(b)}$ are taken to be multivariate Gaussian, inverse-gamma,
and multivariate Gaussian, respectively, with $q_{\phi}^{(b)}$ shared
across all children. Their batch-specific parameterizations are introduced in
Section~\ref{subsec:subsec3.2} together with the coordinate-ascent
updates. Marginalizing \eqref{eq:augmented_independent_vf} over $G_b$
gives
\begin{align}
\tilde q^{(b)}
(\boldsymbol{\beta},\sigma^2,\boldsymbol{\phi})
&=
\left(
\sum_{g\in\mathcal G_b}
\omega_{b,g}\,
q_{\beta,g}^{(b)}(\boldsymbol{\beta})
q_{\sigma^2,g}^{(b)}(\sigma^2)
\right)
q_{\phi}^{(b)}(\boldsymbol{\phi}).
\label{eq:independent_vf_sec3}
\end{align}

The distribution $\tilde q^{(b)}$ in
\eqref{eq:independent_vf_sec3} is the candidate variational
approximation before pruning. After the batchwise optimization
converges, a subset of children is retained and its weights are
renormalized to obtain $q^{(b)}$. By construction, $q^{(b)}$ has the
same form as \eqref{eq:retained_posterior}, with $b$ replacing $b-1$,
and serves as the propagated state for batch $b+1$. The detailed
pruning, relabeling, and state-propagation rules are given in
Section~\ref{subsec:subsec3.2}.

\begin{remark}[No-reset reduction]
For $b\geq2$ and $g=(r,0,0)$, the inherited component is unchanged,
so that
$\tilde{\pi}_{(r,0,0)}^{(b)}
=
q_{\beta,r}^{(b-1)}
q_{\sigma^2,r}^{(b-1)}
q_{\phi}^{(b-1)}$.
Since $(1-h_\beta)(1-h_{\sigma^2})$ does not depend on $r$,
restricting to the no-reset children and renormalizing over the parent
index leaves the parent weights equal to $\omega_{b-1,r}$. Therefore,
\[
\sum_{r=1}^{K_{b-1}}
\omega_{b-1,r}\,
\tilde{\pi}_{(r,0,0)}^{(b)}
(\boldsymbol{\beta},\sigma^2,\boldsymbol{\phi})
=
q^{(b-1)}
(\boldsymbol{\beta},\sigma^2,\boldsymbol{\phi}).
\]
Thus, restricting to the no-reset children recovers
$q^{(b-1)}$ as the propagated pseudo-prior mixture. Under the PM
boundary treatment, however, the batch-$b$ update continues to use
the PM boundary-adjusted working likelihood defined in
Section~\ref{sec:sec2}. Hence, the no-reset reduction corresponds to
posterior recycling under the PM working likelihood and does not, in
general, coincide with the static-parameter recursion in
\eqref{eq:sequential_posterior}.
\end{remark}

\subsection{Batchwise Coordinate-Ascent Updates}
\label{subsec:subsec3.2}

We next derive the batchwise optimization and propagation steps for
the default PM implementation.
At batch $b$, the variational optimization approximates the adaptive
pseudo-posterior $\tilde p^{(b)}$ defined in
\eqref{eq:augmented_batch_target} by the pre-pruning candidate
distribution $\tilde q^{(b)}$ in
\eqref{eq:augmented_independent_vf}. For a fixed candidate set
$\mathcal G_b$, the batchwise ELBO is
\begin{equation}
\begin{aligned}
\mathcal L_b(q)
&=
\sum_{g\in\mathcal G_b}
\omega_{b,g}
\mathbb E_{
q_{\beta,g}^{(b)}
q_{\sigma^2,g}^{(b)}
q_{\phi}^{(b)}
}
\left[
\log
p\!\left(
\tilde{\mathbf y}_b^\ast
\mid
\tilde{\mathbf X}_b^\ast,
\boldsymbol{\beta},
\sigma^2,
\boldsymbol{\phi}
\right)
\right]
\\
&\quad-
\sum_{g\in\mathcal G_b}
\omega_{b,g}
\operatorname{KL}\!\left(
q_{\beta,g}^{(b)}
\,\middle\|\,
\tilde{\pi}_{\beta,g}^{(b)}
\right)
-
\sum_{g\in\mathcal G_b}
\omega_{b,g}
\operatorname{KL}\!\left(
q_{\sigma^2,g}^{(b)}
\,\middle\|\,
\tilde{\pi}_{\sigma^2,g}^{(b)}
\right)
\\
&\quad-
\operatorname{KL}\!\left(
q_{\phi}^{(b)}
\,\middle\|\,
q_{\phi}^{(b-1)}
\right)
+
\sum_{g\in\mathcal G_b}
\omega_{b,g}\log\tilde{\omega}_{b,g}
-
\sum_{g\in\mathcal G_b}
\omega_{b,g}\log\omega_{b,g}.
\end{aligned}
\label{eq:independent_mixture_elbo}
\end{equation}
We maximize \eqref{eq:independent_mixture_elbo} by coordinate ascent
over the variational factors and mixture weights. 

First, for fixed $q_{\sigma^2,g}^{(b)}$ and $q_{\phi}^{(b)}$, maximizing the
ELBO with respect to $q_{\beta,g}^{(b)}$ gives
\begin{align*}
q_{\beta,g}^{(b)}(\boldsymbol{\beta})
&\propto
\exp\Bigl\{
\mathbb E_{
q_{\sigma^2,g}^{(b)}
q_{\phi}^{(b)}
}
\left[
\log p\!\left(
\tilde{\mathbf y}_b^\ast
\mid
\tilde{\mathbf X}_b^\ast,
\boldsymbol{\beta},
\sigma^2,
\boldsymbol{\phi}
\right)
\right]
+
\log\tilde{\pi}_{\beta,g}^{(b)}(\boldsymbol{\beta})
\Bigr\}.
\end{align*}
The resulting factor is Gaussian. Its closed-form expression is given
in Proposition~\ref{prop:component_beta_update_independent}. Define
$\kappa_{b,g}
:=
\mathbb E_{q_{\sigma^2,g}^{(b)}}(\sigma^{-2}).$

\begin{proposition}[Componentwise update for the regression coefficients]
\label{prop:component_beta_update_independent}
For fixed $q_{\sigma^2,g}^{(b)}$ and $q_{\phi}^{(b)}$, the
coordinate-optimal regression factor is
\begin{align*}
q_{\beta,g}^{(b)}(\boldsymbol{\beta})
&=
\mathrm N_k\!\left(
\boldsymbol{\beta}
\mid
\mathbf m_{b,g},
\mathbf V_{b,g}^{-1}
\right),
\end{align*}
where
\begin{align*}
\mathbf V_{b,g}
&=
\tilde{\mathbf A}_{b,g}
+
\kappa_{b,g}
\mathbb E_{q_{\phi}^{(b)}}\!\left(
\tilde{\mathbf X}_b^{\ast\top}\tilde{\mathbf X}_b^\ast
\right),\\
\mathbf m_{b,g}
&=
\mathbf V_{b,g}^{-1}
\left\{
\tilde{\mathbf A}_{b,g}
\tilde{\boldsymbol{\beta}}_{b,g}
+
\kappa_{b,g}
\mathbb E_{q_{\phi}^{(b)}}\!\left(
\tilde{\mathbf X}_b^{\ast\top}\tilde{\mathbf y}_b^\ast
\right)
\right\}.
\end{align*}
\end{proposition}

Next, given Proposition~\ref{prop:component_beta_update_independent}, the
coordinate update for the innovation-variance factor is
\begin{align*}
q_{\sigma^2,g}^{(b)}(\sigma^2)
&\propto
\exp\Bigl\{
\mathbb E_{
q_{\beta,g}^{(b)}
q_{\phi}^{(b)}
}
\left[
\log p\!\left(
\tilde{\mathbf y}_b^\ast
\mid
\tilde{\mathbf X}_b^\ast,
\boldsymbol{\beta},
\sigma^2,
\boldsymbol{\phi}
\right)
\right]
+
\log\tilde{\pi}_{\sigma^2,g}^{(b)}(\sigma^2)
\Bigr\}.
\end{align*}
The resulting factor is inverse-gamma.

\begin{proposition}[Componentwise update for the innovation variance]
\label{prop:component_sigma_update_independent}
For fixed $q_{\beta,g}^{(b)}$ and $q_{\phi}^{(b)}$, the
coordinate-optimal innovation-variance factor is
\begin{align*}
q_{\sigma^2,g}^{(b)}(\sigma^2)
&=
\mathrm{IG}\!\left(
\sigma^2
\mid
\alpha_{b,g}/2,
\lambda_{b,g}/2
\right),
\end{align*}
where
\begin{equation}
\begin{aligned}
\alpha_{b,g}
&=
n_b^\ast
+
\tilde{\nu}_{b,g},
\\
\lambda_{b,g}
&=
\tilde{\delta}_{b,g}
+
\mathbb E_{q_{\phi}^{(b)}}\!\left(
\tilde{\mathbf y}_b^{\ast\top}\tilde{\mathbf y}_b^\ast
\right)
-
2\mathbf m_{b,g}^{\top}
\mathbb E_{q_{\phi}^{(b)}}\!\left(
\tilde{\mathbf X}_b^{\ast\top}\tilde{\mathbf y}_b^\ast
\right)
\\
&\quad+
\operatorname{tr}\!\left[
\mathbb E_{q_{\phi}^{(b)}}\!\left(
\tilde{\mathbf X}_b^{\ast\top}\tilde{\mathbf X}_b^\ast
\right)
\left(
\mathbf V_{b,g}^{-1}
+
\mathbf m_{b,g}\mathbf m_{b,g}^{\top}
\right)
\right].
\end{aligned}
\label{eq:sigma_update_independent}
\end{equation}
Consequently,
$\kappa_{b,g}=\alpha_{b,g}/\lambda_{b,g}$.
\end{proposition}

The expectations appearing in
Propositions~\ref{prop:component_beta_update_independent} and
\ref{prop:component_sigma_update_independent} involve
$\boldsymbol{\phi}$ through the prewhitened response and design matrix.
To make this dependence explicit, we collect the $p$ lags associated
with each usable observation into lag matrices.
The resulting lag-matrix formulation uses the stored boundary lags
$\{(y_{t_b-j},\mathbf x_{t_b-j})\}_{j=1}^p$ and yields closed-form
expressions for the required expectations under
$q_{\phi}^{(b)}$.
Define the lagged-response matrix
$\mathbf M_{y,b}\in\mathbb R^{n_b^\ast\times p}$ and the lagged
covariate matrices
$\mathbf F_{b,\ell}\in\mathbb R^{n_b^\ast\times p}$,
$\ell=1,\ldots,k$, entrywise by
\begin{align*}
[\mathbf M_{y,b}]_{ij}
&=
\begin{cases}
y_{p+i-j},
& b=1,\\
y_{t_b+i-j-1},
& b\geq2,\ j<i,\\
y_{t_b+i-j-1}
-
\mathbf x_{t_b+i-j-1}^{\top}
\bar{\boldsymbol{\beta}}_{b-1},
& b\geq2,\ j\geq i,
\end{cases}
\\
[\mathbf F_{b,\ell}]_{ij}
&=
\begin{cases}
x_{p+i-j,\ell},
& b=1,\\
x_{t_b+i-j-1,\ell},
& b\geq2,\ j<i,\\
0,
& b\geq2,\ j\geq i.
\end{cases}
\end{align*}
for $i=1,\ldots,n_b^\ast$ and $j=1,\ldots,p$.
Here, $j\geq i$ corresponds to lags crossing the
batch boundary; their regression contribution is absorbed into
$\mathbf M_{y,b}$ through $\bar{\boldsymbol{\beta}}_{b-1}$, leaving
the corresponding entries of $\mathbf F_{b,\ell}$ equal to zero.
Let
\begin{align*}
\mathbf F_b
&=
\left[
\mathbf F_{b,1}\ 
\mathbf F_{b,2}\ 
\cdots\ 
\mathbf F_{b,k}
\right]
\in\mathbb R^{n_b^\ast\times pk},
\quad
\boldsymbol{\Phi}
=
\mathbf I_k\otimes\boldsymbol{\phi}
\in\mathbb R^{pk\times k},
\end{align*}
where $\otimes$ denotes the Kronecker product.
The PM boundary-adjusted response and design matrix can then be written as
\begin{align}
\tilde{\mathbf y}_b^\ast
&=
\mathbf y_{b\mid T_b^\ast}
-
\mathbf M_{y,b}\boldsymbol{\phi},
\quad
\tilde{\mathbf X}_b^\ast
=
\mathbf X_{b\mid T_b^\ast}
-
\mathbf F_b\boldsymbol{\Phi}. \nonumber
\end{align}

\begin{proposition}[Expected PM boundary-adjusted sufficient statistics]
\label{prop:phi_expectations}
For mini-batch $b$, under
$q_{\phi}^{(b)}(\boldsymbol{\phi})
=
\mathrm N_p(\boldsymbol{\phi}\mid\mathbf c_b,\mathbf D_b^{-1})$,
the expected prewhitened moments satisfy
\begin{align*}
\mathbb E_{q_{\phi}^{(b)}}\!\left[
\tilde{\mathbf X}_b^{\ast\top}\tilde{\mathbf X}_b^\ast
\right]
& =
\left(
\mathbf X_{b\mid T_b^\ast}
-
\mathbf F_b\bar{\boldsymbol{\Phi}}_b
\right)^\top
\left(
\mathbf X_{b\mid T_b^\ast}
-
\mathbf F_b\bar{\boldsymbol{\Phi}}_b
\right)
+
\left[
\operatorname{tr}\!\left(
\mathbf F_{b,i}^\top
\mathbf F_{b,j}
\mathbf D_b^{-1}
\right)
\right]_{i,j=1}^k,\\
\mathbb E_{q_{\phi}^{(b)}}\!\left[
\tilde{\mathbf X}_b^{\ast\top}\tilde{\mathbf y}_b^\ast
\right]
& =
\left(
\mathbf X_{b\mid T_b^\ast}
-
\mathbf F_b\bar{\boldsymbol{\Phi}}_b
\right)^\top
\left(
\mathbf y_{b\mid T_b^\ast}
-
\mathbf M_{y,b}\mathbf c_b
\right)
+
\left[
\operatorname{tr}\!\left(
\mathbf F_{b,i}^\top
\mathbf M_{y,b}
\mathbf D_b^{-1}
\right)
\right]_{i=1}^k,\\
\mathbb E_{q_{\phi}^{(b)}}\!\left[
\tilde{\mathbf y}_b^{\ast\top}\tilde{\mathbf y}_b^\ast
\right]
& =
\left(
\mathbf y_{b\mid T_b^\ast}
-
\mathbf M_{y,b}\mathbf c_b
\right)^\top
\left(
\mathbf y_{b\mid T_b^\ast}
-
\mathbf M_{y,b}\mathbf c_b
\right)
+
\operatorname{tr}\!\left(
\mathbf M_{y,b}^\top
\mathbf M_{y,b}
\mathbf D_b^{-1}
\right).
\end{align*}
Here,
$\bar{\boldsymbol{\Phi}}_b
=
\mathbf I_k\otimes\mathbf c_b
\in\mathbb R^{pk\times k}$. The first bracketed trace term denotes
the $k\times k$ matrix indexed by $(i,j)$, and the second denotes the
$k$-vector indexed by $i$.
\end{proposition}

Using the closed-form expressions in Proposition~\ref{prop:phi_expectations},
the expectations in
Propositions~\ref{prop:component_beta_update_independent}
and~\ref{prop:component_sigma_update_independent}
are evaluated under the current shared AR factor and reused across candidate children.

The common AR factor is then updated according to
\begin{align*}
q_{\phi}^{(b)}(\boldsymbol{\phi})
&\propto
\exp\Biggl\{
\sum_{g\in\mathcal G_b}
\omega_{b,g}
\mathbb E_{
q_{\beta,g}^{(b)}
q_{\sigma^2,g}^{(b)}
}
\left[
\log p\!\left(
\tilde{\mathbf y}_b^\ast
\mid
\tilde{\mathbf X}_b^\ast,
\boldsymbol{\beta},
\sigma^2,
\boldsymbol{\phi}
\right)
\right]
+
\log q_{\phi}^{(b-1)}(\boldsymbol{\phi})
\Biggr\}.
\end{align*}
The resulting coordinate update is multivariate Gaussian. To state it,
define the component-specific lagged residual design matrix
$\mathbf E_{b,g}\in\mathbb R^{n_b^\ast\times p}$ entrywise by
\begin{align*}
[\mathbf E_{b,g}]_{ij}
&=
\begin{cases}
y_{p+i-j}
-
\mathbf x_{p+i-j}^\top\mathbf m_{b,g},
& b=1,\\
y_{t_b+i-j-1}
-
\mathbf x_{t_b+i-j-1}^\top\mathbf m_{b,g},
& b\geq2,\ j<i,\\
y_{t_b+i-j-1}
-
\mathbf x_{t_b+i-j-1}^\top
\bar{\boldsymbol{\beta}}_{b-1},
& b\geq2,\ j\geq i.
\end{cases}
\end{align*}
for $i=1,\ldots,n_b^\ast$ and $j=1,\ldots,p$.

\begin{proposition}[Mixture-weighted update for the AR coefficients]
\label{prop:vb_phi}
For fixed component factors and mixture weights, the coordinate-optimal
common AR factor is
\begin{align*}
q_{\phi}^{(b)}(\boldsymbol{\phi})
&=
\mathrm N_p\!\left(
\boldsymbol{\phi}
\mid
\mathbf c_b,
\mathbf D_b^{-1}
\right),
\end{align*}
where
\begin{align*}
\mathbf D_b
&=
\boldsymbol{\Phi}_0^{(b-1)}
+
\sum_{g\in\mathcal G_b}
\omega_{b,g}
\kappa_{b,g}
\left\{
\mathbf E_{b,g}^{\top}\mathbf E_{b,g}
+
\sum_{i=1}^k\sum_{j=1}^k
\left[
\mathbf V_{b,g}^{-1}
\right]_{ij}
\mathbf F_{b,i}^{\top}\mathbf F_{b,j}
\right\},
\\
\mathbf c_b
&=
\mathbf D_b^{-1}
\Biggl\{
\boldsymbol{\Phi}_0^{(b-1)}
\boldsymbol{\phi}_0^{(b-1)}
\\
&\quad+
\sum_{g\in\mathcal G_b}
\omega_{b,g}
\kappa_{b,g}
\left[
\mathbf E_{b,g}^{\top}
\left(
\mathbf y_{b\mid T_b^\ast}
-
\mathbf X_{b\mid T_b^\ast}\mathbf m_{b,g}
\right)
+
\sum_{i=1}^k\sum_{j=1}^k
\left[
\mathbf V_{b,g}^{-1}
\right]_{ij}
\mathbf F_{b,i}^{\top}
\mathbf X_{b\mid T_b^\ast}\mathbf e_j
\right]
\Biggr\}.
\end{align*}
Here, $\mathbf e_j\in\mathbb R^k$ denotes the $j$th standard basis
vector, and $[\mathbf V_{b,g}^{-1}]_{ij}$ denotes the $(i,j)$th entry of
the covariance matrix of $q_{\beta,g}^{(b)}$.
\end{proposition}

Finally, we update the variational factor for the latent child index
$G_b$, which is categorical with probabilities
$q_G^{(b)}(G_b=g)=\omega_{b,g}$ for $g\in\mathcal G_b$.
To express its coordinate update, define the component-local
variational score
\begin{align}
J_{b,g}
&=
\mathbb E_{
q_{\beta,g}^{(b)}
q_{\sigma^2,g}^{(b)}
q_{\phi}^{(b)}
}
\left[
\log
p\!\left(
\tilde{\mathbf y}_b^\ast
\mid
\tilde{\mathbf X}_b^\ast,
\boldsymbol{\beta},
\sigma^2,
\boldsymbol{\phi}
\right)
\right]
-
\operatorname{KL}\!\left(
q_{\beta,g}^{(b)}
\,\middle\|\,
\tilde{\pi}_{\beta,g}^{(b)}
\right)
-
\operatorname{KL}\!\left(
q_{\sigma^2,g}^{(b)}
\,\middle\|\,
\tilde{\pi}_{\sigma^2,g}^{(b)}
\right).
\label{eq:J_bg_independent}
\end{align}
Then the coordinate-optimal categorical factor satisfies
\begin{align*}
q_G^{(b)}(G_b=g)
&\propto
\tilde{\omega}_{b,g}\exp(J_{b,g}),
\quad
g\in\mathcal G_b.
\end{align*}

\begin{proposition}[Mixture-weight update]
\label{prop:mixture_weight_update}
For fixed component-specific factors
$\{q_{\beta,g}^{(b)},q_{\sigma^2,g}^{(b)}\}_{g\in\mathcal G_b}$
and fixed shared AR factor $q_{\phi}^{(b)}$, the coordinate-optimal
variational factor for $G_b$ is categorical with probabilities
\begin{align}
\omega_{b,g}
&=
\frac{
\tilde{\omega}_{b,g}\exp(J_{b,g})
}{
\displaystyle
\sum_{j\in\mathcal G_b}
\tilde{\omega}_{b,j}\exp(J_{b,j})
},
\quad
g\in\mathcal G_b.
\label{eq:weight_update_independent}
\end{align}
\end{proposition}

The marginal regression- and variance-reset probabilities and the
joint reset probability under the pre-pruning candidate approximation
$\tilde q^{(b)}$ are
\begin{align}
\hat p_{\beta,b}
&=
\sum_{g\in\mathcal G_b}
\omega_{b,g}z_{\beta,g},
\quad
\hat p_{\sigma^2,b}
=
\sum_{g\in\mathcal G_b}
\omega_{b,g}z_{\sigma^2,g}, \quad 
\hat p_{\beta,\sigma^2,b}
=
\sum_{g\in\mathcal G_b}
\omega_{b,g}
z_{\beta,g}z_{\sigma^2,g}.
\label{eq:reset_probs}
\end{align}
Without pruning, the number of candidate components would grow by a
factor of four at every batch. After the batchwise CAVI iterations
converge, we therefore retain a set
$\mathcal I_b\subseteq\mathcal G_b$ satisfying
$1\leq|\mathcal I_b|\leq K_{\max}$, typically by selecting the
components with the largest posterior weights. Write
$\mathcal I_b
=
\{g_{b,1},\ldots,g_{b,K_b}\},$ and $
K_b=|\mathcal I_b|.$
We renormalize and relabel the retained child $g_{b,r}$ as component
$r=1,\ldots,K_b$ by setting
\begin{align*}
\omega_{b,r}
&=
\frac{
\omega_{b,g_{b,r}}
}{
\displaystyle
\sum_{g'\in\mathcal I_b}
\omega_{b,g'}
},
\quad
\mathbf m_{b,r}
=
\mathbf m_{b,g_{b,r}},
\quad
\mathbf V_{b,r}
=
\mathbf V_{b,g_{b,r}},
\quad
\alpha_{b,r}
=
\alpha_{b,g_{b,r}},
\quad
\lambda_{b,r}
=
\lambda_{b,g_{b,r}}.
\end{align*}
The retained weights $\omega_{b,r}$ are used directly as the parent
weights in batch $b+1$. For $r=1,\ldots,K_b$, the remaining state is
propagated as
\begin{equation}
\begin{alignedat}{2}
\boldsymbol{\beta}_{0,r}^{(b)}
&\leftarrow \mathbf m_{b,r},
&\quad
\mathbf A_{0,r}^{(b)}
&\leftarrow \mathbf V_{b,r},
\\
\nu_{0,r}^{(b)}
&\leftarrow \alpha_{b,r},
&
\delta_{0,r}^{(b)}
&\leftarrow \lambda_{b,r},
\\
\boldsymbol{\phi}_0^{(b)}
&\leftarrow \mathbf c_b,
&
\boldsymbol{\Phi}_0^{(b)}
&\leftarrow \mathbf D_b.
\end{alignedat}
\label{eq:propagate_state}
\end{equation}
Proofs of all propositions and additional derivation details are
provided in Supplementary Sections~\ref{app:variational_pm} and~\ref{app:elbo_pm}, while the overall updating
scheme is summarized in Algorithm~\ref{alg:algo1}.

\begin{algorithm}[t!]
\caption{SR-SVB algorithm for the linear regression model with AR($p$) errors under the PM boundary treatment}
\label{alg:algo1}
\LinesNotNumbered

\textbf{Initialize:}
Set $K_0=1$, $\omega_{0,1}=1$, and
$(\boldsymbol{\beta}_{0,1}^{(0)},\mathbf A_{0,1}^{(0)},
\nu_{0,1}^{(0)},\delta_{0,1}^{(0)})
=
(\boldsymbol{\beta}_0^{(0)},\mathbf A_0^{(0)},
\nu_0^{(0)},\delta_0^{(0)})$.
Initialize $\boldsymbol{\phi}_0^{(0)}$ and
$\boldsymbol{\Phi}_0^{(0)}$, choose
$h_\beta,h_{\sigma^2},\eta_\beta,\eta_{\sigma^2},
K_{\max}$, and $\epsilon_{\mathrm{ELBO}}>0$, and set
$\mathcal Z_1=\{(0,0)\}$ and
$\mathcal Z_b=\{0,1\}^2$ for $b\geq2$.\\

\For{$b=1,2,\ldots$}{

Form $\mathcal G_b
=
\{(r,z_\beta,z_{\sigma^2}):
r=1,\ldots,K_{b-1},\;
(z_\beta,z_{\sigma^2})\in\mathcal Z_b\}.$

For $b\geq2$, compute
$\bar{\boldsymbol{\beta}}_{b-1}$ from
\eqref{eq:retained_beta_mean} and construct the boundary-adjusted lag
matrices $\mathbf M_{y,b}$ and $\{\mathbf F_{b,\ell}\}_{\ell=1}^k$;
for $b=1$, use the ordinary prewhitening lag matrices.\\

Compute the candidate prior states
$\{
\tilde{\omega}_{b,g},
\tilde{\boldsymbol{\beta}}_{b,g},
\tilde{\mathbf A}_{b,g},
\tilde{\nu}_{b,g},
\tilde{\delta}_{b,g}
\}_{g\in\mathcal G_b}.$

For every $g\in\mathcal G_b$, initialize
$(\omega_{b,g},\mathbf m_{b,g},\mathbf V_{b,g},
\alpha_{b,g},\lambda_{b,g})
\leftarrow
(\tilde{\omega}_{b,g},
\tilde{\boldsymbol{\beta}}_{b,g},
\tilde{\mathbf A}_{b,g},
\tilde{\nu}_{b,g},
\tilde{\delta}_{b,g}).$

Set
$\mathbf c_b\leftarrow\boldsymbol{\phi}_0^{(b-1)}$,
$\mathbf D_b\leftarrow\boldsymbol{\Phi}_0^{(b-1)}$,
$\mathcal L_b^{\mathrm{old}}\leftarrow-\infty$, and
$\mathit{change}\leftarrow+\infty$.\\

\While{$\mathit{change}>\epsilon_{\mathrm{ELBO}}$}{

For every $g\in\mathcal G_b$, update
\[
q_{\beta,g}^{(b)}(\boldsymbol{\beta})
=
\mathrm N_k(\mathbf m_{b,g},\mathbf V_{b,g}^{-1}),
\quad
q_{\sigma^2,g}^{(b)}(\sigma^2)
=
\mathrm{IG}
\left(
\frac{\alpha_{b,g}}{2},
\frac{\lambda_{b,g}}{2}
\right).
\]

Update
\[
q_\phi^{(b)}(\boldsymbol{\phi})
=
\mathrm N_p(\mathbf c_b,\mathbf D_b^{-1})
\]
and $\{\omega_{b,g}\}_{g\in\mathcal G_b}$.\\

Compute $\mathcal L_b^{\mathrm{new}}$ and set $\mathit{change}
\leftarrow
|\mathcal L_b^{\mathrm{new}}-\mathcal L_b^{\mathrm{old}}|$
and
$\mathcal L_b^{\mathrm{old}}
\leftarrow\mathcal L_b^{\mathrm{new}}$.\\
}

Compute $\hat p_{\beta,b}$, $\hat p_{\sigma^2,b}$, and
$\hat p_{\beta,\sigma^2,b}$ from the pre-pruning candidate
approximation $\tilde q^{(b)}$.\\

Select
$\mathcal I_b=\{g_{b,1},\ldots,g_{b,K_b}\}
\subseteq\mathcal G_b$
with $1\leq K_b\leq K_{\max}$, and renormalize and relabel the
retained components as $r=1,\ldots,K_b$.

Use $\{\omega_{b,r}\}_{r=1}^{K_b}$ as the parent weights for
batch $b+1$ and propagate the remaining parameter states according to
\eqref{eq:propagate_state}.\\
Retain the last $p$ raw response--covariate pairs
for the next batch.
}
\end{algorithm}

\section{Experiments}
\label{sec:sec4}

In this section, we evaluate SR-SVB under both stable and changing data-generating regimes.
Section~\ref{subsec:sec4.1} examines estimation, UQ, and reported false alarms under a stable data-generating mechanism, whereas Section~\ref{subsec:sec4.2} evaluates type-specific changepoint detection and posterior coverage under structural changes.

\subsection{Stable-Regime Performance}
\label{subsec:sec4.1}

We evaluate finite-sample performance under a stable data-generating
mechanism.  The comparison includes the full-batch references MLE, Gibbs sampling, and full-batch variational Bayes (FBVB), and two streaming variational procedures: SVB, which uses the same componentwise independent variational family
as the proposed method but without reset states or mixture expansion, and SR-SVB, the proposed method. Under the stable regime, we evaluate performance through blockwise
parameter errors and empirical coverage of nominal 95\% regions for
$(\boldsymbol{\beta},\boldsymbol{\phi},\sigma^2)$.
For the proposed method, we additionally report false-alarm rates. 

\paragraph{Parameter setting.}
We set $p=k=50$ and consider
$n\in\{2000,10000,50000\}$.  Data are generated from
\eqref{eq:linearreg} with
$\mathbf x_t=(1,x_{t,2},\ldots,x_{t,k})^\top$, where
$x_{t,2},\ldots,x_{t,k}\overset{\mathrm{i.i.d.}}{\sim}\mathrm N(0,1)$.
We fix $\sigma^2=1$ and draw
$\boldsymbol{\beta}
\sim
\mathrm N_k(\mathbf 0,\mathbf I_k),\ 
\boldsymbol{\phi}
\sim
\mathrm{TN}_{p,\mathcal S_{\boldsymbol{\phi}}}
(
\mathbf 0,
0.5^2\operatorname{diag}
\left(1^{-1.7},2^{-1.7},\ldots,p^{-1.7}\right)),$
where $\mathrm{TN}_{p,\mathcal S_{\boldsymbol{\phi}}}$ denotes a
$p$-dimensional multivariate normal distribution truncated to
$\mathcal S_{\boldsymbol{\phi}}$, the set of AR($p$) coefficient vectors
that yield a stationary process.  One draw of
$(\boldsymbol{\beta},\boldsymbol{\phi})$ is held fixed across the three
sample-size scenarios, and we use $R=200$ Monte Carlo replications for each
sample size.

\paragraph{Inference settings.}
FBVB and SVB use the same componentwise independent variational
specification as the proposed method, with Gaussian factors for
$\boldsymbol{\beta}$ and $\boldsymbol{\phi}$ and an inverse-gamma factor
for $\sigma^2$. 
FBVB fits the full dataset under the static-parameter conditional
likelihood in Section~\ref{sec:sec2}, whereas SVB processes observations
in chronological mini-batches of size $n_b=500$ and propagates the
variational posterior from one batch to the next.
SR-SVB uses the same mini-batch size and the finite-mixture family
developed in Section~\ref{sec:sec3}.
For SR-SVB, we set
$h_\beta=h_{\sigma^2}=1-\sqrt{0.99}$,
$\eta_\beta=0.001$,
$\eta_{\sigma^2}=0.001$, and
$K_{\max}=10$.
These values are fixed across all replications and sample sizes.
Posterior summaries are computed from the retained post-pruning
approximation, whereas reset probabilities are evaluated before pruning.
All variational methods use the baseline priors
$\boldsymbol{\beta}\sim\mathrm N_k(\mathbf 0,10^6\mathbf I_k)$,
$\sigma^2\sim\mathrm{IG}(0.01,1)$, and
$\boldsymbol{\phi}\sim\mathrm N_p(\mathbf 0,10^6\mathbf I_p)$.
As additional full-batch references, MLE maximizes the Gaussian conditional
log-likelihood \citep[see, e.g.,][]{Robinson2006}, and the Gibbs
sampler follows \citet{chib1993bayesAR} to target the exact Bayesian posterior.

For runtime comparisons, convergence is assessed using changes in the
parameter estimates. For FBVB and SVB, we use the variational posterior
means, for MLE the current parameter estimates, and for Gibbs sampling the posterior means. We monitor the maximum relative change in
$\boldsymbol{\beta}$, $\boldsymbol{\phi}$, and $\sigma^2$ between successive
iterations, and all methods use the same convergence tolerance of $10^{-3}$. For Gibbs sampling, we enforce a minimum of 3,000 iterations regardless of whether the convergence criterion is met earlier. 
For SR-SVB, the ELBO stopping rule in Algorithm~\ref{alg:algo1} is replaced
by the same parameter-based criterion applied to the mixture posterior means,
with the relative change in the mixture weights additionally monitored.

\paragraph{Point-estimation accuracy.}

For each method, we report
$\lVert\hat{\boldsymbol{\beta}}-\boldsymbol{\beta}\rVert_2$,
$\lVert\hat{\boldsymbol{\phi}}-\boldsymbol{\phi}\rVert_2$, and
$|\hat{\sigma}^2-\sigma^2|$.
MLE uses the conditional-likelihood maximizer, and Gibbs uses posterior means
after burn-in.
For FBVB and SVB, the posterior means are
$\hat{\boldsymbol{\beta}}=\boldsymbol{m}$,
$\hat{\sigma}^2=\lambda/(\alpha-2)$, and
$\hat{\boldsymbol{\phi}}=\boldsymbol{c}$.
For SR-SVB, the retained post-pruning mixture has $K$ components
indexed by $r=1,\ldots,K$, and the corresponding posterior means are
$\hat{\boldsymbol{\beta}}=\sum_{r=1}^{K}\omega_r\boldsymbol{m}_r$,
$\hat{\sigma}^2=\sum_{r=1}^{K}\omega_r\lambda_r/(\alpha_r-2)$, and
$\hat{\boldsymbol{\phi}}=\boldsymbol{c}$.
For the streaming methods, all quantities refer to the final batch,
whose index is suppressed.

Table~\ref{tab:stable_regime_error} 
shows that the estimation errors for
$\boldsymbol{\beta}$ and $\boldsymbol{\phi}$ decrease rapidly with $n$,
and the full-batch methods are nearly indistinguishable at the two larger
sample sizes. Most importantly for the stable-regime comparison, SVB and
SR-SVB give very similar point-estimation accuracy. At $n=2000$, their
mean errors are respectively $0.164$ and $0.165$ for
$\boldsymbol{\beta}$, $0.187$ and $0.184$ for
$\boldsymbol{\phi}$, and $0.030$ and $0.029$ for $\sigma^2$.
Their differences remain negligible at $n=10000$ and $n=50000$, where
both procedures closely track FBVB. Thus, under this stable-regime design, SR-SVB achieves point-estimation
accuracy comparable to the nonadaptive SVB benchmark.
\paragraph{False-alarm assessment.}
For alarm reporting and false-alarm evaluation, we use the pre-pruning
reset probabilities defined in \eqref{eq:reset_probs}. Let $B_n=n/500$ denote the number
of batches for sample size $n$, and let $\hat p_{j,b}^{(r)}$ denote
the reset probability for channel $j\in\{\beta,\sigma^2\}$ at batch $b$
in Monte Carlo replication $r=1,\ldots,R$.
An alarm is declared in channel $j$ when $\hat p_{j,b}^{(r)}>c$,
with the fixed reporting threshold $c=0.50$.
The batch-level false-alarm frequency and stream-level false-alarm
probability are respectively computed as
\begin{equation}
\label{eq:false_alarm_metrics}
\begin{aligned}
\text{Batch level:}\quad
&\frac{1}{R(B_n-1)}
\sum_{r=1}^{R}
\sum_{b=2}^{B_n}
\mathbf 1
\left\{
\hat p_{j,b}^{(r)}>c
\right\},\\
\text{Stream level:}\quad
&\frac{1}{R}
\sum_{r=1}^{R}
\mathbf 1
\left\{
\max_{2\leq b\leq B_n}
\hat p_{j,b}^{(r)}>c
\right\}.
\end{aligned}
\end{equation}
The first batch is excluded because it has no preceding state to retain or
reset. 

As shown in Table~\ref{tab:stable_regime_false_alarm}, no reported false alarms were observed at $n=2000$ or $n=10000$ using the fixed threshold of 0.50. At $n=50000$, no regression-reset alarm
was observed, and only two variance-reset alarms occurred among the
$19{,}800$ eligible transitions. This gives a batch-level false-alarm
frequency of $2/19{,}800=1.01\times10^{-4}$ for the variance block.
The two alarms occurred in two different streams, giving a stream-level
false-alarm probability of $2/200=0.010$. Hence, at this reporting threshold, false-alarm exceedances were rare under the stable-regime design.

\paragraph{Coverage.}
To assess UQ, we report
95\% blockwise coverage across simulated datasets.
For MLE, the $\boldsymbol{\beta}$ and $\boldsymbol{\phi}$ regions are
asymptotic-normal ellipsoids, while $\sigma^2$ is treated on the
log-variance scale and mapped back to the original scale. For Gibbs sampling, the 95\% regions for
$\boldsymbol{\beta}$ and $\boldsymbol{\phi}$ are constructed directly
from the posterior draws using empirical Mahalanobis-distance cutoffs,
while the interval for $\sigma^2$ is obtained from the corresponding
posterior quantiles. 
For FBVB and SVB, the Gaussian variational factors for
$\boldsymbol{\beta}$ and $\boldsymbol{\phi}$ yield ellipsoidal 95\%
regions with $\chi_k^2$ and $\chi_p^2$ cutoffs, respectively, and
the interval for $\sigma^2$ is formed from the 0.025 and 0.975 quantiles
of its inverse-gamma marginal.
For SR-SVB, the 95\% region for $\boldsymbol{\beta}$ is an ellipsoid based
on the mixture posterior mean and covariance, with the cutoff given by the empirical 0.95 quantile of the squared Mahalanobis distances evaluated on posterior draws. The $\sigma^2$ interval is formed from the 0.025 and 0.975 quantiles of the
retained inverse-gamma mixture. The AR block has the shared Gaussian factor, so its 95\% region is
computed analytically using a $\chi_p^2$ cutoff.

At $n=2000$, the empirical coverage of the nominal 95\% region for the
50-dimensional $\boldsymbol{\beta}$ block is $0.565$ for both SVB and
SR-SVB, indicating substantial small-sample undercoverage despite their
similar point-estimation accuracy. The corresponding
$\boldsymbol{\phi}$ coverage is $0.740$ under SVB and $0.830$ under
SR-SVB, although both remain below the nominal level, while the
$\sigma^2$ coverages are $0.920$ and $0.925$, respectively. Coverage
improves as $n$ increases. At $n=10000$, SR-SVB and SVB have similar
$\boldsymbol{\beta}$ coverage, while SR-SVB is slightly closer to nominal
for $\boldsymbol{\phi}$ and $\sigma^2$. By $n=50000$, both streaming
variational procedures are substantially closer to nominal coverage
across all three blocks, with only small differences between them. Gibbs
and FBVB likewise approach nominal coverage as $n$ grows. Overall, the
stable-regime results show that SR-SVB attains estimation and UQ
performance similar to SVB at the larger sample sizes with low reported
false-alarm rates under the null.

\begin{table}[!t]
\centering
\caption{
Performance under the stable regime for $p=k=50$ over $R=200$ datasets.
Mean runtimes refer to full-data fits for batch methods and the final
mini-batch for streaming methods.
}
\label{tab:stable_regime}

\begin{subtable}[t]{\textwidth}
\centering
\caption{Blockwise estimation errors, empirical coverage, and runtime.}
\label{tab:stable_regime_error}
\resizebox{\linewidth}{!}{%
\begin{tabular}{ccc ccc ccc r}
\toprule
\multirow{2}{*}{$n$}
& \multirow{2}{*}{Method}
& \multirow{2}{*}{\shortstack{Batch\\size}}
& \multicolumn{1}{c}{$L_2$ error of $\boldsymbol{\beta}$}
& \multicolumn{1}{c}{$L_2$ error of $\boldsymbol{\phi}$}
& \multicolumn{1}{c}{$L_2$ error of $\sigma^2$}
& \multicolumn{3}{c}{95\% coverage}
& \multicolumn{1}{c}{\multirow{2}{*}{\shortstack{Runtime\\(seconds)}}} \\
& &
& \multicolumn{1}{c}{Mean (SD)}
& \multicolumn{1}{c}{Mean (SD)}
& \multicolumn{1}{c}{Mean (SD)}
& $\boldsymbol{\beta}$
& $\boldsymbol{\phi}$
& $\sigma^2$
& \\
\midrule

\multirow{5}{*}{$2000$}
& MLE
& $n$
& 0.156 (0.016)
& 0.182 (0.021)
& 0.051 (0.028)
& 0.785 & 0.840 & 0.640
& 0.102 \\

& Gibbs
& $n$
& 0.165 (0.146)
& 0.175 (0.019)
& 0.033 (0.023)
& 0.895 & 0.945 & 0.865
& 244.588 \\

& FBVB
& $n$
& 0.155 (0.016)
& 0.175 (0.019)
& 0.033 (0.023)
& 0.835 & 0.905 & 0.865
& 70.663 \\

& SVB
& $500$
& 0.164 (0.017)
& 0.187 (0.023)
& 0.030 (0.021)
& 0.565 & 0.740 & 0.920
& 0.148 \\

& SR-SVB
& $500$
& 0.165 (0.018)
& 0.184 (0.022)
& 0.029 (0.023)
& 0.565 & 0.830 & 0.925
& 0.328 \\

\midrule

\multirow{5}{*}{$10000$}
& MLE
& $n$
& 0.067 (0.007)
& 0.077 (0.009)
& 0.013 (0.010)
& 0.910 & 0.955 & 0.905
& 0.248 \\

& Gibbs
& $n$
& 0.067 (0.007)
& 0.076 (0.009)
& 0.012 (0.009)
& 0.900 & 0.970 & 0.935
& 420.968 \\

& FBVB
& $n$
& 0.067 (0.007)
& 0.076 (0.009)
& 0.012 (0.009)
& 0.920 & 0.975 & 0.935
& 219.969 \\

& SVB
& $500$
& 0.069 (0.008)
& 0.078 (0.010)
& 0.011 (0.009)
& 0.805 & 0.910 & 0.940
& 0.138 \\

& SR-SVB
& $500$
& 0.069 (0.008)
& 0.078 (0.010)
& 0.012 (0.010)
& 0.810 & 0.945 & 0.950
& 0.265 \\

\midrule

\multirow{5}{*}{$50000$}
& MLE
& $n$
& 0.030 (0.003)
& 0.034 (0.004)
& 0.005 (0.004)
& 0.940 & 0.960 & 0.940
& 2.246 \\

& Gibbs
& $n$
& 0.030 (0.003)
& 0.034 (0.004)
& 0.005 (0.004)
& 0.920 & 0.955 & 0.945
& 1676.842 \\

& FBVB
& $n$
& 0.030 (0.003)
& 0.034 (0.004)
& 0.005 (0.004)
& 0.935 & 0.955 & 0.945
& 730.291 \\

& SVB
& $500$
& 0.030 (0.003)
& 0.034 (0.004)
& 0.005 (0.004)
& 0.920 & 0.940 & 0.940
& 0.130 \\

& SR-SVB
& $500$
& 0.030 (0.003)
& 0.034 (0.004)
& 0.006 (0.005)
& 0.915 & 0.945 & 0.965
& 0.231 \\

\bottomrule
\end{tabular}%
}
\end{subtable}

\vspace{1em}

\begin{subtable}[t]{\textwidth}
\centering
\caption{False-alarm frequency and probability of SR-SVB at threshold $0.50$.}
\label{tab:stable_regime_false_alarm}

\resizebox{0.6\linewidth}{!}{%
\begin{tabular}{
c
*{4}{>{\centering\arraybackslash}p{1.6cm}}
}
\toprule
\multirow{2}{*}{$n$}
& \multicolumn{2}{c}{\shortstack{False-alarm frequency\\(batch level)}}
& \multicolumn{2}{c}{\shortstack{False-alarm probability\\(stream level)}} \\
\cmidrule(lr){2-3}
\cmidrule(lr){4-5}
& $\boldsymbol{\beta}$
& $\sigma^2$
& $\boldsymbol{\beta}$
& $\sigma^2$ \\
\midrule

$2000$
& 0.000
& 0.000
& 0.000
& 0.000 \\

$10000$
& 0.000
& 0.000
& 0.000
& 0.000 \\

$50000$
& 0.000
& $<0.001$
& 0.000
& 0.010 \\

\bottomrule
\end{tabular}%
}
\end{subtable}
\end{table}

\paragraph{Runtime.}
As shown in Table~\ref{tab:stable_regime_error},
the streaming procedures maintain low computational cost as the cumulative
sample size increases because each update processes only a fixed-size
mini-batch. The mean final-batch runtime of SVB ranges from
$0.130$ to $0.148$ seconds, while that of SR-SVB ranges from
$0.231$ to $0.328$ seconds. The additional cost of SR-SVB reflects the
mixture updates required for separate reset states, but the final-batch update times of SVB and SR-SVB are substantially smaller than the full-data fitting times of the batch Bayesian methods.

\subsection{Changepoint Detection and Adaptive Uncertainty Quantification}
\label{subsec:sec4.2}

We next evaluate SR-SVB under repeated structural changes, focusing on
type-specific changepoint detection, changepoint location error, and
detection delay, together with empirical coverage of the batchwise
posterior uncertainty summaries.
We compare SR-SVB with two frequentist sequential monitoring procedures:
FOCuS, a likelihood-ratio-based procedure implemented in the R package
\texttt{focus} \citep{romano2023focus,romano2026focus}, and fxmonitor,
a score-based regression monitoring procedure provided by the R package
\texttt{fxregime} \citep{zeileis2010testing}.

Our direct detection comparison is restricted to methods that yield
regression- and innovation-variance-sensitive monitoring statistics to
which a common type-specific matching rule can be applied. We do not include
\citet{li2021irregular}, \citet{niang2026online}, or
\citet{knoblauch2018spatiotemporal} because their original detection
outputs do not support the same regression-versus-variance
classification: the linear-regression formulation of
\citet{li2021irregular} treats the observation-noise variance as fixed,
\citet{niang2026online} targets coefficient changes, and
\citet{knoblauch2018spatiotemporal} infers segment-level changepoints
without assigning them to regression- or variance-specific change
types. Including these methods in the present type-specific comparison
would therefore require an additional parameter-block attribution
procedure not defined by the original methods.

Detection performance is compared under a common stratum-specific
calibration rule that limits the stream-level probability of at least
one alarm in either channel to at most $5\%$ on the calibration streams,
while the coverage trajectories assess the batchwise posterior UQ
provided by SR-SVB. Finally, we examine the role of the separate-reset
construction through a joint-reset ablation and a targeted
simultaneous-change sensitivity experiment, with full results reported
in Supplementary Section~\ref{sec:supp_reset_ablation}.

\paragraph{Data-generating setting.}

\begin{table}[t]
\centering
\caption{Mapping of the sixteen changepoint scenarios to the simulation settings.}
\label{tab:cp_scenarios}
\small
\setlength{\tabcolsep}{6pt}
\renewcommand{\arraystretch}{1.15}
\begin{tabular}{lcccccccc}
\toprule
& \multicolumn{4}{c}{Large AR}
& \multicolumn{4}{c}{Small AR} \\
\cmidrule(lr){2-5}
\cmidrule(lr){6-9}
Signal
& \multicolumn{2}{c}{Strong}
& \multicolumn{2}{c}{Weak}
& \multicolumn{2}{c}{Strong}
& \multicolumn{2}{c}{Weak} \\
$n$
& 25K & 50K
& 25K & 50K
& 25K & 50K
& 25K & 50K \\
\midrule
$N_{\mathrm{cp}}=10$
& S01 & S02 & S03 & S04
& S05 & S06 & S07 & S08 \\
$N_{\mathrm{cp}}=4$
& S09 & S10 & S11 & S12
& S13 & S14 & S15 & S16 \\
\bottomrule
\end{tabular}
\end{table}

We consider sixteen scenarios formed by crossing two fixed AR
specifications, strong- and weak-signal settings, sample sizes
$n\in\{25{,}000,50{,}000\}$, and numbers of changepoints
$N_{\mathrm{cp}}\in\{4,10\}$.
Data follow \eqref{eq:linearreg} with $p=k=30$,
piecewise-constant $\boldsymbol{\beta}$ and $\sigma^2$,
and time-invariant $\boldsymbol{\phi}$.
The covariates satisfy $x_{t,1}=1$ and
$x_{t,j}\overset{\mathrm{iid}}{\sim}\mathrm{N}(0,1)$
for $j=2,\ldots,30$.
We use an initial mini-batch of 500 observations followed by batches
of 250, giving 99 and 199 batches for the two sample sizes,
respectively.

We use two fixed stationary AR(30) specifications with
$\sum_{j=1}^{30}|\phi_j|=0.70$ and $0.945$.
The AR coefficients are held fixed across scenarios and replications,
and the error process continues through changepoints without restarting
its lag history.
Regression changes affect the first fifteen coefficients, including
the intercept, while the remaining coefficients are unchanged.
The jump directions are generated once and held fixed across scenarios
and replications.
Under the strong-signal setting, each regression jump satisfies
$\|\boldsymbol{\beta}_\ell-\boldsymbol{\beta}_{\ell-1}\|_2
=1.05\sigma_{\ell-1}$.
For $N_{\mathrm{cp}}=10$, the innovation-variance states are
$(1,2.05,2.05^2,2.05,1,2.05^{-1})$, whereas
$N_{\mathrm{cp}}=4$ uses the corresponding initial portion of the
same regression and variance paths.
The weak-signal setting preserves the same directions and ordering but
scales the regression path toward its initial state by $0.7$ and the
log-variance path by $0.7$, giving
$\boldsymbol{\beta}_0+0.7(\boldsymbol{\beta}_\ell-\boldsymbol{\beta}_0)$
and $(\sigma_\ell^2)^{0.7}$, respectively.
Regression and variance changes alternate, beginning with a regression
change, with no simultaneous changes.
For each $N_{\mathrm{cp}}$, a fixed irregular changepoint schedule is
defined on the mini-batch boundaries at $n=25{,}000$ and shared across
AR and signal-strength settings.
For $n=50{,}000$, the changepoint locations are doubled to preserve
their relative positions.
A changepoint $\tau$ marks the end of the preceding regime, with the
new parameters taking effect at $\tau+1$.

Each scenario is replicated $R=200$ times with fixed parameter paths
and changepoint locations.
For a given replicate, the same covariate and innovation draws are
used across scenarios, while different replicates use independent
draws. Table~\ref{tab:cp_scenarios} summarizes the sixteen changepoint
scenarios obtained by crossing the two AR specifications, sample sizes,
signal strengths, and numbers of changepoints.
Figure~\ref{fig:parameter_paths} shows the true paths of one active
regression coefficient, $\beta_{11}$, and the innovation variance,
$\sigma^2$, under the $N_{\mathrm{cp}}=4$ and $N_{\mathrm{cp}}=10$
settings. Solid and dashed lines denote the strong- and weak-signal
settings, respectively.

\begin{figure}[t]
\centering
\includegraphics[width=0.75\textwidth]{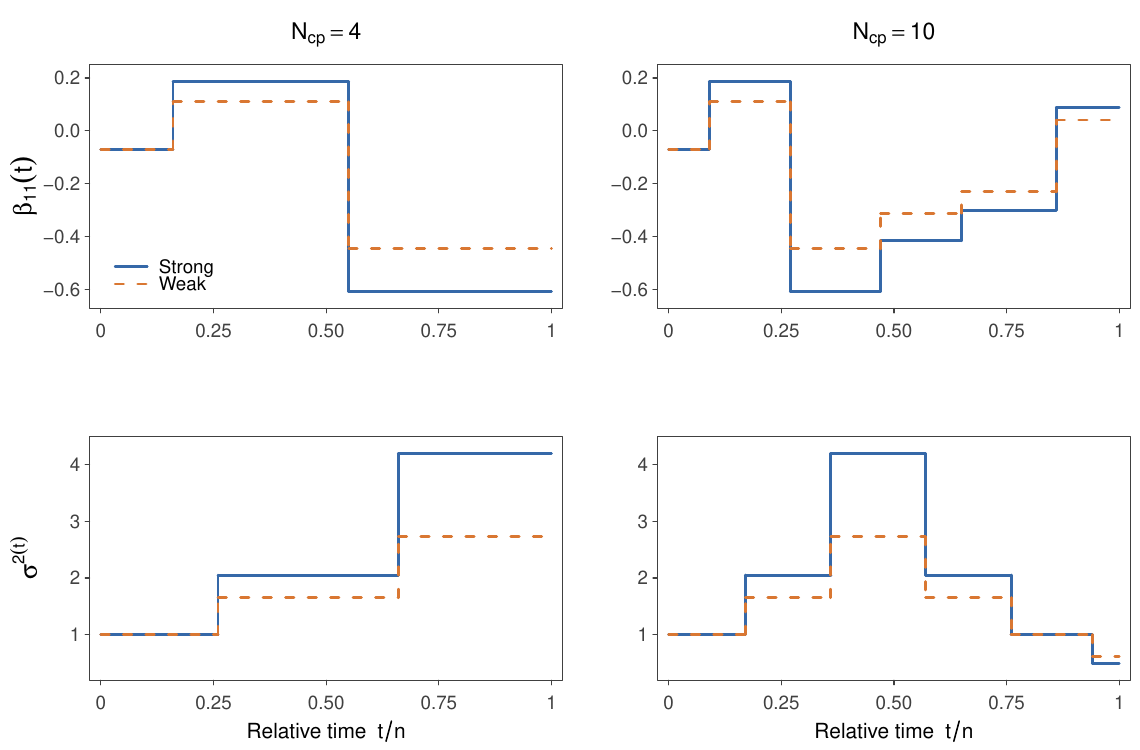}
\caption{
True paths of $\beta_{11}$ and $\sigma^2$ for the strong- and
weak-signal settings.
}
\label{fig:parameter_paths}
\end{figure}

\paragraph{Inference and comparison methods.}
SR-SVB uses the same reset and pruning settings as in the stable-regime
experiment, with $h_\beta=h_{\sigma^2}=1-\sqrt{0.99}$,
$\eta_\beta=\eta_{\sigma^2}=0.001$, $K_{\max}=10$, and a
within-batch relative-change tolerance of $10^{-3}$ for posterior
point estimates and mixture weights. SR-SVB models the AR errors and updates $q_{\phi}^{(b)}$ sequentially with each incoming batch.

FOCuS and fxmonitor are applied directly to the raw stream, using the first
500 observations as the initial reference and evaluating alarms at the
endpoints of subsequent 250-observation batches.
FOCuS uses an unknown-baseline multivariate Gaussian detector for the
regression channel \citep{romano2026focus} and a Gamma detector for the
variance channel \citep{ward2024exponential}.
Their inputs are reference-standardized batchwise ordinary least-squares coefficient
summaries and residual sums of squares, with batch sizes and residual
degrees of freedom as the respective weights. In fxmonitor, a
block-diagonal score covariance is estimated from the reference sample.
The regression and variance statistics are the maximum absolute
regression-score component and the absolute variance score, respectively,
each normalized by the package boundary.
{To evaluate repeated changes, we apply a common
channel-labeling and reference-refresh protocol to both monitoring methods.}
Under this protocol, an alarm is labeled as regression or variance when
only that channel exceeds its calibrated threshold, and as joint when both
do so. After an alarm, monitoring restarts once four new full batches have
been collected to form the next reference window.

\paragraph{False-alarm calibration and evaluation metrics.}
For each sample-size and AR-specification combination, we generate
500 no-change streams, split into 400 calibration and 100 independent
validation streams, with all methods evaluated on the same null data
and batch grid. For each method, channel-specific thresholds are
selected jointly from the calibration-stream maxima using a balanced-union
rule that allows at most 20 of the 400 streams to trigger an alarm in
either channel. The resulting threshold pair is shared across signal
strengths and changepoint counts and fixed before validation and
change-stream evaluation.
{Supplementary Table~\ref{tab:cp_calibration_validation} reports the corresponding
calibration and independent validation results. 
}

Alarm episodes are matched chronologically to true changes of the
same type using one-to-one matching.
The matching windows, episode construction, and reported-boundary
conventions are specified in Supplementary Section~\ref{app:changepoint_details}.
Matched events are counted as TP, unmatched alarm episodes as FP,
and unmatched true changes as FN.
For each method, scenario, and scope, we pool counts across
replications and report
$\mathrm{Recall}=\mathrm{TP}/(\mathrm{TP}+\mathrm{FN})$,
$\mathrm{Precision}=\mathrm{TP}/(\mathrm{TP}+\mathrm{FP})$, and
$F_1=2\mathrm{TP}/(2\mathrm{TP}+\mathrm{FP}+\mathrm{FN})$.
For each matched event $j$, let $b_j$ denote the last batch before
the true change, $\hat b_j$ the reported boundary index, and $a_j$
the alarm batch.
Over the matched-event set $\mathcal M$, we define
\[
L=|\mathcal M|^{-1}\sum_{j\in\mathcal M}|\hat b_j-b_j|,
\quad
D=|\mathcal M|^{-1}\sum_{j\in\mathcal M}\{a_j-(b_j+1)\},
\]
both measured in batches. Here, $L$ measures the absolute changepoint
location error on the batch grid, whereas $D$ measures the additional
detection delay from the first affected batch to the alarm batch.
Thus, $D=0$ denotes detection at the end of the first affected batch.
Both measures are conditional on matched events and should be
interpreted alongside recall.
Overall results include all alarm types, whereas block-specific
results include only alarms and true changes of the corresponding
type.
Because no simultaneous changes are simulated, each joint alarm
contributes one FP to the overall counts and is excluded from
block-specific results.

\begin{figure*}[!t]
\centering

\begin{subfigure}[t]{\textwidth}
    \centering
    \includegraphics[width=\textwidth]
    {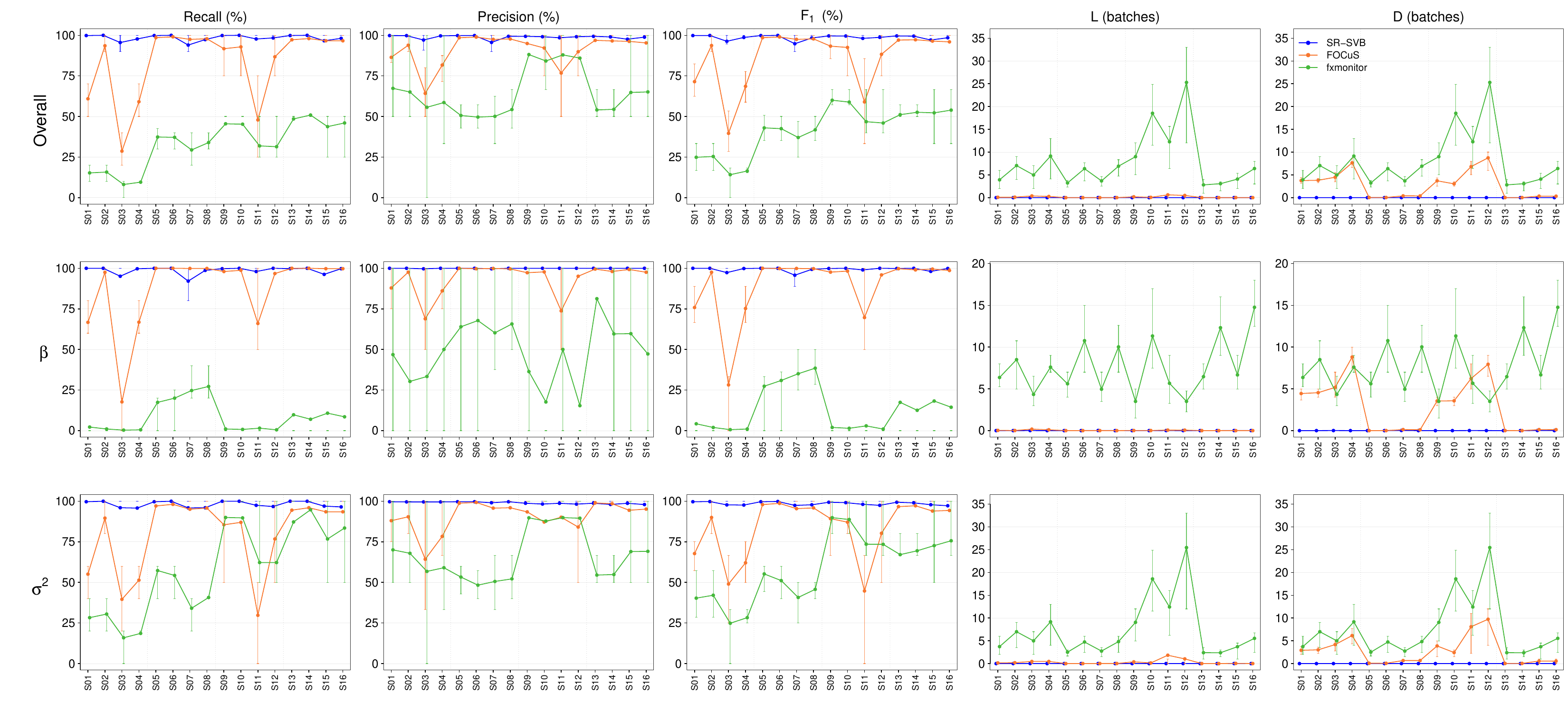}
    \caption{Detection performance.}
    \label{fig:cp_detection_panel}
\end{subfigure}

\vspace{0.3em}

\begin{subfigure}[t]{0.4\textwidth}
    \centering
    \includegraphics[width=\textwidth]
    {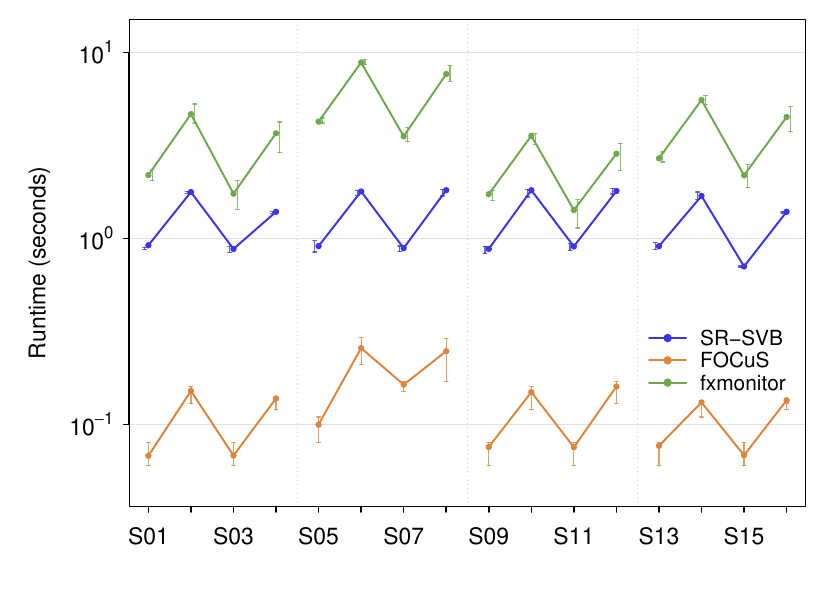}
    \caption{Mean runtime per dataset on a logarithmic scale.}
    \label{fig:runtime_panel}
\end{subfigure}

\caption{
Performance of SR-SVB, FOCuS, and fxmonitor across sixteen scenarios.
Recall, precision, and $F_1$ use pooled event counts; $L$ and $D$ are
averages over matched events.
Vertical bars show replicate-level interquartile ranges.
}
\label{fig:cp_detection}

\end{figure*}

\paragraph{Changepoint detection results.}
Figure~\ref{fig:cp_detection_panel} summarizes exact-type changepoint detection performance across the sixteen scenarios, with $R=200$ replications per scenario.
The plotted recall, precision, and $F_1$ values are computed from event counts pooled within each scenario and expressed as percentages, whereas $L$ and $D$ are averages over matched events on the batch grid; vertical bars show the corresponding replicate-level interquartile ranges ($Q_1$--$Q_3$).
SR-SVB maintains high recall, precision, and $F_1$ for both regression and innovation-variance changes across all sixteen scenarios, while keeping conditional location errors and detection delays small.
FOCuS performs competitively in several settings, with location errors comparable to those of SR-SVB, but shows greater scenario dependence: its recall and $F_1$ drop markedly in some scenarios, particularly S03 and S11, and its detection delays increase in several scenarios.
fxmonitor generally shows lower recall and $F_1$, most notably for regression changes, together with larger location errors and detection delays and wider replicate-level variation in several scenarios.
Overall, SR-SVB provides the most consistent detection performance, with generally the shortest detection delays among the three procedures and small location errors; these properties persist across the considered AR specifications, signal strengths, sample sizes, and numbers of changepoints.

\paragraph{Runtime results.}
Figure~\ref{fig:runtime_panel} reports the mean runtime per dataset
across the sixteen scenarios on a logarithmic scale.
Vertical bars indicate the first and third quartiles
of the replicate-level runtimes across the
$R=200$ replications.
FOCuS is consistently the fastest method, with runtimes on the order
of $10^{-1}$ seconds.
SR-SVB has the second-lowest computational cost, while fxmonitor is
the slowest of the three methods across the scenarios.
Although SR-SVB requires more computation than FOCuS, its runtime
remains below that of fxmonitor while additionally providing
finite-mixture variational inference and batchwise posterior
UQ for the model parameters.

\paragraph{Adaptive uncertainty quantification.}
\begin{figure*}[!t]
\centering
\includegraphics[width=1.0\textwidth]{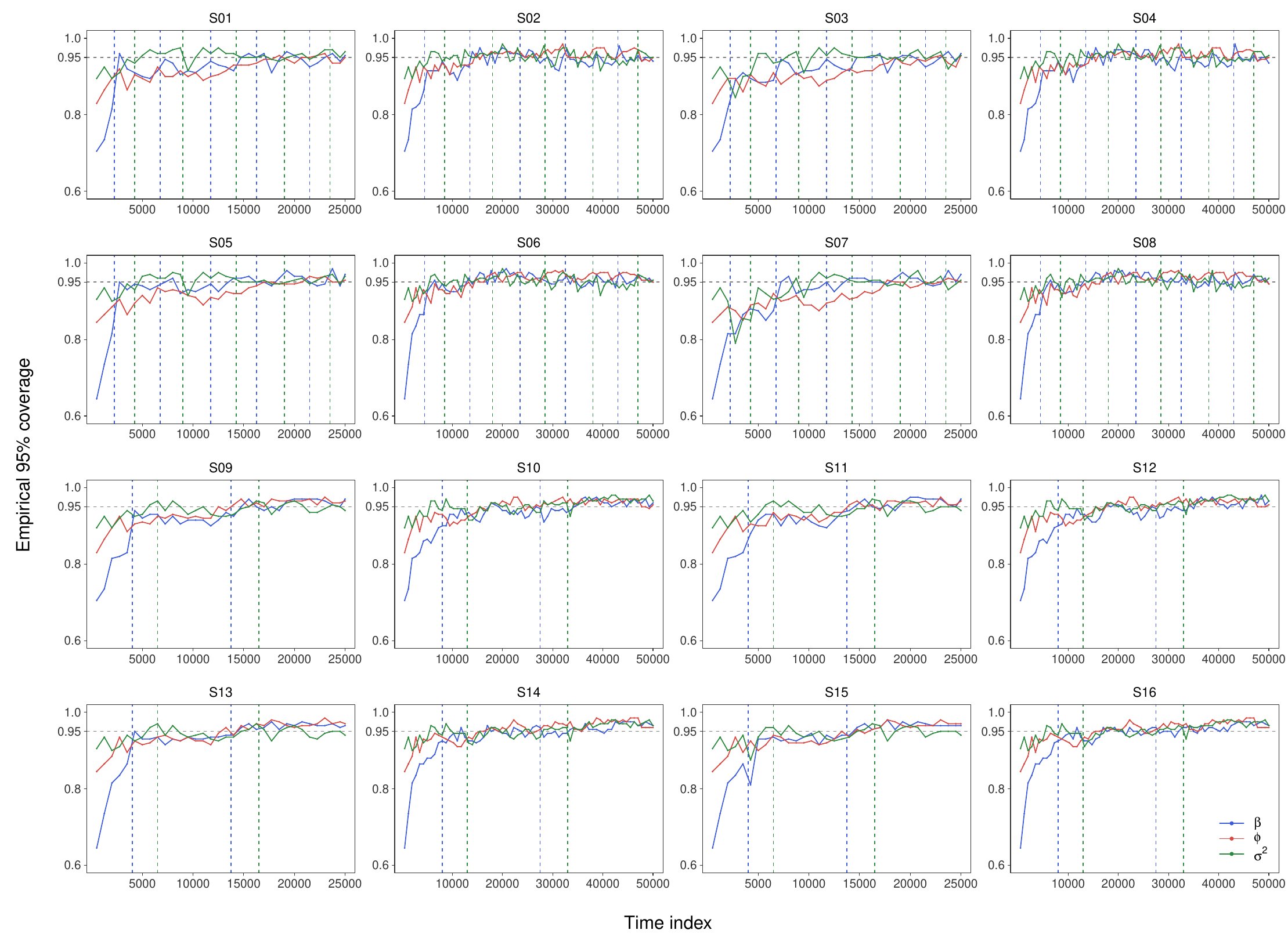}
\caption{
Empirical coverage of nominal 95\% SR-SVB regions across S01--S16.
Every third fitted batch and the final batch are shown.
Vertical dashed lines mark true changepoints; the horizontal line
marks 0.95.
}

\label{fig:adaptive_uq_coverage}
\end{figure*}

Figure~\ref{fig:adaptive_uq_coverage} reports empirical coverage of the
nominal 95\% SR-SVB regions across the sixteen change scenarios.
Coverage is initially below the nominal level for all three parameter
blocks, with the most pronounced undercoverage occurring for
$\boldsymbol{\beta}$.
As additional batches are processed, coverage for
$\boldsymbol{\beta}$, $\boldsymbol{\phi}$, and $\sigma^2$ generally
moves toward the nominal level and subsequently fluctuates around it.
The larger-sample scenarios tend to exhibit smoother coverage
trajectories after this initial adaptation period.
Although temporary departures from 0.95 occur around some changepoints, these deviations are generally followed by recovery and do
not accumulate as further changepoints occur.
In particular, the trajectories do not show a systematic deterioration
in coverage under weaker signals or with more changepoints.

\paragraph{Role of separate reset states.} 
To assess the role of separate reset states in changepoint
detection and parameter inference, we compare SR-SVB with
a restricted joint-reset SVB that permits only
\((z_\beta,z_{\sigma^2})\in\{(0,0),(1,1)\}\),
with all other settings held fixed; full results are reported in Supplementary Section~\ref{sec:supp_reset_ablation}. Across S01–S16, type-agnostic recall decreases from 98.99\% under SR-SVB to 58.54\% under joint resetting. This loss is driven primarily by variance-only changes: variance-change recall decreases from 98.06\% to 17.48\%, whereas regression-change recall remains near 100\%. Thus, forcing the two blocks to reset together can impede adaptation when only the innovation variance changes because adapting the variance then also requires discounting accumulated regression information. 
A complementary simultaneous-change experiment shows an asymmetric trade-off: the joint-reset model adapts fastest immediately after a simultaneous change. SR-SVB detects all simultaneous changes when alarm type is ignored, whereas exact-type joint-change recall is 84.3\% under the default prior and 95.6\% under the balanced prior defined in Supplementary Section~\ref{sec:supp_joint_change}. Thus, whereas restricting resets to occur jointly severely impairs detection of variance-only changes, allowing separate resets retains substantial detection performance even when both blocks change simultaneously. Increasing the prior mass assigned to the joint-reset state improves its early adaptation in this setting. These results support separate reset states when the underlying change structure may differ across parameter blocks, while indicating that the allocation of reset-state prior mass matters when simultaneous changes are common.

\section{Real Data Analysis}
\label{sec:sec5}

We use the Bitcoin application to illustrate separate adaptation of the
regression and innovation-variance blocks while learning residual
autoregressive dependence in a long streaming series.

\paragraph{Dataset.}
We analyze one-minute BTC/USD price and volume data from Bitstamp in the Kaggle 
\emph{Bitcoin Historical Data} dataset maintained by mczielinski.%
\footnote{\url{https://www.kaggle.com/datasets/mczielinski/bitcoin-historical-data}}
Five-minute closes are constructed on a UTC grid, and squared
five-minute log returns are summed over non-overlapping 30-minute
intervals to obtain realized variance. We use five-minute sampling
as a compromise between sensitivity to market microstructure noise
and information loss from coarser sampling. The time-of-week seasonal profile and final
predictor-standardization constants are estimated from the 2014
warm-up sample and held fixed throughout the streaming analysis.
Missing log realized variances in 39 half-hour intervals are replaced
by their warm-up seasonal medians, which can locally attenuate residual
variation. The streaming sample spans January~1, 2015 through
November~26, 2025, with $n=191{,}183$ half-hour intervals.

\paragraph{Variables.}
Let $r_{t,i}$, $i=1,\ldots,6$, denote the six constituent log returns
in half-hour interval $t$. We define $RV_t=\sum_{i=1}^{6}r_{t,i}^{2}$
and use $y_t=\log(RV_t+c)-m_{s(t)}$, where $c$ is a fixed offset
and $m_{s(t)}$ is the warm-up median of $\log(RV+c)$ for the
corresponding weekday and half-hour slot.
The predictors form four groups. Weekly volatility persistence is
represented by the preceding weekly response average
(\texttt{har\_week}). Lagged trading activity and price dynamics are
captured by relative log notional volume
(\texttt{rel\_log\_notional\_lag1}), measured against its recent
same-time-of-week baseline, and the half-hour return
(\texttt{ret\_30m\_lag1}). The 24-hour downside and jump shares
(\texttt{downside\_share\_24h\_lag1},
\texttt{jump\_share\_24h\_lag1}) summarize negative-return variation
and jump-like movements, while the trend ratio
(\texttt{trend\_ratio\_lag1}) measures net directional movement.
Macro controls comprise lagged daily return signals for E-mini Nasdaq-100
futures (\texttt{NQ\_prev\_day}) and the U.S.\ dollar index
(\texttt{DXY\_prev\_day}), together with the sign-reversed daily return signal
for Treasury-note futures as a yield-change proxy
(\texttt{Yield\_prev\_day}). Finally, a U.S.\ regular-hours proxy
(\texttt{us\_open}) and its interaction with lagged relative log
notional volume (\texttt{rel\_volume\_x\_us\_open}) allow the
activity--volatility association to differ during those hours.
Continuous predictors use fixed warm-up centering and scaling;
the intercept and \texttt{us\_open} remain unscaled, and missing
continuous predictors are set to zero after standardization.
Full definitions and interpretations are given in
Table~\ref{tab:btc_predictors}, with preprocessing details in
Supplementary Section~\ref{app:btc_details}.

\paragraph{Model setting.}
We apply SR-SVB with $k=12$ and $p=100$, treating the full regression
vector, including the intercept, and the innovation variance as
separate reset blocks. The initial batch contains $2{,}688$ observations
(eight weeks), followed by batches of up to $1{,}344$ observations
(four weeks), giving 142 batches. We use the diffuse baseline prior
$\pi^{(0)}$ in \eqref{eq:prior}, with the hyperparameters in
Section~\ref{subsec:variational_family} and $\alpha=0$.
The baseline reset parameters are
$h_\beta=h_{\sigma^2}=1-\sqrt{0.9999}\approx5.0\times10^{-5}$,
with $\eta_\beta=\eta_{\sigma^2}=10^{-4}$ and $K_{\max}=16$.
Compared with the simulations, smaller prior reset probabilities favor
retention, while stronger discounting permits greater adaptation
when a reset is selected. CAVI uses an absolute ELBO-change tolerance
of $10^{-3}$. For descriptive reporting, a channel alarm is declared
when its posterior reset probability is at least $0.50$.
Alarms are evaluated at batch ends, whereas reported boundaries
mark batch starts.

\begin{figure}[!t]
    \centering
    \includegraphics[width=0.9\linewidth]
    {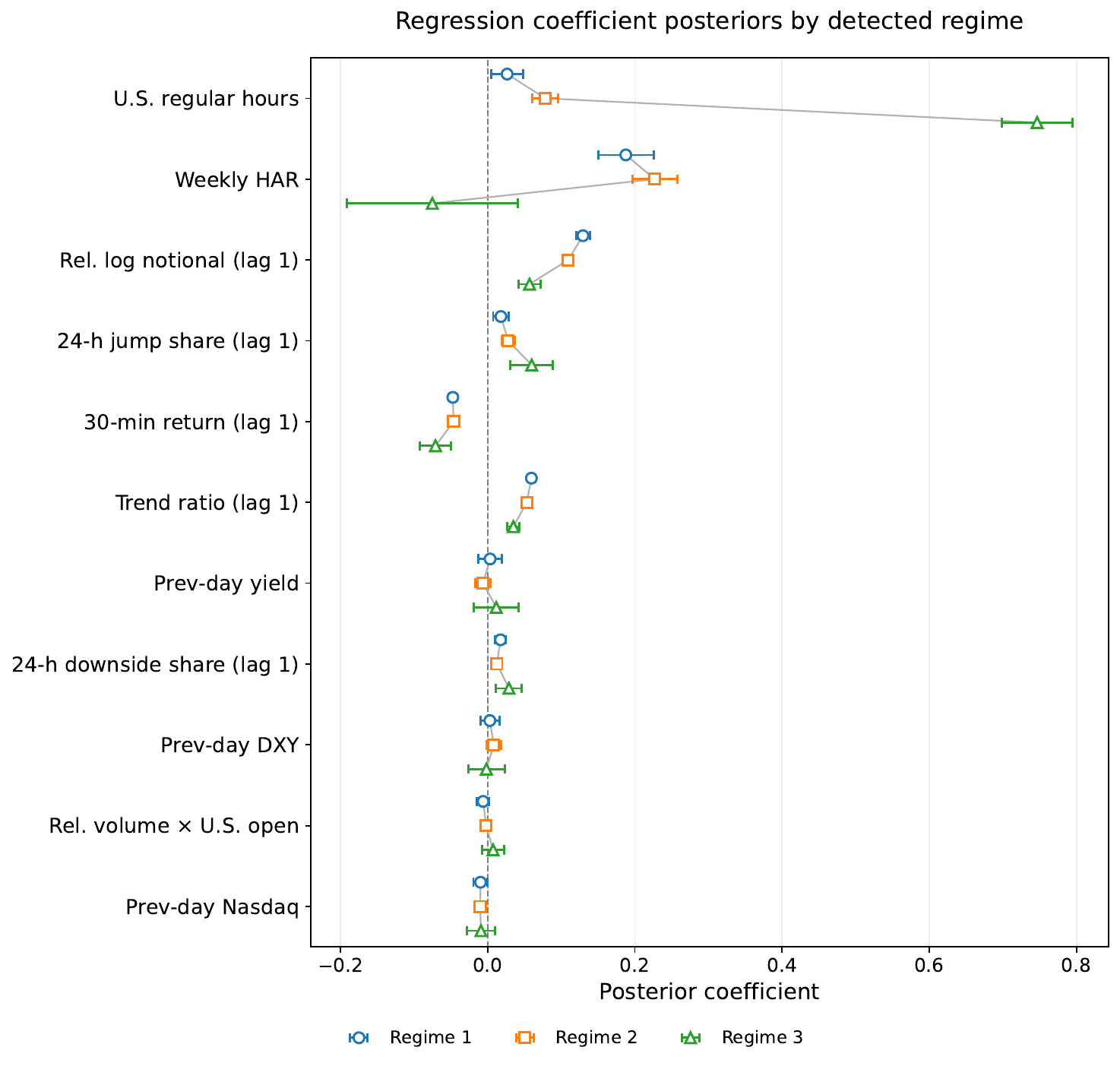}
    \caption{Regime-end variational posterior summaries for the regression
    coefficients, excluding the intercept.
    Open markers denote posterior means, horizontal bars show marginal
    95\% credible intervals, and gray segments connect the means across regimes.}
    \label{fig:btc_beta_regimes}
\end{figure}

\paragraph{Posterior summaries and detected changes.}
Regression-block alarms are reported for the batches beginning
May~13, 2021 and September~26, 2024. The first favors
regression-specific adaptation, with regression- and variance-reset
probabilities of approximately $1.000$ and $1.1\times10^{-4}$.
The second supports joint adaptation, with corresponding probabilities
of $0.998$ and $0.877$ and a joint-reset probability of $0.877$.
These block-level reports neither imply that every coefficient changes
nor identify exact economic break dates.

Figure~\ref{fig:btc_beta_regimes} compares regime-end coefficient
posteriors, conditional on the fixed seasonal adjustment and AR-error
specification. The weekly heterogeneous autoregressive (HAR) predictor and residual AR component
both represent temporal persistence.
The most pronounced difference concerns \texttt{us\_open}, whose
posterior mean is modest in the first two regimes and rises sharply
in the final regime. With the interaction included, this coefficient
represents the fitted U.S.-hours contrast in seasonally adjusted
log realized variance when lagged log notional volume equals its
recent baseline for the same time of week.
The interaction coefficient \texttt{rel\_volume\_x\_us\_open}
remains near zero, with marginal $95\%$ credible intervals containing
zero in all three regimes, providing limited evidence of an additional
activity-slope difference.
The \texttt{har\_week} posterior mean is positive in the first two
regimes and negative in the final regime, although its final credible
interval includes zero. In the final regime, the posterior mean for
\texttt{rel\_log\_notional\_lag1} decreases and that for
\texttt{jump\_share\_24h\_lag1} increases. The
\texttt{ret\_30m\_lag1} coefficient becomes more negative, indicating
a stronger negative conditional association with next-interval
log realized variance. Macro-financial coefficients remain small.

\begin{figure}[!t]
    \centering
    \includegraphics[width=0.95\linewidth]
    {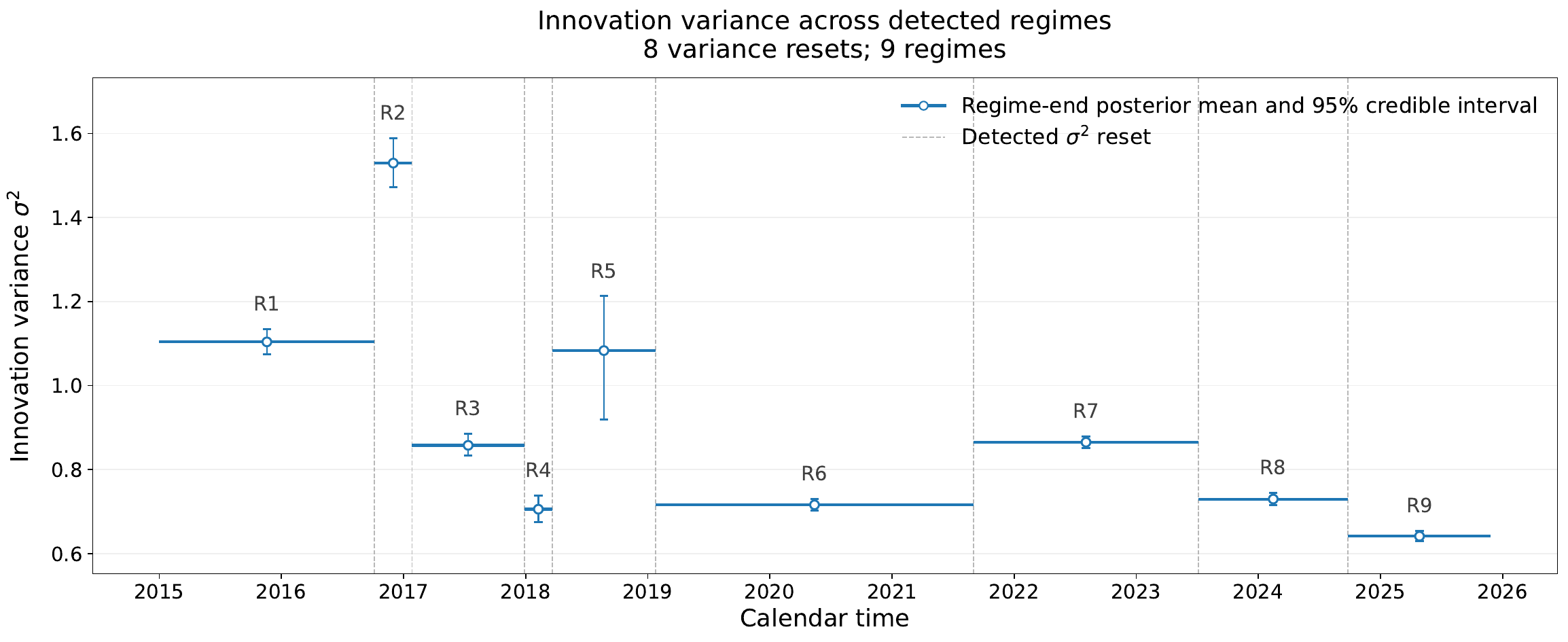}
    \caption{Retrospective innovation-variance summaries.
    Horizontal segments extend regime-end posterior means across the
    displayed intervals; bars show marginal 95\% credible intervals.
    Dashed lines mark reported variance-reset boundaries.}
    \label{fig:btc_sigma2_regimes}
\end{figure}

The variance channel meets the reporting threshold in nine batches.
For retrospective display, consecutive alarmed batches are grouped
and represented by the batch with the largest reset probability,
yielding eight variance-reset episodes and nine working-model regimes.
The next reported variance boundary after the May~2021 regression
boundary is in September~2021; both blocks have boundaries in
September~2024. Figure~\ref{fig:btc_sigma2_regimes} shows regime-end
variance summaries. The largest posterior mean occurs in the
October~2016--January~2017 regime, while the final-regime mean is
approximately $0.641$.

Economic developments provide context for these patterns. The first
reported regression boundary falls in May~2021, the month of the
crypto-market sell-off discussed by the IMF \citep{imf2021crypto},
while the January~2024 approval of U.S.\ spot Bitcoin ETFs provides
institutional context for the later regime \citep{gensler2024spot}.
These developments are not established causes of the reported changes;
in particular, the September~2024 boundary occurs well after the
ETF approval.

Figure~\ref{fig:btc_phi_final} displays the final posterior means and
marginal $95\%$ credible intervals for the AR(100) coefficients.
The largest coefficients occur at the first few half-hour lags, with
smaller positive peaks at lags 48 and 96, corresponding to 24 and
48 hours. Since $\alpha=0$ gives equal prior variances across lags,
these peaks are not imposed by lag-dependent shrinkage. They indicate
residual daily dependence that may reflect temporal structure beyond
the fixed seasonal adjustment and covariates. We regard $p=100$,
spanning slightly more than two days, as a truncation aimed at
capturing the leading short-run and daily components.

\begin{figure}[!t]
    \centering
    \includegraphics[width=0.95\linewidth]
    {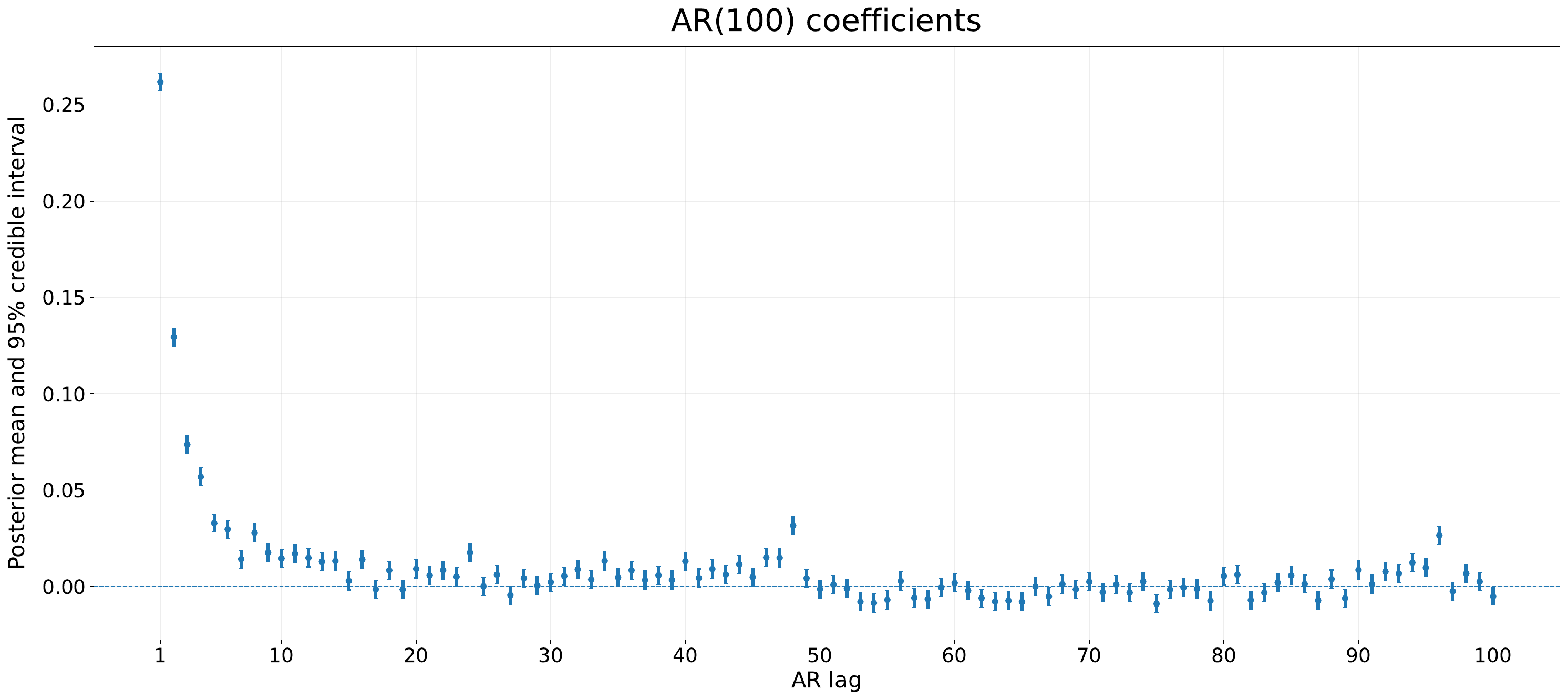}
    \caption{Posterior means and marginal 95\% credible intervals
    for the AR(100) coefficients $\phi_j$ after processing the
    final batch.}
    \label{fig:btc_phi_final}
\end{figure}

\section{Conclusion}\label{sec:sec6}

This study demonstrates that block-specific posterior forgetting provides a useful framework for combining online changepoint detection with sequential Bayesian inference in regression models with serially correlated errors. The numerical studies suggest that adaptivity can be introduced without substantially altering the behavior of the underlying streaming procedure when no structural change is present. When changes occur, the block-specific reset mechanism allows the affected parameter block to adapt without unnecessarily discarding information from blocks that remain stable, while empirical coverage generally recovers toward its nominal level after the transition. The ablation results support the benefit of this selective adaptation most clearly for variance-only changes. The Bitcoin application further illustrates distinct adaptation in the regression and innovation-variance components while residual autoregressive dependence continues to be learned sequentially.

The present formulation treats the regression coefficients and innovation variance as adaptive blocks while sharing the AR coefficients across regimes. Extensions could allow changes in the AR structure or extend changepoint localization to within mini-batches. A theoretical characterization of sequential variational approximation error and posterior uncertainty calibration under repeated structural changes would also be an important direction for future work.

\bibliographystyle{apalike}
\bibliography{ref}

\clearpage

\setcounter{section}{0}
\setcounter{subsection}{0}
\setcounter{subsubsection}{0}
\setcounter{equation}{0}
\setcounter{figure}{0}
\setcounter{table}{0}
\renewcommand{\thesection}{S\arabic{section}}
\renewcommand{\thesubsection}{\thesection.\arabic{subsection}}
\renewcommand{\thesubsubsection}{\thesubsection.\arabic{subsubsection}}
\renewcommand{\theequation}{S\arabic{equation}}
\renewcommand{\thefigure}{S\arabic{figure}}
\renewcommand{\thetable}{S\arabic{table}}

\renewcommand{\theHsection}{S\arabic{section}}
\renewcommand{\theHsubsection}{\theHsection.\arabic{subsection}}
\renewcommand{\theHsubsubsection}{\theHsubsection.\arabic{subsubsection}}
\renewcommand{\theHequation}{S\arabic{equation}}
\renewcommand{\theHfigure}{S\arabic{figure}}
\renewcommand{\theHtable}{S\arabic{table}}

\makeatletter
\def\@currentlabelname{Supplementary Material}
\makeatother
\phantomsection
\label{sec:supplement}
\addcontentsline{toc}{section}{Supplementary Material}
\SRsupplementtitle

\begin{abstract}
This supplementary material provides additional derivations, implementation
details, and sensitivity analyses for the proposed separate-reset streaming
variational Bayes method. We derive the variational updates and evidence lower
bounds under both the previous-mean and full-joint boundary treatments, and
compare their empirical performance. We further examine the role of the reset
structure through separate-versus-joint resetting and simultaneous-change
experiments. Additional details are provided for the changepoint simulation,
including comparator implementation, false-alarm calibration, and event
matching, as well as for the preprocessing, covariate construction, and
reporting procedures used in the Bitcoin application.
\end{abstract}

\section{Derivation of the Variational Distributions}
\label{app:variational_distributions}
This section provides the derivations underlying the variational updates. We first establish the conditional prewhitened likelihood and the corresponding static-posterior benchmark. Section~\ref{app:variational_pm} derives the batchwise variational updates under the previous-mean (PM) boundary treatment, while Section~\ref{app:variational_fj} develops the corresponding full-joint (FJ) formulation and updates that explicitly account for the current and previous regression states across batch boundaries.

We first recall the conditional prewhitened likelihood under the static-parameter benchmark. Conditional on the model parameters and the initial or preceding boundary lags, let $\sigma^2\boldsymbol\Sigma_b$ denote the covariance of $\mathbf y_{b\mid T_b^\ast}$. The within-batch part of the AR filter is represented by the lower-triangular matrix $\mathbf L_b$ with entries $[\mathbf L_b]_{rc}=1\{r=c\}-\sum_{j=1}^p\phi_j1\{r-c=j\}$. The boundary-lag terms contribute only to the conditional mean, whereas the transformed residual $\mathbf y_b^\ast-\mathbf X_b^\ast\boldsymbol\beta$ has covariance $\sigma^2\mathbf I_{n_b^\ast}$. Hence,
\begin{equation*}
\begin{aligned}
\mathbf L_b\boldsymbol\Sigma_b\mathbf L_b^\top
&=\mathbf I_{n_b^\ast},
\quad
\boldsymbol\Sigma_b^{-1}
=\mathbf L_b^\top\mathbf L_b, \\
|\boldsymbol\Sigma_b|
&=|\mathbf L_b^\top\mathbf L_b|^{-1}
=|\mathbf L_b|^{-2}
=1, \\
|\sigma^2\boldsymbol\Sigma_b|^{-1/2}
&=(\sigma^2)^{-n_b^\ast/2}
=|\sigma^2\mathbf I_{n_b^\ast}|^{-1/2},
\end{aligned}
\end{equation*}
where $|\mathbf L_b|=1$ because its diagonal entries are one. Thus the prewhitening transformation has unit Jacobian and introduces no additional determinant factor depending on $\boldsymbol\phi$. Denote the cumulative transformed response vector and design matrix as  $\mathbf{y}_{(B)}^\ast:=(y_t^\ast)_{t\in T_{(B)}^\ast}\in\mathbb{R}^{|T_{(B)}^\ast|}$ and $\mathbf{X}_{(B)}^\ast:=[\mathbf{x}_t^{\ast\top}]_{t\in T_{(B)}^\ast}\in\mathbb{R}^{|T_{(B)}^\ast|\times k}$, respectively. Suppressing the conditioning on the initial observations and the boundary lags in the batchwise factors, the likelihood through batch $B$ is
\begin{equation*}
\begin{aligned}
&p\!\left(
\mathbf y_{(B)}
\mid\mathbf X_{(B)},
\boldsymbol\beta,\sigma^2,\boldsymbol\phi
\right) \\
&\quad=
p\!\left(
\mathbf y_{(B)}^\ast
\mid\mathbf X_{(B)}^\ast,
\boldsymbol\beta,\sigma^2,\boldsymbol\phi
\right) \\
&\quad=
\prod_{b=1}^B
(2\pi\sigma^2)^{-n_b^\ast/2}
\exp\!\left\{
-\frac{1}{2\sigma^2}
\left(\mathbf y_b^\ast-\mathbf X_b^\ast\boldsymbol\beta\right)^\top
\left(\mathbf y_b^\ast-\mathbf X_b^\ast\boldsymbol\beta\right)
\right\}.
\end{aligned}
\end{equation*}
Consequently, the full posterior under this static-parameter model can be written as
\begin{equation*}
\begin{aligned}
&p\!\left(
\boldsymbol\beta,\sigma^2,\boldsymbol\phi
\mid\mathbf y_{(B)},\mathbf X_{(B)}
\right) \\
&\quad\propto
p\!\left(
\mathbf y_{(B)}^\ast
\mid\mathbf X_{(B)}^\ast,
\boldsymbol\beta,\sigma^2,\boldsymbol\phi
\right)
\pi_\beta^{(0)}(\boldsymbol\beta)
\pi_{\sigma^2}^{(0)}(\sigma^2)
\pi_\phi^{(0)}(\boldsymbol\phi).
\end{aligned}
\end{equation*}

\subsection{Previous-Mean Boundary Treatment}
\label{app:variational_pm}

For the batch-$b$ PM update, $\boldsymbol\beta$ and $\sigma^2$ denote the current states $\boldsymbol\beta_b$ and $\sigma_b^2$. We instead use the PM boundary-adjusted quantities $\tilde{\mathbf y}_b^\ast$ and $\tilde{\mathbf X}_b^\ast$ defined in Section~\ref{sec:sec2}. For $b\geq2$, we use the retained mixture mean $\bar{\boldsymbol\beta}_{b-1}$ as a plug-in value in the PM boundary adjustment. Conditional on the retained state and boundary lags, the adjustment changes only the boundary contribution to the conditional mean; the within-batch filter remains $\mathbf L_b$, so its unit-Jacobian calculation above still applies. The PM working likelihood is therefore
\begin{equation*}
\begin{aligned}
p\!\left(
\tilde{\mathbf y}_b^\ast\mid\tilde{\mathbf X}_b^\ast,
\boldsymbol\beta,\sigma^2,\boldsymbol\phi
\right)
&=(2\pi\sigma^2)^{-n_b^\ast/2}
\exp\!\left\{
-\frac{1}{2\sigma^2}
\left\|\tilde{\mathbf y}_b^\ast-
\tilde{\mathbf X}_b^\ast\boldsymbol\beta\right\|^2
\right\}.
\end{aligned}
\end{equation*}
For $b=1$, the adjusted quantities reduce to the ordinary prewhitened quantities. The cumulative posterior above is the static benchmark, not the adaptive PM target. At batch $b$, we combine the PM working likelihood with the child-specific pseudo-prior in Section~\ref{subsec:reset_prior_vf} to obtain the target for the variational updates:

\begin{equation*}
\begin{aligned}
&\tilde p^{(b)}\!\left(
G_b=g,\boldsymbol\beta,\sigma^2,\boldsymbol\phi
\mid\mathbf y_{(b)},\mathbf X_{(b)}
\right) \\
&\quad\propto
\tilde\omega_{b,g}\,
p\!\left(
\tilde{\mathbf y}_b^\ast
\mid\tilde{\mathbf X}_b^\ast,
\boldsymbol\beta,\sigma^2,\boldsymbol\phi
\right)
\tilde\pi_{\beta,g}^{(b)}(\boldsymbol\beta)
\tilde\pi_{\sigma^2,g}^{(b)}(\sigma^2)
q_\phi^{(b-1)}(\boldsymbol\phi),
\quad g\in\mathcal G_b.
\end{aligned}
\end{equation*}
We obtain the variational updates by taking the expected log of this target with respect to the factors held fixed, absorbing terms independent of the updated coordinate into $\mathrm{const}$.
For the proofs below, define
\begin{equation*}
\begin{aligned}
\mathbf C_{xx,b}
&:=\mathbb E_{q_\phi^{(b)}}\!\left[
\tilde{\mathbf X}_b^{\ast\top}\tilde{\mathbf X}_b^\ast\right],\quad
\mathbf C_{xy,b}
:=\mathbb E_{q_\phi^{(b)}}\!\left[
\tilde{\mathbf X}_b^{\ast\top}\tilde{\mathbf y}_b^\ast\right],\\
C_{yy,b}
&:=\mathbb E_{q_\phi^{(b)}}\!\left[
\tilde{\mathbf y}_b^{\ast\top}\tilde{\mathbf y}_b^\ast\right],\quad
\kappa_{b,g}:=\mathbb E_{q_{\sigma^2,g}^{(b)}}(\sigma^{-2}).
\end{aligned}
\end{equation*}
The matrices $\mathbf M_{y,b}$ and $\mathbf F_{b,i}$ below are the PM lag matrices defined in Section~\ref{subsec:subsec3.2}. In particular, for $b\geq2$ and a boundary-crossing lag $j\geq i$, $[\mathbf M_{y,b}]_{ij}=y_{t_b+i-j-1}-\mathbf x_{t_b+i-j-1}^\top\bar{\boldsymbol\beta}_{b-1}$ and $[\mathbf F_{b,\ell}]_{ij}=0$.

\begin{proof}[Proof of Proposition~\ref{prop:component_beta_update_independent}]
For fixed $q_{\sigma^2,g}^{(b)}$ and $q_\phi^{(b)}$, the child-specific target above gives
\begin{equation*}
\begin{aligned}
\log q_{\beta,g}^{(b)}(\boldsymbol\beta)
&=
\mathbb E_{q_{\sigma^2,g}^{(b)}q_\phi^{(b)}}
\!\left[
\log\tilde p^{(b)}\!\left(
G_b=g,\boldsymbol\beta,\sigma^2,\boldsymbol\phi
\mid\mathbf y_{(b)},\mathbf X_{(b)}
\right)
\right]
+\mathrm{const} \\
&=
\mathbb E_{q_{\sigma^2,g}^{(b)}q_\phi^{(b)}}
\!\left[
\log p\!\left(
\tilde{\mathbf y}_b^\ast
\mid\tilde{\mathbf X}_b^\ast,
\boldsymbol\beta,\sigma^2,\boldsymbol\phi
\right)
\right]
+\log\tilde\pi_{\beta,g}^{(b)}(\boldsymbol\beta)
+\mathrm{const}.
\end{aligned}
\end{equation*}
Since
\[
\mathbb E_{q_\phi^{(b)}}
\!\left[
\left\|\tilde{\mathbf y}_b^\ast-\tilde{\mathbf X}_b^\ast\boldsymbol\beta\right\|^2
\right]
=
C_{yy,b}
-2\boldsymbol\beta^\top\mathbf C_{xy,b}
+\boldsymbol\beta^\top\mathbf C_{xx,b}\boldsymbol\beta,
\]
we obtain
\begin{equation*}
\begin{aligned}
\log q_{\beta,g}^{(b)}(\boldsymbol\beta)
&=
-\frac{\kappa_{b,g}}{2}
\left(
C_{yy,b}
-2\boldsymbol\beta^\top\mathbf C_{xy,b}
+\boldsymbol\beta^\top\mathbf C_{xx,b}\boldsymbol\beta
\right) \\
&\quad
-\frac12
\left(\boldsymbol\beta-\tilde{\boldsymbol\beta}_{b,g}\right)^\top
\tilde{\mathbf A}_{b,g}
\left(\boldsymbol\beta-\tilde{\boldsymbol\beta}_{b,g}\right)
+\mathrm{const} \\
&=
-\frac12\boldsymbol\beta^\top
\left(
\tilde{\mathbf A}_{b,g}
+\kappa_{b,g}\mathbf C_{xx,b}
\right)\boldsymbol\beta
+\boldsymbol\beta^\top
\left(
\tilde{\mathbf A}_{b,g}\tilde{\boldsymbol\beta}_{b,g}
+\kappa_{b,g}\mathbf C_{xy,b}
\right)
+\mathrm{const}.
\end{aligned}
\end{equation*}
With
\begin{equation*}
\begin{aligned}
\mathbf V_{b,g}
&=
\tilde{\mathbf A}_{b,g}
+\kappa_{b,g}\mathbf C_{xx,b}, \\
\mathbf m_{b,g}
&=
\mathbf V_{b,g}^{-1}
\left(
\tilde{\mathbf A}_{b,g}\tilde{\boldsymbol\beta}_{b,g}
+\kappa_{b,g}\mathbf C_{xy,b}
\right),
\end{aligned}
\end{equation*}
the preceding kernel becomes
\[
\log q_{\beta,g}^{(b)}(\boldsymbol\beta)
=
-\frac12
\left(\boldsymbol\beta-\mathbf m_{b,g}\right)^\top
\mathbf V_{b,g}
\left(\boldsymbol\beta-\mathbf m_{b,g}\right)
+\mathrm{const}.
\]
Hence $q_{\beta,g}^{(b)}(\boldsymbol\beta)=\mathrm N_k(\boldsymbol\beta\mid\mathbf m_{b,g},\mathbf V_{b,g}^{-1})$.
\end{proof}

\begin{proof}[Proof of Proposition~\ref{prop:component_sigma_update_independent}]
For fixed $q_{\beta,g}^{(b)}$ and $q_\phi^{(b)}$,
\[
\log q_{\sigma^2,g}^{(b)}(\sigma^2)
=
\mathbb E_{q_{\beta,g}^{(b)}q_\phi^{(b)}}
\!\left[
\log p\!\left(\tilde{\mathbf y}_b^\ast\mid\tilde{\mathbf X}_b^\ast,\boldsymbol\beta,\sigma^2,\boldsymbol\phi\right)
\right]
+
\log\tilde\pi_{\sigma^2,g}^{(b)}(\sigma^2)
+
\mathrm{const}.
\]
The expected residual sum of squares is
\begin{equation*}
\begin{aligned}
\mathcal R_{b,g}
&:=
\mathbb E_{q_{\beta,g}^{(b)}q_\phi^{(b)}}
\!\left[
\left\|\tilde{\mathbf y}_b^\ast-\tilde{\mathbf X}_b^\ast\boldsymbol\beta\right\|^2
\right] \\
&=
C_{yy,b}
-2\mathbf m_{b,g}^\top\mathbf C_{xy,b}
+\operatorname{tr}\!\left
\{
\mathbf C_{xx,b}
\mathbb E_{q_{\beta,g}^{(b)}}
\!\left(\boldsymbol\beta\boldsymbol\beta^\top\right)
\right\} \\
&=
C_{yy,b}
-2\mathbf m_{b,g}^\top\mathbf C_{xy,b}
+\operatorname{tr}\!\left[
\mathbf C_{xx,b}
\left(
\mathbf V_{b,g}^{-1}
+\mathbf m_{b,g}\mathbf m_{b,g}^\top
\right)
\right].
\end{aligned}
\end{equation*}
Therefore,
\begin{equation*}
\begin{aligned}
\log q_{\sigma^2,g}^{(b)}(\sigma^2)
&=
-\frac{n_b^\ast}{2}\log\sigma^2
-\frac{\mathcal R_{b,g}}{2\sigma^2}
-\left(\frac{\tilde\nu_{b,g}}{2}+1\right)\log\sigma^2
-\frac{\tilde\delta_{b,g}}{2\sigma^2}
+\mathrm{const} \\
&=
-\left(\frac{n_b^\ast+\tilde\nu_{b,g}}{2}+1\right)\log\sigma^2
-\frac{\tilde\delta_{b,g}+\mathcal R_{b,g}}{2\sigma^2}
+\mathrm{const}.
\end{aligned}
\end{equation*}
Thus $q_{\sigma^2,g}^{(b)}(\sigma^2)=\mathrm{IG}(\sigma^2\mid\alpha_{b,g}/2,\lambda_{b,g}/2)$, where
\[
\alpha_{b,g}=n_b^\ast+\tilde\nu_{b,g},
\quad
\lambda_{b,g}=\tilde\delta_{b,g}+\mathcal R_{b,g},
\]
which is the stated update. \end{proof}

\begin{proof}[Proof of Proposition~\ref{prop:phi_expectations}]
Write $\mathbf X=\mathbf X_{b\mid T_b^\ast}$, $\mathbf y=\mathbf y_{b\mid T_b^\ast}$, $\mathbf M=\mathbf M_{y,b}$, $\mathbf F_i=\mathbf F_{b,i}$, $\boldsymbol\Phi=\mathbf I_k\otimes\boldsymbol\phi$, and $\bar{\boldsymbol\Phi}_b=\mathbf I_k\otimes\mathbf c_b$. Under $q_\phi^{(b)}(\boldsymbol\phi)=\mathrm N_p(\boldsymbol\phi\mid\mathbf c_b,\mathbf D_b^{-1})$,
\[
\mathbb E(\boldsymbol\phi)=\mathbf c_b,
\quad
\mathbb E(\boldsymbol\phi\boldsymbol\phi^\top)=\mathbf D_b^{-1}+\mathbf c_b\mathbf c_b^\top,
\quad
\mathbb E(\boldsymbol\phi^\top\mathbf A\boldsymbol\phi)
=\mathbf c_b^\top\mathbf A\mathbf c_b+\operatorname{tr}(\mathbf A\mathbf D_b^{-1}).
\]
First,
\begin{equation*}
\begin{aligned}
\tilde{\mathbf X}_b^{\ast\top}\tilde{\mathbf X}_b^\ast
&=
(\mathbf X-\mathbf F_b\boldsymbol\Phi)^\top
(\mathbf X-\mathbf F_b\boldsymbol\Phi) \\
&=
\mathbf X^\top\mathbf X
-\mathbf X^\top\mathbf F_b\boldsymbol\Phi
-\boldsymbol\Phi^\top\mathbf F_b^\top\mathbf X
+\boldsymbol\Phi^\top\mathbf F_b^\top\mathbf F_b\boldsymbol\Phi,
\end{aligned}
\end{equation*}
and
\[
\boldsymbol\Phi^\top\mathbf F_b^\top\mathbf F_b\boldsymbol\Phi
=
\left[
\boldsymbol\phi^\top\mathbf F_i^\top\mathbf F_j\boldsymbol\phi
\right]_{i,j=1}^k.
\]
Hence,
\begin{equation*}
\begin{aligned}
\mathbb E_{q_\phi^{(b)}}
\!\left(
\boldsymbol\Phi^\top\mathbf F_b^\top\mathbf F_b\boldsymbol\Phi
\right)
&=
\bar{\boldsymbol\Phi}_b^\top\mathbf F_b^\top\mathbf F_b\bar{\boldsymbol\Phi}_b
+
\left[
\operatorname{tr}\!\left(
\mathbf F_i^\top\mathbf F_j\mathbf D_b^{-1}
\right)
\right]_{i,j=1}^k,
\end{aligned}
\end{equation*}
which gives
\[
\mathbf C_{xx,b}
=
(\mathbf X-\mathbf F_b\bar{\boldsymbol\Phi}_b)^\top
(\mathbf X-\mathbf F_b\bar{\boldsymbol\Phi}_b)
+
\left[
\operatorname{tr}\!\left(
\mathbf F_i^\top\mathbf F_j\mathbf D_b^{-1}
\right)
\right]_{i,j=1}^k.
\]
Second,
\begin{equation*}
\begin{aligned}
\tilde{\mathbf X}_b^{\ast\top}\tilde{\mathbf y}_b^\ast
&=
(\mathbf X-\mathbf F_b\boldsymbol\Phi)^\top
(\mathbf y-\mathbf M\boldsymbol\phi) \\
&=
\mathbf X^\top\mathbf y
-\mathbf X^\top\mathbf M\boldsymbol\phi
-\boldsymbol\Phi^\top\mathbf F_b^\top\mathbf y
+\boldsymbol\Phi^\top\mathbf F_b^\top\mathbf M\boldsymbol\phi,
\end{aligned}
\end{equation*}
where
\[
\boldsymbol\Phi^\top\mathbf F_b^\top\mathbf M\boldsymbol\phi
=
\left[
\boldsymbol\phi^\top\mathbf F_i^\top\mathbf M\boldsymbol\phi
\right]_{i=1}^k.
\]
Therefore,
\begin{equation*}
\begin{aligned}
\mathbb E_{q_\phi^{(b)}}
\!\left(
\boldsymbol\Phi^\top\mathbf F_b^\top\mathbf M\boldsymbol\phi
\right)
&=
\bar{\boldsymbol\Phi}_b^\top\mathbf F_b^\top\mathbf M\mathbf c_b
+
\left[
\operatorname{tr}\!\left(
\mathbf F_i^\top\mathbf M\mathbf D_b^{-1}
\right)
\right]_{i=1}^k,
\end{aligned}
\end{equation*}
so that
\[
\mathbf C_{xy,b}
=
(\mathbf X-\mathbf F_b\bar{\boldsymbol\Phi}_b)^\top
(\mathbf y-\mathbf M\mathbf c_b)
+
\left[
\operatorname{tr}\!\left(
\mathbf F_i^\top\mathbf M\mathbf D_b^{-1}
\right)
\right]_{i=1}^k.
\]
Finally,
\begin{equation*}
\begin{aligned}
\mathbb E_{q_\phi^{(b)}}
\!\left(\tilde{\mathbf y}_b^{\ast\top}\tilde{\mathbf y}_b^\ast\right)
&=
\mathbb E_{q_\phi^{(b)}}
\!\left[
(\mathbf y-\mathbf M\boldsymbol\phi)^\top
(\mathbf y-\mathbf M\boldsymbol\phi)
\right] \\
&=
\mathbf y^\top\mathbf y
-2\mathbf y^\top\mathbf M\mathbf c_b
+\mathbf c_b^\top\mathbf M^\top\mathbf M\mathbf c_b
+\operatorname{tr}\!\left(
\mathbf M^\top\mathbf M\mathbf D_b^{-1}
\right) \\
&=
(\mathbf y-\mathbf M\mathbf c_b)^\top
(\mathbf y-\mathbf M\mathbf c_b)
+\operatorname{tr}\!\left(
\mathbf M^\top\mathbf M\mathbf D_b^{-1}
\right),
\end{aligned}
\end{equation*}
which proves all three identities. \end{proof}

\begin{proof}[Proof of Proposition~\ref{prop:vb_phi}]
For compactness, retain the shorthand from the preceding proof and set $\boldsymbol\Sigma_{b,g}=\mathbf V_{b,g}^{-1}$, $\mathbf r_{b,g}=\mathbf y-\mathbf X\mathbf m_{b,g}$, and $\boldsymbol\delta_{b,g}=\boldsymbol\beta-\mathbf m_{b,g}$. Since $\mathbf E_{b,g}=\mathbf M-\sum_{i=1}^k m_{b,g,i}\mathbf F_i$,
\begin{equation*}
\begin{aligned}
\tilde{\mathbf y}_b^\ast-\tilde{\mathbf X}_b^\ast\boldsymbol\beta
&=
\mathbf y-\mathbf X\boldsymbol\beta
-\left(
\mathbf M-\sum_{i=1}^k\beta_i\mathbf F_i
\right)\boldsymbol\phi \\
&=
\mathbf r_{b,g}-\mathbf E_{b,g}\boldsymbol\phi
-(\mathbf X-\mathbf F_b\boldsymbol\Phi)\boldsymbol\delta_{b,g}.
\end{aligned}
\end{equation*}
Using $\mathbb E(\boldsymbol\delta_{b,g})=0$ and $\mathbb E(\boldsymbol\delta_{b,g}\boldsymbol\delta_{b,g}^\top)=\boldsymbol\Sigma_{b,g}$,
\begin{equation*}
\begin{aligned}
\mathbb E_{q_{\beta,g}^{(b)}}
\!\left[
\left\|\tilde{\mathbf y}_b^\ast-\tilde{\mathbf X}_b^\ast\boldsymbol\beta\right\|^2
\right]
&=
\left\|\mathbf r_{b,g}-\mathbf E_{b,g}\boldsymbol\phi\right\|^2 \\
&\quad+
\operatorname{tr}\!\left[
(\mathbf X-\mathbf F_b\boldsymbol\Phi)^\top
(\mathbf X-\mathbf F_b\boldsymbol\Phi)
\boldsymbol\Sigma_{b,g}
\right].
\end{aligned}
\end{equation*}
The two terms on the right satisfy
\begin{equation*}
\begin{aligned}
\left\|\mathbf r_{b,g}-\mathbf E_{b,g}\boldsymbol\phi\right\|^2
&=
\mathbf r_{b,g}^\top\mathbf r_{b,g}
-2\boldsymbol\phi^\top\mathbf E_{b,g}^\top\mathbf r_{b,g}
+\boldsymbol\phi^\top\mathbf E_{b,g}^\top\mathbf E_{b,g}\boldsymbol\phi, \\
\operatorname{tr}\!\left[
(\mathbf X-\mathbf F_b\boldsymbol\Phi)^\top
(\mathbf X-\mathbf F_b\boldsymbol\Phi)
\boldsymbol\Sigma_{b,g}
\right]
&=
\operatorname{tr}\!\left(\mathbf X^\top\mathbf X\boldsymbol\Sigma_{b,g}\right) \\
&\quad
-2\boldsymbol\phi^\top
\sum_{i=1}^k\sum_{j=1}^k
[\boldsymbol\Sigma_{b,g}]_{ij}
\mathbf F_i^\top\mathbf X\mathbf e_j \\
&\quad
+\boldsymbol\phi^\top
\left(
\sum_{i=1}^k\sum_{j=1}^k
[\boldsymbol\Sigma_{b,g}]_{ij}
\mathbf F_i^\top\mathbf F_j
\right)
\boldsymbol\phi.
\end{aligned}
\end{equation*}
Consequently,
\begin{equation*}
\begin{aligned}
\mathbb E_{q_{\beta,g}^{(b)}}
\!\left[
\left\|\tilde{\mathbf y}_b^\ast-\tilde{\mathbf X}_b^\ast\boldsymbol\beta\right\|^2
\right]
&=
a_{b,g}
-2\boldsymbol\phi^\top
\left[
\mathbf E_{b,g}^\top\mathbf r_{b,g}
+
\sum_{i=1}^k\sum_{j=1}^k
[\mathbf V_{b,g}^{-1}]_{ij}
\mathbf F_i^\top\mathbf X\mathbf e_j
\right] \\
&\quad
+\boldsymbol\phi^\top
\left[
\mathbf E_{b,g}^\top\mathbf E_{b,g}
+
\sum_{i=1}^k\sum_{j=1}^k
[\mathbf V_{b,g}^{-1}]_{ij}
\mathbf F_i^\top\mathbf F_j
\right]
\boldsymbol\phi,
\end{aligned}
\end{equation*}
where $a_{b,g}$ is independent of $\boldsymbol\phi$. The common-factor coordinate update is
\[
\log q_\phi^{(b)}(\boldsymbol\phi)
=
\sum_{g\in\mathcal G_b}
\omega_{b,g}
\mathbb E_{q_{\beta,g}^{(b)}q_{\sigma^2,g}^{(b)}}
\!\left[
\log p\!\left(\tilde{\mathbf y}_b^\ast\mid\tilde{\mathbf X}_b^\ast,\boldsymbol\beta,\sigma^2,\boldsymbol\phi\right)
\right]
+
\log q_\phi^{(b-1)}(\boldsymbol\phi)
+
\mathrm{const}.
\]
Substitution of the preceding quadratic form gives
\begin{equation*}
\begin{aligned}
&\log q_\phi^{(b)}(\boldsymbol\phi)\\
& =  
-\frac12\boldsymbol\phi^\top
\Bigg[
\boldsymbol\Phi_0^{(b-1)}
+
\sum_{g\in\mathcal G_b}
\omega_{b,g}\kappa_{b,g}
\Bigg\{
\mathbf E_{b,g}^\top\mathbf E_{b,g} 
+
\sum_{i=1}^k\sum_{j=1}^k
[\mathbf V_{b,g}^{-1}]_{ij}
\mathbf F_{b,i}^\top\mathbf F_{b,j}
\Bigg\}
\Bigg]
\boldsymbol\phi \\
&\quad
+\boldsymbol\phi^\top
\Bigg\{
\boldsymbol\Phi_0^{(b-1)}\boldsymbol\phi_0^{(b-1)}
+
\sum_{g\in\mathcal G_b}
\omega_{b,g}\kappa_{b,g}
\Bigg[
\mathbf E_{b,g}^\top
\left(
\mathbf y_{b\mid T_b^\ast}
-\mathbf X_{b\mid T_b^\ast}\mathbf m_{b,g}
\right) 
+
\sum_{i=1}^k\sum_{j=1}^k
[\mathbf V_{b,g}^{-1}]_{ij}
\mathbf F_{b,i}^\top
\mathbf X_{b\mid T_b^\ast}\mathbf e_j
\Bigg]
\Bigg\}\\
&\quad +\mathrm{const}.
\end{aligned}
\end{equation*}
The quadratic coefficient is $\mathbf D_b$, and the linear coefficient is $\mathbf D_b\mathbf c_b$ by the definitions in Proposition~\ref{prop:vb_phi}. Therefore,
\[
\log q_\phi^{(b)}(\boldsymbol\phi)
=
-\frac12
(\boldsymbol\phi-\mathbf c_b)^\top
\mathbf D_b
(\boldsymbol\phi-\mathbf c_b)
+\mathrm{const},
\]
so $q_\phi^{(b)}(\boldsymbol\phi)=\mathrm N_p(\boldsymbol\phi\mid\mathbf c_b,\mathbf D_b^{-1})$. \end{proof}

\begin{proof}[Proof of Proposition~\ref{prop:mixture_weight_update}]
For fixed component factors and fixed $q_\phi^{(b)}$, the terms of the batchwise ELBO depending on $\boldsymbol\omega_b=(\omega_{b,g})_{g\in\mathcal G_b}$ are
\[
\mathcal L_b(\boldsymbol\omega_b)
=
\sum_{g\in\mathcal G_b}
\omega_{b,g}
\left
\{
J_{b,g}+\log\tilde\omega_{b,g}-\log\omega_{b,g}
\right\}
+\mathrm{const},
\quad
\sum_{g\in\mathcal G_b}\omega_{b,g}=1.
\]
With multiplier $\zeta$,
\begin{equation*}
\begin{aligned}
\mathcal H(\boldsymbol\omega_b,\zeta)
&=
\sum_{g\in\mathcal G_b}
\omega_{b,g}
\left\{
J_{b,g}+\log\tilde\omega_{b,g}-\log\omega_{b,g}
\right\}
+\zeta
\left(
\sum_{g\in\mathcal G_b}\omega_{b,g}-1
\right), \\
0
&=
\frac{\partial\mathcal H}{\partial\omega_{b,g}}
=
J_{b,g}+\log\tilde\omega_{b,g}-\log\omega_{b,g}-1+\zeta,
\end{aligned}
\end{equation*}
so $\omega_{b,g}=\exp(\zeta-1)\tilde\omega_{b,g}\exp(J_{b,g})$. Enforcing the simplex constraint yields
\[
\omega_{b,g}
=
\frac{
\tilde\omega_{b,g}\exp(J_{b,g})
}{
\displaystyle\sum_{j\in\mathcal G_b}
\tilde\omega_{b,j}\exp(J_{b,j})
},
\quad g\in\mathcal G_b.
\]
The objective is strictly concave on the interior of the simplex, so this stationary point is the unique coordinate maximizer. \end{proof}

\subsection{Full-Joint Boundary Treatment}
\label{app:variational_fj}

The FJ treatment uses the same separate-reset candidate construction, child-weight rule, and regression and variance discount formulas as the PM treatment. In this subsection, the retained factors and their parameters are those produced by the FJ recursion. We write $\sigma^2$ for the current innovation variance $\sigma_b^2$. 
At $b=1$, the singleton candidate uses the baseline prior and ordinary conditional prewhitening, so the PM and FJ updates coincide.

For $b\geq2$, retain the assumption $n_{b-1}\geq p$. To separate the two regression states, define the raw lagged-response matrix $\mathbf M_b\in\mathbb R^{n_b^\ast\times p}$ and the lagged-covariate matrices $\mathbf F_{b,\ell}^{c},\mathbf F_{b,\ell}^{p}\in\mathbb R^{n_b^\ast\times p}$ by
\begin{equation*}
\begin{aligned}
{}[\mathbf M_b]_{ij}
&=y_{t_b+i-j-1},\\
[\mathbf F_{b,\ell}^{c}]_{ij}
&=x_{t_b+i-j-1,\ell}\,1\{j<i\},\quad
[\mathbf F_{b,\ell}^{p}]_{ij}
=x_{t_b+i-j-1,\ell}\,1\{j\geq i\},
\end{aligned}
\end{equation*}
for $i=1,\ldots,n_b^\ast$, $j=1,\ldots,p$, and $\ell=1,\ldots,k$. The superscripts $c$ and $p$ indicate current- and previous-batch contributions. Set $\mathbf F_b^{c}=[\mathbf F_{b,1}^{c}\ \cdots\ \mathbf F_{b,k}^{c}]$, $\mathbf F_b^{p}=[\mathbf F_{b,1}^{p}\ \cdots\ \mathbf F_{b,k}^{p}]$, and $\boldsymbol\Phi=\mathbf I_k\otimes\boldsymbol\phi$. The boundary representation in main-text equation~\eqref{eq:boundary-working} becomes
\begin{equation*}
\begin{aligned}
\mathbf z_b(\boldsymbol\phi)&=\mathbf y_{b\mid T_b^\ast}-\mathbf M_b\boldsymbol\phi,\\
\mathbf A_b(\boldsymbol\phi)&=\mathbf X_{b\mid T_b^\ast}-\mathbf F_b^{c}\boldsymbol\Phi,\quad
\mathbf B_b(\boldsymbol\phi)=-\mathbf F_b^{p}\boldsymbol\Phi,\\
\mathbf z_b(\boldsymbol\phi)
&=\mathbf A_b(\boldsymbol\phi)\boldsymbol\beta_b
+\mathbf B_b(\boldsymbol\phi)\boldsymbol\beta_{b-1}
+\mathbf u_b,
\quad \mathbf u_b\sim\mathrm N_{n_b^\ast}(\mathbf 0,\sigma^2\mathbf I_{n_b^\ast}).
\end{aligned}
\end{equation*}
Thus $\mathbf z_b=\mathbf y_b^\ast$ and $\mathbf A_b+\mathbf B_b=\mathbf X_b^\ast$ for ordinary prewhitening. Conditional on both regression states and the stored boundary lags, the within-batch filter is again $\mathbf L_b$, so no additional determinant factor is introduced.

For a regression-reset candidate $g$ with $z_{\beta,g}=1$, define $\boldsymbol\theta_b=(\boldsymbol\beta_b^\top,\boldsymbol\beta_{b-1}^\top)^\top\in\mathbb R^{2k}$ and $\mathbf W_b(\boldsymbol\phi)=[\mathbf A_b(\boldsymbol\phi)\ \mathbf B_b(\boldsymbol\phi)]$. Its local regression pseudo-prior is
\begin{equation*}
\begin{aligned}
\tilde\pi_{\beta,g}^{(b),\mathrm{FJ}}(\boldsymbol\theta_b)
&=\tilde\pi_{\beta,g}^{(b)}(\boldsymbol\beta_b)
q_{\beta,r_g}^{(b-1)}(\boldsymbol\beta_{b-1})\\
&=\mathrm N_{2k}(\boldsymbol\theta_b\mid\boldsymbol\mu_{0,b,g},\mathbf P_{0,b,g}^{-1}),
\end{aligned}
\end{equation*} where
\begin{align} 
\label{app:mu_P_fj}\boldsymbol\mu_{0,b,g}
=\begin{pmatrix}
\tilde{\boldsymbol\beta}_{b,g}\\
\boldsymbol\beta_{0,r_g}^{(b-1)}
\end{pmatrix},\quad
\mathbf P_{0,b,g}
=\begin{pmatrix}
\tilde{\mathbf A}_{b,g}&\mathbf 0\\
\mathbf 0&\mathbf A_{0,r_g}^{(b-1)}
\end{pmatrix}.
\end{align}
Conditional on the retained state and parent index, the current pseudo-prior parameters are fixed. The candidate target is
\begin{equation*}
\begin{aligned}
&\tilde p_{\mathrm{FJ}}^{(b)}\!\left(
G_b=g,\boldsymbol\theta_b,\sigma^2,\boldsymbol\phi
\mid\mathbf y_{(b)},\mathbf X_{(b)}\right)\\
&\quad\propto\tilde\omega_{b,g}\,
p\!\left(\mathbf z_b\mid\mathbf W_b,\boldsymbol\theta_b,\sigma^2,\boldsymbol\phi\right)
\tilde\pi_{\beta,g}^{(b),\mathrm{FJ}}(\boldsymbol\theta_b)
\tilde\pi_{\sigma^2,g}^{(b)}(\sigma^2)
q_\phi^{(b-1)}(\boldsymbol\phi),
\end{aligned}
\end{equation*}
where the Gaussian likelihood has residual $\mathbf z_b-\mathbf W_b\boldsymbol\theta_b$ and covariance $\sigma^2\mathbf I_{n_b^\ast}$. Its candidate variational factorization is
\begin{equation*}
\tilde q_{\mathrm{FJ}}^{(b)}(G_b=g,\boldsymbol\theta_b,\sigma^2,\boldsymbol\phi)
=\omega_{b,g}\,q_{\beta,g}^{(b)}(\boldsymbol\theta_b)
q_{\sigma^2,g}^{(b)}(\sigma^2)q_\phi^{(b)}(\boldsymbol\phi).
\end{equation*}
The regression factor is joint Gaussian in the two regression states.

For a regression-retention candidate with $z_{\beta,g}=0$, impose $\boldsymbol\beta_b=\boldsymbol\beta_{b-1}$ and use the inherited Gaussian $q_{\beta,r_g}^{(b-1)}(\boldsymbol\beta_b)$ once. The target and candidate approximation are instead defined on the $k$-dimensional regression coordinate:
\begin{equation*}
\begin{aligned}
\tilde p_{\mathrm{FJ}}^{(b)}\!\left(
G_b=g,\boldsymbol\beta_b,\sigma^2,\boldsymbol\phi
\mid\mathbf y_{(b)},\mathbf X_{(b)}\right)&
\propto\tilde\omega_{b,g}\,
p\!\left(\mathbf z_b\mid\mathbf A_b+\mathbf B_b,
\boldsymbol\beta_b,\sigma^2,\boldsymbol\phi\right)
q_{\beta,r_g}^{(b-1)}(\boldsymbol\beta_b)
\tilde\pi_{\sigma^2,g}^{(b)}(\sigma^2)
q_\phi^{(b-1)}(\boldsymbol\phi),\\
\tilde q_{\mathrm{FJ}}^{(b)}(G_b=g,\boldsymbol\beta_b,\sigma^2,\boldsymbol\phi)
&=\omega_{b,g}\,q_{\beta,g}^{(b)}(\boldsymbol\beta_b)
q_{\sigma^2,g}^{(b)}(\sigma^2)q_\phi^{(b)}(\boldsymbol\phi).
\end{aligned}
\end{equation*}
This includes variance-only reset candidates with $(z_{\beta,g},z_{\sigma^2,g})=(0,1)$. Densities and KL divergences are evaluated in the corresponding coordinate space.

\begin{proof}[Derivation of the regression-reset update]
Fix $q_{\sigma^2,g}^{(b)}$ and $q_\phi^{(b)}$, and let $\kappa_{b,g}=\mathbb E_{q_{\sigma^2,g}^{(b)}}(\sigma^{-2})$. Suppressing the argument $\boldsymbol\phi$ of $\mathbf z_b$ and $\mathbf W_b$, the coordinate update is
\begin{equation*}
\begin{aligned}
\log q_{\beta,g}^{(b)}(\boldsymbol\theta_b)
&=-\frac{\kappa_{b,g}}{2}
\mathbb E_{q_\phi^{(b)}}\!\left[
\|\mathbf z_b-\mathbf W_b\boldsymbol\theta_b\|^2\right]\\
&\quad-\frac12(\boldsymbol\theta_b-\boldsymbol\mu_{0,b,g})^\top
\mathbf P_{0,b,g}(\boldsymbol\theta_b-\boldsymbol\mu_{0,b,g})
+\mathrm{const}\\
&=-\frac12\boldsymbol\theta_b^\top\mathbf H_{b,g}\boldsymbol\theta_b
+\boldsymbol\theta_b^\top\mathbf H_{b,g}\mathbf m_{b,g}
+\mathrm{const}.
\end{aligned}
\end{equation*}
Completing the square gives $q_{\beta,g}^{(b)}(\boldsymbol\theta_b)=\mathrm N_{2k}(\boldsymbol\theta_b\mid\mathbf m_{b,g},\boldsymbol\Sigma_{b,g})$, where
\begin{equation*}
\begin{aligned}
\mathbf H_{b,g}
&=\mathbf P_{0,b,g}+\kappa_{b,g}
\mathbb E_{q_\phi^{(b)}}(\mathbf W_b^\top\mathbf W_b),\quad
\boldsymbol\Sigma_{b,g}=\mathbf H_{b,g}^{-1},\\
\mathbf m_{b,g}
&=\boldsymbol\Sigma_{b,g}
\left\{\mathbf P_{0,b,g}\boldsymbol\mu_{0,b,g}
+\kappa_{b,g}\mathbb E_{q_\phi^{(b)}}(\mathbf W_b^\top\mathbf z_b)\right\}.
\end{aligned}
\end{equation*}
Partition the updated mean and covariance as
\begin{equation*}
\mathbf m_{b,g}=\begin{pmatrix}\mathbf m_{b,g}^{c}\\\mathbf m_{b,g}^{p}\end{pmatrix},\quad
\boldsymbol\Sigma_{b,g}=\begin{pmatrix}
\boldsymbol\Sigma_{b,g}^{cc}&\boldsymbol\Sigma_{b,g}^{cp}\\
\boldsymbol\Sigma_{b,g}^{pc}&\boldsymbol\Sigma_{b,g}^{pp}
\end{pmatrix}.
\end{equation*}
For $q_\phi^{(b)}=\mathrm N_p(\mathbf c_b,\mathbf D_b^{-1})$ and affine matrix functions $\mathbf U(\boldsymbol\phi)=\mathbf U_0-\sum_{i=1}^p\phi_i\mathbf U_i$ and $\mathbf V(\boldsymbol\phi)=\mathbf V_0-\sum_{i=1}^p\phi_i\mathbf V_i$ of compatible dimensions, the calculation in the proof of Proposition~\ref{prop:phi_expectations} gives
\begin{equation*}
\mathbb E_{q_\phi^{(b)}}\!\left[\mathbf U(\boldsymbol\phi)^\top\mathbf V(\boldsymbol\phi)\right]
=\mathbf U(\mathbf c_b)^\top\mathbf V(\mathbf c_b)
+\sum_{i=1}^p\sum_{j=1}^p[\mathbf D_b^{-1}]_{ij}\mathbf U_i^\top\mathbf V_j.
\end{equation*}
Taking $\mathbf U,\mathbf V$ from $\mathbf A_b$, $\mathbf B_b$, and $\mathbf z_b$ gives the required expectations.
\end{proof}

\begin{proof}[Derivation of the regression-retention update]
Under $\boldsymbol\beta_b=\boldsymbol\beta_{b-1}$, the boundary residual is $\mathbf z_b-(\mathbf A_b+\mathbf B_b)\boldsymbol\beta_b=\mathbf y_b^\ast-\mathbf X_b^\ast\boldsymbol\beta_b$. Hence,
\begin{equation*}
\begin{aligned}
\log q_{\beta,g}^{(b)}(\boldsymbol\beta_b)
&=-\frac{\kappa_{b,g}}{2}
\mathbb E_{q_\phi^{(b)}}\!\left[
\|\mathbf z_b-(\mathbf A_b+\mathbf B_b)\boldsymbol\beta_b\|^2\right]\\
&\quad-\frac12(\boldsymbol\beta_b-\boldsymbol\beta_{0,r_g}^{(b-1)})^\top
\mathbf A_{0,r_g}^{(b-1)}
(\boldsymbol\beta_b-\boldsymbol\beta_{0,r_g}^{(b-1)})+\mathrm{const}\\
&=-\frac12\boldsymbol\beta_b^\top\mathbf V_{b,g}\boldsymbol\beta_b
+\boldsymbol\beta_b^\top\mathbf V_{b,g}\mathbf m_{b,g}+\mathrm{const},
\end{aligned}
\end{equation*}
where
\begin{equation*}
\begin{aligned}
\mathbf V_{b,g}
&=\mathbf A_{0,r_g}^{(b-1)}+\kappa_{b,g}
\mathbb E_{q_\phi^{(b)}}\!\left[(\mathbf A_b+\mathbf B_b)^\top(\mathbf A_b+\mathbf B_b)\right],\\
\mathbf m_{b,g}
&=\mathbf V_{b,g}^{-1}\left\{
\mathbf A_{0,r_g}^{(b-1)}\boldsymbol\beta_{0,r_g}^{(b-1)}
+\kappa_{b,g}\mathbb E_{q_\phi^{(b)}}\!\left[(\mathbf A_b+\mathbf B_b)^\top\mathbf z_b\right]
\right\}.
\end{aligned}
\end{equation*}
Completing the square yields $q_{\beta,g}^{(b)}(\boldsymbol\beta_b)=\mathrm N_k(\boldsymbol\beta_b\mid\mathbf m_{b,g},\mathbf V_{b,g}^{-1})$.
\end{proof}

\begin{proof}[Derivation of the innovation-variance update]
For a regression-reset candidate, define
\begin{equation*}
\begin{aligned}
\mathcal R_{b,g}^{\mathrm{FJ}}
&:=\mathbb E_{q_{\beta,g}^{(b)}q_\phi^{(b)}}
\!\left[\|\mathbf z_b-\mathbf W_b\boldsymbol\theta_b\|^2\right]\\
&=\mathbb E_{q_\phi^{(b)}}\!\left[
\|\mathbf z_b-\mathbf W_b\mathbf m_{b,g}\|^2
+\operatorname{tr}(\mathbf W_b\boldsymbol\Sigma_{b,g}\mathbf W_b^\top)
\right].
\end{aligned}
\end{equation*}
For a regression-retention candidate, the corresponding quantity is
\begin{equation*}
\begin{aligned}
\mathcal R_{b,g}^{\mathrm{FJ}}
&=\mathbb E_{q_\phi^{(b)}}\Big[
\|\mathbf z_b-(\mathbf A_b+\mathbf B_b)\mathbf m_{b,g}\|^2\\
&\quad+\operatorname{tr}\!\left\{
(\mathbf A_b+\mathbf B_b)^\top(\mathbf A_b+\mathbf B_b)\mathbf V_{b,g}^{-1}
\right\}\Big].
\end{aligned}
\end{equation*}
In both cases, the same inverse-gamma pseudo-prior gives
\begin{equation*}
\begin{aligned}
\log q_{\sigma^2,g}^{(b)}(\sigma^2)
&=-\left(\frac{n_b^\ast+\tilde\nu_{b,g}}{2}+1\right)\log\sigma^2
-\frac{\tilde\delta_{b,g}+\mathcal R_{b,g}^{\mathrm{FJ}}}{2\sigma^2}
+\mathrm{const}.
\end{aligned}
\end{equation*}
Consequently, $q_{\sigma^2,g}^{(b)}=\mathrm{IG}(\alpha_{b,g}/2,\lambda_{b,g}/2)$ with $\alpha_{b,g}=n_b^\ast+\tilde\nu_{b,g}$, $\lambda_{b,g}=\tilde\delta_{b,g}+\mathcal R_{b,g}^{\mathrm{FJ}}$, and $\kappa_{b,g}=\alpha_{b,g}/\lambda_{b,g}$.
\end{proof}

\begin{proof}[Derivation of the shared AR update]
Define the lagged-residual matrix as a function of the two regression states:
\begin{equation*}
\mathbf E_b(\boldsymbol\beta_b,\boldsymbol\beta_{b-1})
=\mathbf M_b-\sum_{i=1}^k\beta_{b,i}\mathbf F_{b,i}^{c}
-\sum_{i=1}^k\beta_{b-1,i}\mathbf F_{b,i}^{p}.
\end{equation*}
Then the innovation vector is
\begin{equation*}
\mathbf z_b-\mathbf A_b\boldsymbol\beta_b-\mathbf B_b\boldsymbol\beta_{b-1}
=\mathbf y_{b\mid T_b^\ast}-\mathbf X_{b\mid T_b^\ast}\boldsymbol\beta_b
-\mathbf E_b(\boldsymbol\beta_b,\boldsymbol\beta_{b-1})\boldsymbol\phi.
\end{equation*}
For a regression-reset candidate, write
\begin{equation*}
\mathbf r_{b,g}=\mathbf y_{b\mid T_b^\ast}-\mathbf X_{b\mid T_b^\ast}\mathbf m_{b,g}^{c},\quad
\mathbf E_{b,g}=\mathbf M_b
-\sum_{i=1}^k m_{b,g,i}^{c}\mathbf F_{b,i}^{c}
-\sum_{i=1}^k m_{b,g,i}^{p}\mathbf F_{b,i}^{p}.
\end{equation*}
Expanding around the joint regression mean gives the two expectations needed for the AR quadratic:
\begin{equation*}
\begin{aligned}
\mathbb E_{q_{\beta,g}^{(b)}}(\mathbf E_b^\top\mathbf E_b)
&=\mathbf E_{b,g}^\top\mathbf E_{b,g}\\
&\quad+\sum_{i=1}^k\sum_{j=1}^k\Big[
[\boldsymbol\Sigma_{b,g}^{cc}]_{ij}(\mathbf F_{b,i}^{c})^\top\mathbf F_{b,j}^{c}
+[\boldsymbol\Sigma_{b,g}^{cp}]_{ij}(\mathbf F_{b,i}^{c})^\top\mathbf F_{b,j}^{p}\\
&\hspace{27mm}
+[\boldsymbol\Sigma_{b,g}^{pc}]_{ij}(\mathbf F_{b,i}^{p})^\top\mathbf F_{b,j}^{c}
+[\boldsymbol\Sigma_{b,g}^{pp}]_{ij}(\mathbf F_{b,i}^{p})^\top\mathbf F_{b,j}^{p}
\Big],\\
\mathbb E_{q_{\beta,g}^{(b)}}\!\left[\mathbf E_b^\top(\mathbf y_{b\mid T_b^\ast}-\mathbf X_{b\mid T_b^\ast}\boldsymbol\beta_b)\right]
&=\mathbf E_{b,g}^\top\mathbf r_{b,g}\\
&\quad+\sum_{i=1}^k\sum_{j=1}^k\Big[
[\boldsymbol\Sigma_{b,g}^{cc}]_{ij}(\mathbf F_{b,i}^{c})^\top\mathbf X_{b\mid T_b^\ast}\mathbf e_j
+[\boldsymbol\Sigma_{b,g}^{pc}]_{ij}(\mathbf F_{b,i}^{p})^\top\mathbf X_{b\mid T_b^\ast}\mathbf e_j
\Big].
\end{aligned}
\end{equation*}
Here $\mathbf e_j$ is the $j$th standard basis vector in $\mathbb R^k$.

For a retention candidate, the same moment identities hold after setting $\mathbf m_{b,g}^{c}=\mathbf m_{b,g}^{p}=\mathbf m_{b,g}$ and $\boldsymbol\Sigma_{b,g}^{cc}=\boldsymbol\Sigma_{b,g}^{cp}=\boldsymbol\Sigma_{b,g}^{pc}=\boldsymbol\Sigma_{b,g}^{pp}=\mathbf V_{b,g}^{-1}$. These substitutions are used only for the moment calculations.

For either candidate type, collecting the terms in $\boldsymbol\phi$ gives
\begin{equation*}
\log q_\phi^{(b)}(\boldsymbol\phi)
=-\frac12\boldsymbol\phi^\top\mathbf D_b\boldsymbol\phi
+\boldsymbol\phi^\top\mathbf D_b\mathbf c_b+\mathrm{const},
\end{equation*}
where
\begin{equation*}
\begin{aligned}
\mathbf D_b
&=\boldsymbol\Phi_0^{(b-1)}
+\sum_{g\in\mathcal G_b}\omega_{b,g}\kappa_{b,g}
\mathbb E_{q_{\beta,g}^{(b)}}(\mathbf E_b^\top\mathbf E_b),\\
\mathbf c_b
&=\mathbf D_b^{-1}\left\{
\boldsymbol\Phi_0^{(b-1)}\boldsymbol\phi_0^{(b-1)}
+\sum_{g\in\mathcal G_b}\omega_{b,g}\kappa_{b,g}
\mathbb E_{q_{\beta,g}^{(b)}}\!\left[
\mathbf E_b^\top(\mathbf y_{b\mid T_b^\ast}-\mathbf X_{b\mid T_b^\ast}\boldsymbol\beta_b)\right]
\right\}.
\end{aligned}
\end{equation*}
Completing the square gives $q_\phi^{(b)}(\boldsymbol\phi)=\mathrm N_p(\boldsymbol\phi\mid\mathbf c_b,\mathbf D_b^{-1})$.
\end{proof}

\begin{proof}[Derivation of the candidate scores and weights]
For fixed continuous factors, the calculation in the proof of Proposition~\ref{prop:mixture_weight_update} applies with the component-local score $J_{b,g}^{\mathrm{FJ}}$ derived in Section~\ref{app:elbo_fj}, giving
\begin{equation*}
\omega_{b,g}
=\frac{\tilde\omega_{b,g}\exp(J_{b,g}^{\mathrm{FJ}})}
{\displaystyle\sum_{j\in\mathcal G_b}\tilde\omega_{b,j}\exp(J_{b,j}^{\mathrm{FJ}})}.
\end{equation*}
\end{proof}

After convergence, the reset probabilities are computed from the full candidate weights as in main-text equation~\eqref{eq:reset_probs}. For a regression-reset candidate, the current regression marginal is
\begin{equation*}
\begin{aligned}
q_{\beta,g}^{(b),c}(\boldsymbol\beta_b)
&=\int q_{\beta,g}^{(b)}(\boldsymbol\beta_b,\boldsymbol\beta_{b-1})
\,d\boldsymbol\beta_{b-1}\\
&=\mathrm N_k(\boldsymbol\beta_b\mid\mathbf m_{b,g}^{c},\boldsymbol\Sigma_{b,g}^{cc}).
\end{aligned}
\end{equation*}
For retention, set $q_{\beta,g}^{(b),c}=q_{\beta,g}^{(b)}$, $\mathbf m_{b,g}^{c}=\mathbf m_{b,g}$, and $\boldsymbol\Sigma_{b,g}^{cc}=\mathbf V_{b,g}^{-1}$. Retain $\mathcal I_b=\{g_{b,1},\ldots,g_{b,K_b}\}$ by the same pruning rule as in the main text, and renormalize $\omega_{b,r}=\omega_{b,g_{b,r}}/\sum_{g\in\mathcal I_b}\omega_{b,g}$. The propagated state is
\begin{equation*}
q^{(b)}(\boldsymbol\beta,\sigma^2,\boldsymbol\phi)
=\left\{\sum_{r=1}^{K_b}\omega_{b,r}\,
q_{\beta,g_{b,r}}^{(b),c}(\boldsymbol\beta)
q_{\sigma^2,g_{b,r}}^{(b)}(\sigma^2)\right\}q_\phi^{(b)}(\boldsymbol\phi),
\end{equation*}
with stored parameters
\begin{equation*}
\begin{alignedat}{2}
\boldsymbol\beta_{0,r}^{(b)}&\leftarrow\mathbf m_{b,g_{b,r}}^{c},
&\quad\mathbf A_{0,r}^{(b)}&\leftarrow(\boldsymbol\Sigma_{b,g_{b,r}}^{cc})^{-1},\\
\nu_{0,r}^{(b)}&\leftarrow\alpha_{b,g_{b,r}},
&\delta_{0,r}^{(b)}&\leftarrow\lambda_{b,g_{b,r}},\\
\boldsymbol\phi_0^{(b)}&\leftarrow\mathbf c_b,
&\boldsymbol\Phi_0^{(b)}&\leftarrow\mathbf D_b.
\end{alignedat}
\end{equation*}
The previous regression vector and its cross-covariance are used only within the current batch.

\section{Derivation of the ELBO}
\label{app:elbo}

\subsection{Previous-Mean Boundary Treatment}
\label{app:elbo_pm}
We use the PM working likelihood and the moments defined in Section~\ref{app:variational_pm}.
At batch $b$, the ELBO associated with the augmented candidate distribution $\tilde q^{(b)}$ is
\begin{equation*}
\begin{aligned}
\mathcal L_b(q)
&=
\mathbb E_{\tilde q^{(b)}}
\Big[
\log p\!\left(
\tilde{\mathbf y}_b^\ast
\mid
\tilde{\mathbf X}_b^\ast,
\boldsymbol\beta,
\sigma^2,
\boldsymbol\phi
\right)
+
\log \tilde\omega_{b,G_b}
+
\log \tilde\pi_{\beta,G_b}^{(b)}(\boldsymbol\beta) \\
&\hspace{24mm}
+
\log \tilde\pi_{\sigma^2,G_b}^{(b)}(\sigma^2)
+
\log q_\phi^{(b-1)}(\boldsymbol\phi)
-
\log \omega_{b,G_b}
-
\log q_{\beta,G_b}^{(b)}(\boldsymbol\beta) \\
&\hspace{24mm}
-
\log q_{\sigma^2,G_b}^{(b)}(\sigma^2)
-
\log q_\phi^{(b)}(\boldsymbol\phi)
\Big].
\end{aligned}
\end{equation*}
The factorization over $G_b$ gives
\begin{equation*}
\begin{aligned}
\mathbb E_{\tilde q^{(b)}}
\!\left[
\log \tilde q^{(b)}
\right]
&=
\sum_{g\in\mathcal G_b}
\omega_{b,g}\log\omega_{b,g}
+
\sum_{g\in\mathcal G_b}
\omega_{b,g}
\mathbb E_{q_{\beta,g}^{(b)}}
\!\left[
\log q_{\beta,g}^{(b)}(\boldsymbol\beta)
\right] \\
&\quad+
\sum_{g\in\mathcal G_b}
\omega_{b,g}
\mathbb E_{q_{\sigma^2,g}^{(b)}}
\!\left[
\log q_{\sigma^2,g}^{(b)}(\sigma^2)
\right]
+
\mathbb E_{q_\phi^{(b)}}
\!\left[
\log q_\phi^{(b)}(\boldsymbol\phi)
\right].
\end{aligned}
\end{equation*}
Let $\psi$ denote the digamma function. For $g\in\mathcal G_b$, set $\ell_{b,g}=\log(\lambda_{b,g}/2)-\psi(\alpha_{b,g}/2)$ and $\kappa_{b,g}=\alpha_{b,g}/\lambda_{b,g}$, and write
\[
\mathcal R_{b,g}
=
C_{yy,b}
-2\mathbf m_{b,g}^\top\mathbf C_{xy,b}
+
\operatorname{tr}\!\left[
\mathbf C_{xx,b}
\left(
\mathbf V_{b,g}^{-1}
+
\mathbf m_{b,g}\mathbf m_{b,g}^\top
\right)
\right].
\]
Then
\begin{equation*}
\begin{aligned}
&\mathbb E_{q_{\beta,g}^{(b)}q_{\sigma^2,g}^{(b)}q_\phi^{(b)}}
\!\left[
\log p\!\left(
\tilde{\mathbf y}_b^\ast
\mid
\tilde{\mathbf X}_b^\ast,
\boldsymbol\beta,
\sigma^2,
\boldsymbol\phi
\right)
\right] \\
&\quad=
-\frac{n_b^\ast}{2}\log(2\pi)
-\frac{n_b^\ast}{2}\ell_{b,g}
-\frac{\kappa_{b,g}}{2}\mathcal R_{b,g}.
\end{aligned}
\end{equation*}
Moreover,
\begin{equation*}
\begin{aligned}
\mathbb E_{q_{\beta,g}^{(b)}}
\!\left[
\log q_{\beta,g}^{(b)}(\boldsymbol\beta)
\right]
&=
\mathbb E_{q_{\beta,g}^{(b)}}
\!\left[
-\frac{k}{2}\log(2\pi)
+\frac12\log|\mathbf V_{b,g}|
-\frac12
(\boldsymbol\beta-\mathbf m_{b,g})^\top
\mathbf V_{b,g}
(\boldsymbol\beta-\mathbf m_{b,g})
\right] \\
&=
-\frac{k}{2}\log(2\pi)
+\frac12\log|\mathbf V_{b,g}|
-\frac{k}{2}, \\
\mathbb E_{q_{\sigma^2,g}^{(b)}}
\!\left[
\log q_{\sigma^2,g}^{(b)}(\sigma^2)
\right]
&=
\frac{\alpha_{b,g}}{2}
\log\!\left(\frac{\lambda_{b,g}}{2}\right)
-
\log\Gamma\!\left(\frac{\alpha_{b,g}}{2}\right)
-
\left(\frac{\alpha_{b,g}}{2}+1\right)\ell_{b,g}
-
\frac{\alpha_{b,g}}{2}, \\
\mathbb E_{q_\phi^{(b)}}
\!\left[
\log q_\phi^{(b)}(\boldsymbol\phi)
\right]
&=
\mathbb E_{q_\phi^{(b)}}
\!\left[
-\frac{p}{2}\log(2\pi)
+\frac12\log|\mathbf D_b|
-\frac12
(\boldsymbol\phi-\mathbf c_b)^\top
\mathbf D_b
(\boldsymbol\phi-\mathbf c_b)
\right] \\
&=
-\frac{p}{2}\log(2\pi)
+\frac12\log|\mathbf D_b|
-\frac{p}{2}.
\end{aligned}
\end{equation*}
The corresponding pseudo-prior expectations are
\begin{equation*}
\begin{aligned}
\mathbb E_{q_{\beta,g}^{(b)}}
\!\left[
\log \tilde\pi_{\beta,g}^{(b)}(\boldsymbol\beta)
\right]
&=
-\frac{k}{2}\log(2\pi)
+\frac12\log|\tilde{\mathbf A}_{b,g}| \\
&\quad-
\frac12
\operatorname{tr}\!\left(
\tilde{\mathbf A}_{b,g}\mathbf V_{b,g}^{-1}
\right)
-
\frac12
\left(
\mathbf m_{b,g}-\tilde{\boldsymbol\beta}_{b,g}
\right)^\top
\tilde{\mathbf A}_{b,g}
\left(
\mathbf m_{b,g}-\tilde{\boldsymbol\beta}_{b,g}
\right), \\
\mathbb E_{q_{\sigma^2,g}^{(b)}}
\!\left[
\log \tilde\pi_{\sigma^2,g}^{(b)}(\sigma^2)
\right]
&=
\frac{\tilde\nu_{b,g}}{2}
\log\!\left(\frac{\tilde\delta_{b,g}}{2}\right)
-
\log\Gamma\!\left(\frac{\tilde\nu_{b,g}}{2}\right) \\
&\quad-
\left(\frac{\tilde\nu_{b,g}}{2}+1\right)\ell_{b,g}
-
\frac{\tilde\delta_{b,g}}{2}\kappa_{b,g}, \\
\mathbb E_{q_\phi^{(b)}}
\!\left[
\log q_\phi^{(b-1)}(\boldsymbol\phi)
\right]
&=
-\frac{p}{2}\log(2\pi)
+\frac12\log|\boldsymbol\Phi_0^{(b-1)}| \\
&\quad-
\frac12
\operatorname{tr}\!\left(
\boldsymbol\Phi_0^{(b-1)}\mathbf D_b^{-1}
\right)
-
\frac12
\left(
\mathbf c_b-\boldsymbol\phi_0^{(b-1)}
\right)^\top
\boldsymbol\Phi_0^{(b-1)}
\left(
\mathbf c_b-\boldsymbol\phi_0^{(b-1)}
\right).
\end{aligned}
\end{equation*}
Hence,
\begin{equation*}
\begin{aligned}
\mathcal K_{\beta,b,g}
&:=
\operatorname{KL}\!\left(
q_{\beta,g}^{(b)}
\,\|\,
\tilde\pi_{\beta,g}^{(b)}
\right) \\
&=
\frac12
\Bigg[
\operatorname{tr}\!\left(
\tilde{\mathbf A}_{b,g}\mathbf V_{b,g}^{-1}
\right)
+
\left(
\mathbf m_{b,g}-\tilde{\boldsymbol\beta}_{b,g}
\right)^\top
\tilde{\mathbf A}_{b,g}
\left(
\mathbf m_{b,g}-\tilde{\boldsymbol\beta}_{b,g}
\right) \\
&\hspace{22mm}
-k
+
\log|\mathbf V_{b,g}|
-
\log|\tilde{\mathbf A}_{b,g}|
\Bigg], \\
\mathcal K_{\sigma^2,b,g}
&:=
\operatorname{KL}\!\left(
q_{\sigma^2,g}^{(b)}
\,\|\,
\tilde\pi_{\sigma^2,g}^{(b)}
\right) \\
&=
\frac{\tilde\nu_{b,g}}{2}
\log\!\left(
\frac{\lambda_{b,g}}{\tilde\delta_{b,g}}
\right)
+
\log\Gamma\!\left(
\frac{\tilde\nu_{b,g}}{2}
\right)
-
\log\Gamma\!\left(
\frac{\alpha_{b,g}}{2}
\right) \\
&\quad+
\frac{\alpha_{b,g}-\tilde\nu_{b,g}}{2}
\psi\!\left(
\frac{\alpha_{b,g}}{2}
\right)
-
\frac{\alpha_{b,g}}{2}
+
\frac{\alpha_{b,g}\tilde\delta_{b,g}}{2\lambda_{b,g}}, \\
\mathcal K_{\phi,b}
&:=
\operatorname{KL}\!\left(
q_\phi^{(b)}
\,\|\,
q_\phi^{(b-1)}
\right) \\
&=
\frac12
\Bigg[
\operatorname{tr}\!\left(
\boldsymbol\Phi_0^{(b-1)}\mathbf D_b^{-1}
\right)
+
\left(
\mathbf c_b-\boldsymbol\phi_0^{(b-1)}
\right)^\top
\boldsymbol\Phi_0^{(b-1)}
\left(
\mathbf c_b-\boldsymbol\phi_0^{(b-1)}
\right) \\
&\hspace{22mm}
-p
+
\log|\mathbf D_b|
-
\log|\boldsymbol\Phi_0^{(b-1)}|
\Bigg].
\end{aligned}
\end{equation*}
With
\[
J_{b,g}
=
-\frac{n_b^\ast}{2}\log(2\pi)
-\frac{n_b^\ast}{2}\ell_{b,g}
-\frac{\kappa_{b,g}}{2}\mathcal R_{b,g}
-\mathcal K_{\beta,b,g}
-\mathcal K_{\sigma^2,b,g},
\]
the batchwise ELBO is
\[
\mathcal L_b(q)
=
\sum_{g\in\mathcal G_b}
\omega_{b,g}
\left[
J_{b,g}
+
\log\tilde\omega_{b,g}
-
\log\omega_{b,g}
\right]
-
\mathcal K_{\phi,b}.
\]
Equivalently,
\begin{equation*}
\begin{aligned}
\mathcal L_b(q)
&=
\sum_{g\in\mathcal G_b}
\omega_{b,g}
\Bigg[
-\frac{n_b^\ast}{2}\log(2\pi)
-\frac{n_b^\ast}{2}
\left\{
\log\!\left(\frac{\lambda_{b,g}}{2}\right)
-
\psi\!\left(\frac{\alpha_{b,g}}{2}\right)
\right\}
-
\frac{\alpha_{b,g}}{2\lambda_{b,g}}\mathcal R_{b,g} \\
&\quad-
\frac12
\Bigg\{
\operatorname{tr}\!\left(
\tilde{\mathbf A}_{b,g}\mathbf V_{b,g}^{-1}
\right)
+
\left(
\mathbf m_{b,g}-\tilde{\boldsymbol\beta}_{b,g}
\right)^\top
\tilde{\mathbf A}_{b,g}
\left(
\mathbf m_{b,g}-\tilde{\boldsymbol\beta}_{b,g}
\right) \\
&\hspace{28mm}
-k
+
\log|\mathbf V_{b,g}|
-
\log|\tilde{\mathbf A}_{b,g}|
\Bigg\} \\
&\quad-
\frac{\tilde\nu_{b,g}}{2}
\log\!\left(
\frac{\lambda_{b,g}}{\tilde\delta_{b,g}}
\right)
-
\log\Gamma\!\left(
\frac{\tilde\nu_{b,g}}{2}
\right)
+
\log\Gamma\!\left(
\frac{\alpha_{b,g}}{2}
\right) \\
&\quad-
\frac{\alpha_{b,g}-\tilde\nu_{b,g}}{2}
\psi\!\left(
\frac{\alpha_{b,g}}{2}
\right)
+
\frac{\alpha_{b,g}}{2}
-
\frac{\alpha_{b,g}\tilde\delta_{b,g}}{2\lambda_{b,g}}
+
\log\tilde\omega_{b,g}
-
\log\omega_{b,g}
\Bigg] \\
&\quad-
\frac12
\Bigg[
\operatorname{tr}\!\left(
\boldsymbol\Phi_0^{(b-1)}\mathbf D_b^{-1}
\right)
+
\left(
\mathbf c_b-\boldsymbol\phi_0^{(b-1)}
\right)^\top
\boldsymbol\Phi_0^{(b-1)}
\left(
\mathbf c_b-\boldsymbol\phi_0^{(b-1)}
\right) \\
&\hspace{24mm}
-p
+
\log|\mathbf D_b|
-
\log|\boldsymbol\Phi_0^{(b-1)}|
\Bigg].
\end{aligned}
\end{equation*}

\subsection{Full-Joint Boundary Treatment}
\label{app:elbo_fj}

We use the FJ factors and moments defined in Section~\ref{app:variational_fj}. The factorization over $G_b$ gives
\begin{equation*}
\begin{aligned}
\mathbb E_{\tilde q_{\mathrm{FJ}}^{(b)}}
\!\left[\log\tilde q_{\mathrm{FJ}}^{(b)}\right]
&=\sum_{g\in\mathcal G_b}\omega_{b,g}\log\omega_{b,g}
+\sum_{g\in\mathcal G_b}\omega_{b,g}
\mathbb E_{q_{\beta,g}^{(b)}}\!\left[\log q_{\beta,g}^{(b)}\right]\\
&\quad+\sum_{g\in\mathcal G_b}\omega_{b,g}
\mathbb E_{q_{\sigma^2,g}^{(b)}}\!\left[\log q_{\sigma^2,g}^{(b)}(\sigma^2)\right]
+\mathbb E_{q_\phi^{(b)}}\!\left[\log q_\phi^{(b)}(\boldsymbol\phi)\right].
\end{aligned}
\end{equation*}
Set $\ell_{b,g}=\log(\lambda_{b,g}/2)-\psi(\alpha_{b,g}/2)$ and $\kappa_{b,g}=\alpha_{b,g}/\lambda_{b,g}$. Using $\mathcal R_{b,g}^{\mathrm{FJ}}$ from Section~\ref{app:variational_fj}, the expected log-likelihood for a regression-reset candidate is
\begin{equation*}
\begin{aligned}
&\mathbb E_{q_{\beta,g}^{(b)}q_{\sigma^2,g}^{(b)}q_\phi^{(b)}}
\!\left[\log p\!\left(\mathbf z_b\mid\mathbf W_b,
\boldsymbol\theta_b,\sigma^2,\boldsymbol\phi\right)\right]\\
&\quad=-\frac{n_b^\ast}{2}\log(2\pi)
-\frac{n_b^\ast}{2}\ell_{b,g}
-\frac{\kappa_{b,g}}{2}\mathcal R_{b,g}^{\mathrm{FJ}}.
\end{aligned}
\end{equation*}
The same expression holds for a regression-retention candidate with its corresponding $\mathcal R_{b,g}^{\mathrm{FJ}}$.

For a regression-reset candidate, the Gaussian expectations are
\begin{equation*}
\begin{aligned}
\mathbb E_{q_{\beta,g}^{(b)}}
\!\left[\log q_{\beta,g}^{(b)}(\boldsymbol\theta_b)\right]
&=-k\log(2\pi)+\frac12\log|\mathbf H_{b,g}|-k,\\
\mathbb E_{q_{\beta,g}^{(b)}}
\!\left[\log\tilde\pi_{\beta,g}^{(b),\mathrm{FJ}}(\boldsymbol\theta_b)\right]
&=-k\log(2\pi)+\frac12\log|\mathbf P_{0,b,g}|\\
&\quad-\frac12\operatorname{tr}(\mathbf P_{0,b,g}\boldsymbol\Sigma_{b,g})\\
&\quad-\frac12(\mathbf m_{b,g}-\boldsymbol\mu_{0,b,g})^\top
\mathbf P_{0,b,g}(\mathbf m_{b,g}-\boldsymbol\mu_{0,b,g}).
\end{aligned}
\end{equation*}
Hence,
\begin{equation*}
\begin{aligned}
\mathcal K_{\beta,b,g}^{\mathrm{FJ}}
&:=\operatorname{KL}\!\left(
q_{\beta,g}^{(b)}(\boldsymbol\theta_b)
\,\middle\|\,
\tilde\pi_{\beta,g}^{(b),\mathrm{FJ}}(\boldsymbol\theta_b)\right)\\
&=\frac12\Big[
\operatorname{tr}(\mathbf P_{0,b,g}\boldsymbol\Sigma_{b,g})
+(\mathbf m_{b,g}-\boldsymbol\mu_{0,b,g})^\top
\mathbf P_{0,b,g}(\mathbf m_{b,g}-\boldsymbol\mu_{0,b,g})\\
&\hspace{18mm}-2k+\log|\mathbf H_{b,g}|-\log|\mathbf P_{0,b,g}|
\Big].
\end{aligned}
\end{equation*}
For a regression-retention candidate, the corresponding calculation in Section~\ref{app:elbo_pm} gives
\begin{equation*}
\begin{aligned}
\mathcal K_{\beta,b,g}^{\mathrm{FJ}}
&:=\operatorname{KL}\!\left(q_{\beta,g}^{(b)}(\boldsymbol\beta_b)
\,\middle\|\,q_{\beta,r_g}^{(b-1)}(\boldsymbol\beta_b)\right)\\
&=\frac12\Big[
\operatorname{tr}(\mathbf A_{0,r_g}^{(b-1)}\mathbf V_{b,g}^{-1})\\
&\quad+(\mathbf m_{b,g}-\boldsymbol\beta_{0,r_g}^{(b-1)})^\top
\mathbf A_{0,r_g}^{(b-1)}
(\mathbf m_{b,g}-\boldsymbol\beta_{0,r_g}^{(b-1)})\\
&\quad-k+\log|\mathbf V_{b,g}|
-\log|\mathbf A_{0,r_g}^{(b-1)}|\Big].
\end{aligned}
\end{equation*}
The terms $\mathcal K_{\sigma^2,b,g}$ and $\mathcal K_{\phi,b}$ have the expressions in Section~\ref{app:elbo_pm}, evaluated at the FJ variational and inherited parameters.

With
\begin{equation*}
J_{b,g}^{\mathrm{FJ}}
=-\frac{n_b^\ast}{2}\log(2\pi)
-\frac{n_b^\ast}{2}\ell_{b,g}
-\frac{\kappa_{b,g}}{2}\mathcal R_{b,g}^{\mathrm{FJ}}
-\mathcal K_{\beta,b,g}^{\mathrm{FJ}}
-\mathcal K_{\sigma^2,b,g},
\end{equation*}
the batchwise ELBO is
\begin{equation*}
\mathcal L_b^{\mathrm{FJ}}(q)
=\sum_{g\in\mathcal G_b}\omega_{b,g}
\left[J_{b,g}^{\mathrm{FJ}}+\log\tilde\omega_{b,g}-\log\omega_{b,g}\right]
-\mathcal K_{\phi,b}.
\end{equation*}
The FJ ELBO is evaluated on the full candidate approximation before marginalization and pruning.

\section{Additional Ablation and Sensitivity Experiments}
\label{sec:supp_reset_ablation}

We further examine SR-SVB through three complementary sensitivity
analyses. First, we compare the PM and FJ boundary treatments under the
main sixteen-scenario design. Second, we compare separate and joint-only
resetting under the block-specific changes used in the main simulation.
Third, we consider a targeted simultaneous-change scenario in which
both parameter blocks change at the same boundary.

\subsection{Sensitivity to the Boundary Treatment: PM versus FJ}
\label{sec:supp_pm_fj}

We compare the PM boundary treatment used in the main text with the
FJ treatment developed in Sections~\ref{app:variational_fj} and~\ref{app:elbo_fj}. Both variants use
the same SR-SVB framework and are evaluated on the same sixteen
changepoint scenarios and Monte Carlo replications, with all
non-boundary settings unchanged.

Figure~\ref{fig:pm_fj_metrics} compares recall, precision, $F_1$,
changepoint location error $L$, and detection delay $D$ for the overall,
regression, and innovation-variance results, using the definitions and
matching rules in Section~\ref{subsec:sec4.2} of the main text. The two variants show
similar detection performance, with overall recall of approximately
$94$--$100\%$. The reported $L$ and $D$ values remain below $0.02$
batches for both variants across all three summaries. These errors
and delays are conditional on matched events, and FJ does not
uniformly improve on PM across scenarios.
Figure~\ref{fig:pm_fj_runtime} compares fitting time per dataset.
FJ-SVB has a higher mean fitting time in all sixteen scenarios.
For regression-reset candidates, FJ updates a $2k$-dimensional joint
Gaussian instead of the $k$-dimensional factor used by PM.
With $k=30$, the regression update therefore increases from 30 to
60 dimensions, and cross-covariance terms also enter the variance
and AR updates. These additional matrix and moment calculations
increase the cost of the local batchwise updates.
\begin{figure}[t!]
\centering
\begin{subfigure}{\textwidth}
\centering
\includegraphics[width=0.98\linewidth]
{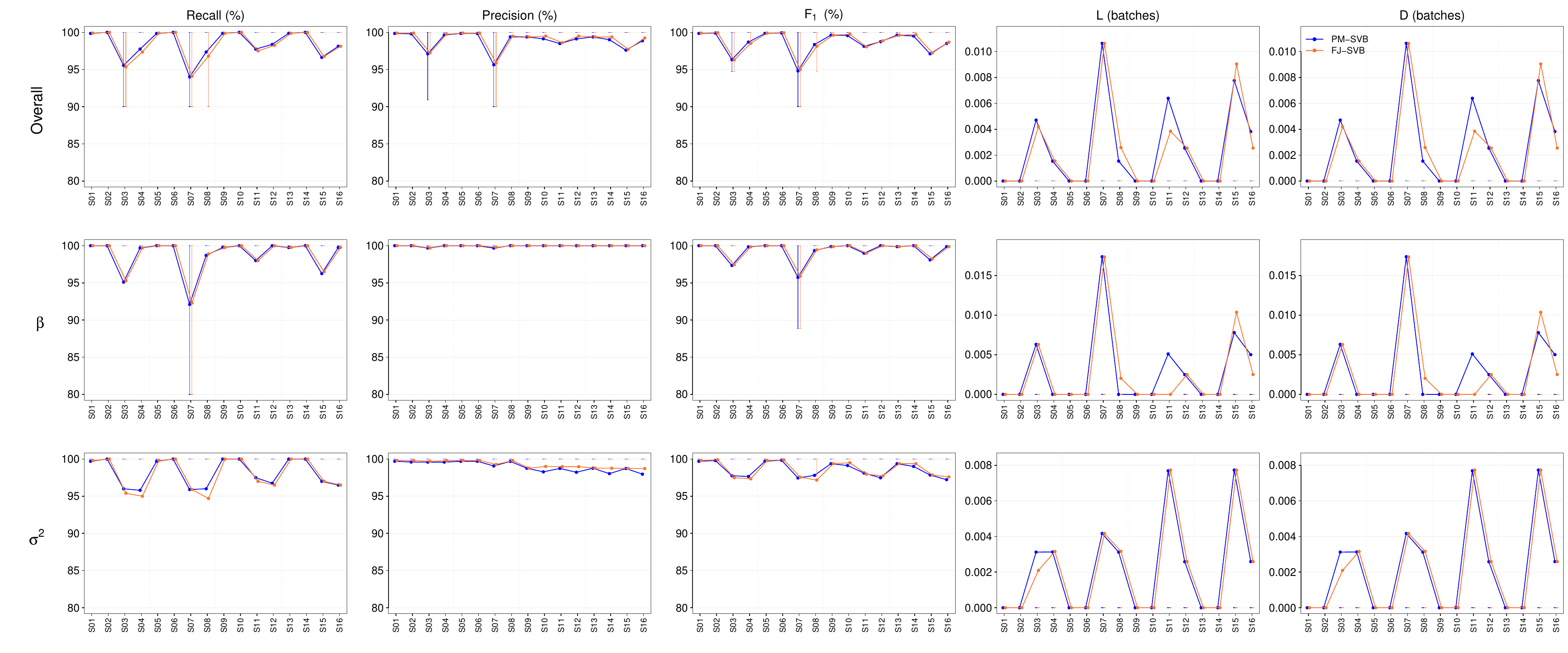}
\caption{Pooled detection metrics; vertical bars show the
replicate-level interquartile range.}
\label{fig:pm_fj_metrics}
\end{subfigure}

\vspace{0.6em}

\begin{subfigure}{0.5\textwidth}
\centering
\includegraphics[width=\linewidth]
{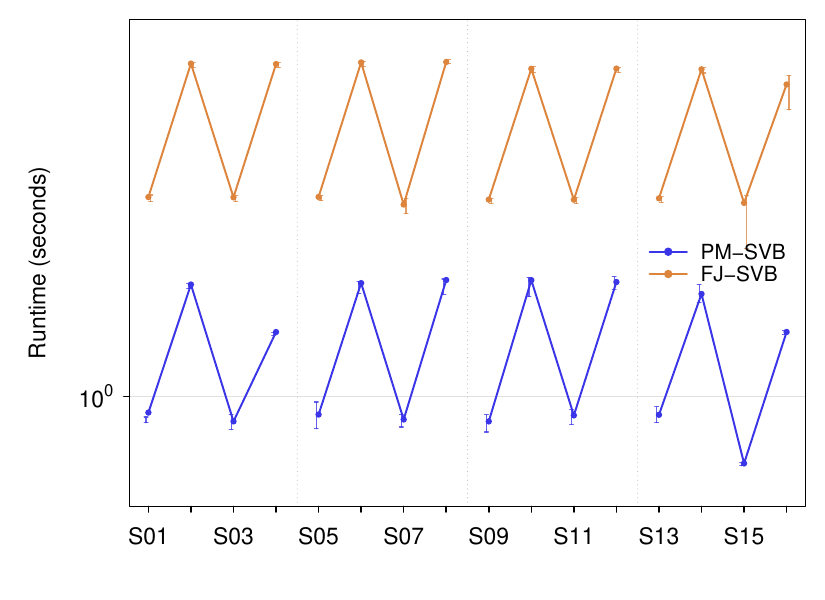}
\caption{Mean fitting time per dataset on a logarithmic scale;
vertical bars show the interquartile range.}
\label{fig:pm_fj_runtime}
\end{subfigure}
\caption{Comparison of the PM and FJ boundary treatments across the
sixteen changepoint scenarios.}
\label{fig:pm_fj_comparison}
\end{figure}
Figure~\ref{fig:fj_coverage} reports the batchwise empirical 95\%
coverage trajectories under FJ-SVB, using the same coverage construction as in Section~\ref{subsec:sec4.2} of the main text. For $\boldsymbol{\beta}$, FJ-SVB yields coverage slightly closer to the nominal level than PM-SVB, particularly under the Large-AR scenarios.

\begin{figure}[t!]
\centering
\includegraphics[width=0.98\textwidth]
{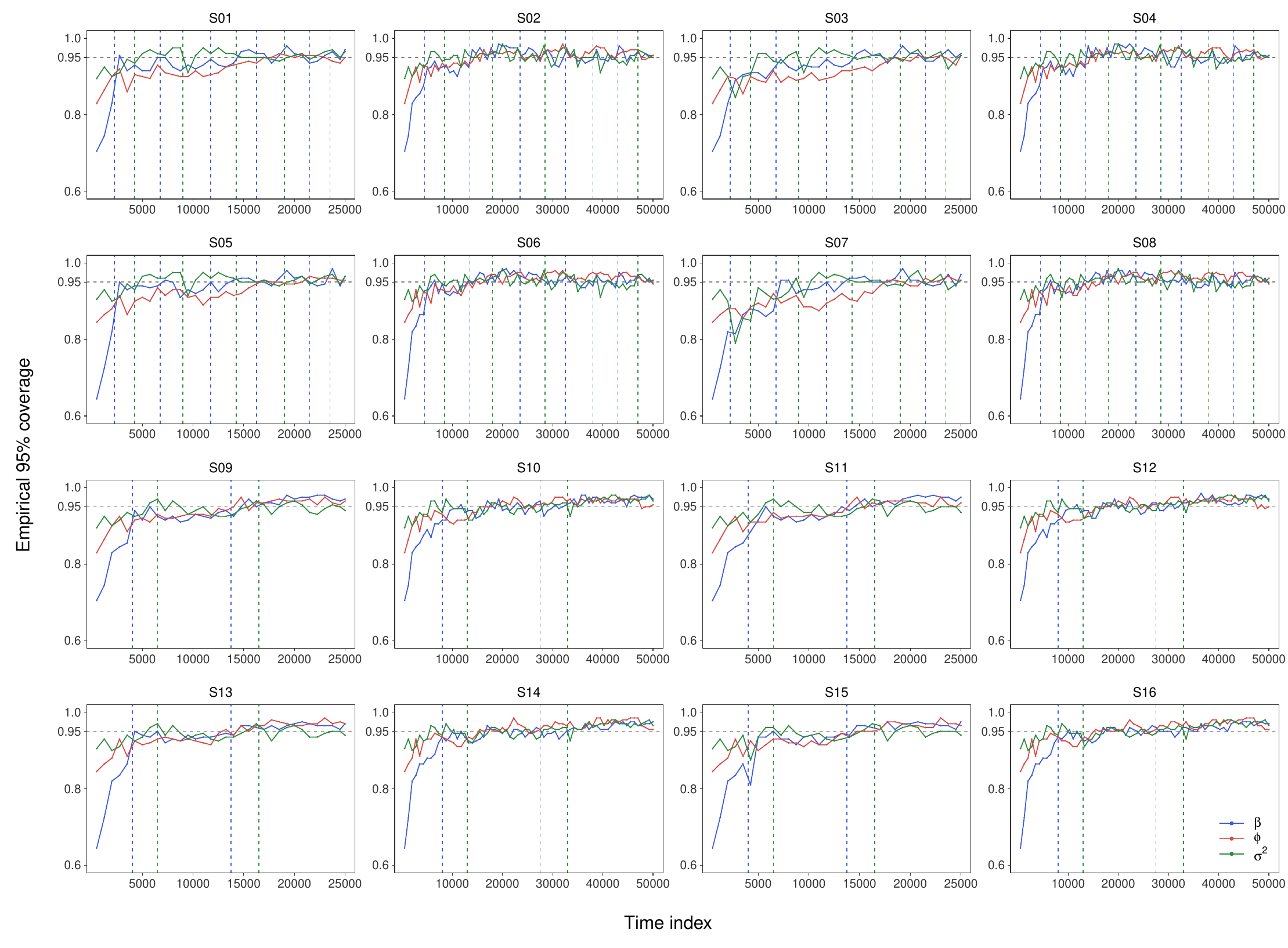}
\caption{Empirical 95\% coverage under FJ-SVB for
$\boldsymbol\beta$, $\boldsymbol\phi$, and $\sigma^2$.}
\label{fig:fj_coverage}
\end{figure}

The similar detection performance and lower fitting cost support PM
as the default boundary treatment in these experiments. FJ provides
an alternative that explicitly retains uncertainty about the preceding
regression state during the local boundary update.

\subsection{Block-Specific Changes: Separate Versus Joint Resetting}
\label{sec:supp_joint_reset_ablation}

We compare SR-SVB with a restricted joint-reset SVB that permits
only $(z_\beta,z_{\sigma^2})\in\{(0,0),(1,1)\}$.
Both methods use the PM boundary treatment and share the remaining
settings, including the AR factor, discount strengths, pruning rule,
batch structure, and convergence criterion.
Both procedures assign prior probability $0.99$ to no reset, so the
joint-reset procedure assigns the remaining $0.01$ entirely to its
single reset state and is not disadvantaged by a smaller overall
reset probability.
Both methods are evaluated on the same S01--S16 datasets and $R=200$
replications described in Section~\ref{subsec:sec4.2} of the main text.
Because joint-reset SVB cannot distinguish regression from variance
changes by construction, detection is evaluated type-agnostically.
Its thresholds are calibrated under the same $5\%$ stream-level
false-alarm budget as in the main experiment, and the independent
validation false-alarm probabilities range from $2\%$ to $6\%$.

Table~\ref{tab:block_specific_detection} summarizes the detection results.
Across 22,400 true changes, joint-reset SVB attains $58.54\%$ recall,
compared with $98.99\%$ for separate SR-SVB.
The difference is almost entirely due to variance changes:
regression-change recall is $99.60\%$ versus $99.92\%$, whereas
variance-change recall is $17.48\%$ versus $98.06\%$.
Table~\ref{tab:block_specific_postchange} reports estimation error
and coverage at the first and fourth batches following each true change.
\begin{table}[t]
\centering
\caption{Block-specific reset ablation pooled across S01--S16.}
\label{tab:block_specific_ablation}

\begin{subtable}{\linewidth}
\centering
\caption{Type-agnostic detection.}
\label{tab:block_specific_detection}
\small
\begin{tabular}{lrrrrrr}
\toprule
Method & TP & FP & FN & Recall (\%) & Precision (\%) & $F_1$ (\%) \\
\midrule
Joint-reset SVB
& 13,113 & 0 & 9,287 & 58.540 & 100.000 & 73.849 \\
Separate SR-SVB
& 22,174 & 85 & 226 & 98.991 & 99.618 & 99.304 \\
\bottomrule
\end{tabular}
\end{subtable}

\vspace{0.8em}

\begin{subtable}{\linewidth}
\centering
\caption{Post-change estimation error and empirical $95\%$ coverage.}
\label{tab:block_specific_postchange}
\small
\setlength{\tabcolsep}{3.8pt}
\begin{tabular}{cclrrrr}
\toprule
& & & \multicolumn{2}{c}{Estimation error}
& \multicolumn{2}{c}{Coverage} \\
\cmidrule(lr){4-5}\cmidrule(lr){6-7}
True change & \shortstack{Post-change\\batch} & Method
& $\beta$ & $\sigma^2$ & $\beta$ & $\sigma^2$ \\
\midrule
$\beta$ only
& First  & Joint-reset & 0.545 & 0.136 & 0.909 & 0.918 \\
&        & Separate    & 0.534 & 0.036 & 0.934 & 0.937 \\
& Fourth & Joint-reset & 0.278 & 0.063 & 0.935 & 0.944 \\
&        & Separate    & 0.278 & 0.028 & 0.942 & 0.949 \\
\addlinespace
$\sigma^2$ only
& First  & Joint-reset & 0.196 & 0.950 & 0.935 & 0.081 \\
&        & Separate    & 0.148 & 0.166 & 0.945 & 0.931 \\
& Fourth & Joint-reset & 0.211 & 0.519 & 0.941 & 0.319 \\
&        & Separate    & 0.134 & 0.072 & 0.942 & 0.948 \\
\bottomrule
\end{tabular}
\end{subtable}
\end{table}
As shown in Table~\ref{tab:block_specific_postchange}, after a
regression-only change, the two methods have similar regression
estimation errors. However, joint-reset SVB also discounts
the inherited variance information when resetting the regression
block, even though the true variance has not changed.
At the first post-change batch, the variance estimation error is
$0.136$ under joint-reset SVB, compared with $0.036$ under separate
SR-SVB. By the fourth post-change batch, the corresponding errors
are $0.063$ and $0.028$.
After a variance-only change, joint-reset SVB cannot discount the
inherited variance information without also discounting the accumulated
information about all $k=30$ unchanged regression coefficients.
The resulting loss of concentration in the regression distribution
can reduce support for the joint-reset candidate, making variance-only
changes harder to detect.
Its lower variance-change recall is accompanied by larger variance
estimation errors and substantial undercoverage.
At the first post-change batch, the variance estimation error is
$0.950$ under joint-reset SVB, compared with $0.166$ under separate
SR-SVB, and variance coverage is $8.1\%$ versus $93.1\%$.
By the fourth post-change batch, variance coverage remains at
$31.9\%$ under joint-reset SVB, compared with $94.8\%$ under separate
SR-SVB.
These results are consistent with the joint-reset constraint
hindering adaptation to variance-only changes.
The separate $(1,0)$ and $(0,1)$ states instead allow information
about the changed block to be discounted without also discounting
the unchanged block.

\subsection{Sensitivity to Simultaneous Regression and Variance Changes}
\label{sec:supp_joint_change}

To examine simultaneous changes, which are absent from S01--S16,
we construct a targeted scenario, J01, from S03.
This setting combines the smaller sample size $n=25{,}000$ with
Large-AR errors and weak signals, providing a challenging case for
post-change adaptation.
The fixed parameter states, AR coefficients, covariate realizations,
and standardized innovations are retained, but the five regression
transitions and five variance transitions are paired to produce five
simultaneous $(\beta,\sigma^2)$ changes.
With $R=200$, this gives 1,000 true joint changes.

We compare the default SR-SVB, a balanced-prior variant, and the
restricted joint-reset SVB, labeled Default, Balanced-prior, and
Joint-reset in Table~\ref{tab:joint_change_sensitivity}.
Under the default independent reset prior,
$P(1,1)=h_\beta h_{\sigma^2}\approx2.5\times10^{-5}$.
The balanced specification retains the total reset probability of
$0.01$ but allocates it equally across the three reset types:
$P(0,0)=0.99$ and
$P(1,0)=P(0,1)=P(1,1)=0.01/3$.
The default and joint-reset specifications reuse the frozen thresholds
from the $n=25{,}000$, Large-AR calibration stratum.
For the balanced-prior specification, thresholds are recalibrated
using the same null calibration streams and frozen before J01
evaluation. Its independent validation false-alarm probability is $3\%$.

Table~\ref{tab:joint_change_detection} reports exact-type detection
using the episode definitions and matching rules in Section~\ref{subsec:sec4.2}
of the main text.
A true positive requires a matched joint alarm; all unmatched
episodes, including single-channel alarms, count as false positives.
For the four-state specifications, an episode is labeled joint when
both calibrated reset channels are active at its first batch.
Ignoring alarm type, all three specifications detect every change
at the first post-change batch.
Requiring joint alarms, however, gives recalls of $84.3\%$ and
$95.6\%$ for the default and balanced-prior specifications,
respectively, compared with $100\%$ for joint-reset SVB.
The corresponding $F_1$ scores are $82.65\%$, $94.70\%$, and $100\%$.
For matched joint events, $L=D=0$ under all three specifications.
\begin{table}[t]
\centering
\caption{Sensitivity to simultaneous regression and variance changes
in J01.}
\label{tab:joint_change_sensitivity}

\begin{subtable}{\linewidth}
\centering
\caption{Exact-type joint-change detection.}
\label{tab:joint_change_detection}
\small
\begin{tabular}{lrrrrrr}
\toprule
Specification & TP & FP & FN
& Recall (\%) & Precision (\%) & $F_1$ (\%) \\
\midrule
Default
& 843 & 197 & 157 & 84.300 & 81.058 & 82.647 \\
Balanced-prior
& 956 & 63 & 44 & 95.600 & 93.817 & 94.700 \\
Joint-reset
& 1,000 & 0 & 0 & 100.000 & 100.000 & 100.000 \\
\bottomrule
\end{tabular}
\end{subtable}

\vspace{0.8em}

\begin{subtable}{\linewidth}
\centering
\caption{Post-change estimation error and empirical $95\%$ coverage.}
\label{tab:joint_change_postchange}
\small
\setlength{\tabcolsep}{4pt}
\begin{tabular}{lcrrrr}
\toprule
& & \multicolumn{2}{c}{Estimation error}
& \multicolumn{2}{c}{Coverage} \\
\cmidrule(lr){3-4}\cmidrule(lr){5-6}
Specification & \shortstack{Post-change\\batch}
& $\beta$ & $\sigma^2$ & $\beta$ & $\sigma^2$ \\
\midrule
Default
& First & 0.615 & 0.322 & 0.707 & 0.678 \\
Balanced-prior
& First & 0.603 & 0.208 & 0.774 & 0.793 \\
Joint-reset
& First & 0.579 & 0.117 & 0.936 & 0.936 \\
\addlinespace
Default
& Fourth & 0.331 & 0.083 & 0.905 & 0.879 \\
Balanced-prior
& Fourth & 0.322 & 0.066 & 0.922 & 0.919 \\
Joint-reset
& Fourth & 0.318 & 0.057 & 0.938 & 0.941 \\
\bottomrule
\end{tabular}
\end{subtable}
\end{table}
Table~\ref{tab:joint_change_postchange} reports estimation error and
coverage at the first and fourth batches following each true change,
regardless of detection outcome.
Joint-reset SVB has the smallest estimation errors and coverage
closest to the nominal level at both reported batches.
The balanced-prior specification improves on the default, reaching
regression and variance coverages of $92.2\%$ and $91.9\%$ by the
fourth post-change batch, compared with $93.8\%$ and $94.1\%$
under joint-reset SVB.

Joint resetting performs best in J01, where its restriction matches
the true change structure.
In contrast, Section~\ref{sec:supp_joint_reset_ablation} shows its
limitations under variance-only changes.
The balanced-prior specification improves joint detection and early
posterior adaptation in this experiment without restricting the
four-state model to joint resets.

\section{Details of the Changepoint Experiment}
\label{app:changepoint_details}

This section provides the fixed changepoint schedules, implementation
details, and calibration and matching rules for the experiment in
Section~\ref{subsec:sec4.2} of the main text.

\paragraph{Fixed simulation design.}
For each value of $N_{\mathrm{cp}}$, the same changepoint schedule is
used across the two AR specifications and signal-strength settings.
For $n=25{,}000$, the schedules are
\begin{equation}
\label{eq:fixed_cp_schedules}
\begin{aligned}
\mathcal T_{10}
&=(2250,4250,6750,9000,11750,14250,16250,19000,21500,23500),\\
\mathcal T_{4}
&=(4000,6500,13750,16500).
\end{aligned}
\end{equation}
For $n=50{,}000$, each location in
\eqref{eq:fixed_cp_schedules} is multiplied by two.
Regression and variance changes alternate, beginning with a regression
change. A changepoint $\tau$ is the final observation under the
preceding parameter state, with the new state applying from $\tau+1$.
Parameter paths are fixed across replications. Within each replication,
covariate and innovation draws are shared across scenarios; different
replications use independent draws.

\paragraph{SR-SVB implementation and uncertainty quantification.}
Detection and calibration use the pre-pruning marginal reset
probabilities, whereas uncertainty summaries use the retained
post-pruning approximation. The regression region uses the mean and
covariance of the retained Gaussian mixture, with its Mahalanobis
cutoff obtained from draws from that mixture. The AR region uses the
shared Gaussian factor, and the variance interval is the equal-tailed
interval of the retained inverse-gamma mixture. Coverage of the
nominal $95\%$ regions is evaluated across all replications at each
fitted batch, regardless of detection outcome.

\paragraph{FOCuS implementation.}
FOCuS receives the raw response and covariates. It applies a
multivariate Gaussian detector to standardized batchwise least-squares
estimates and a Gamma detector to residual sums of squares, with one
update per input batch. To limit the cost of full-dimensional
convex-hull pruning for the $k=30$ regression channel, we use the
package's low-dimensional projection approximation described by
\citet{romano2026focus}. Pruning uses all singleton coordinates and
all circular adjacent-coordinate pairs, while the Gaussian detector
input remains $k$-dimensional.

Let $\mathcal H$ denote the reference sample of $m$ observations,
with design matrix $\mathbf X_{\mathcal H}$, response
$\mathbf y_{\mathcal H}$, and ordinary-least-squares estimate
$\widehat{\boldsymbol\beta}_{\mathcal H}$. Define
\begin{equation}
\label{eq:focus_standardization}
\begin{aligned}
s_{\mathcal H}^{2}
&=\frac{\|\mathbf y_{\mathcal H}
-\mathbf X_{\mathcal H}\widehat{\boldsymbol\beta}_{\mathcal H}\|_2^2}{m-k},
&
\mathbf Q_{\mathcal H}
&=\frac{\mathbf X_{\mathcal H}^{\top}\mathbf X_{\mathcal H}}{m},
&
\mathbf U^\top\mathbf U
&=\frac{\mathbf Q_{\mathcal H}}{s_{\mathcal H}^{2}}.
\end{aligned}
\end{equation}
Diagonal regularization is added when required for the Cholesky
factorization. For batch $b$, let
$\widehat{\boldsymbol\beta}_b$ be its least-squares estimate and
$S_b=\|\mathbf y_b-\mathbf X_b\widehat{\boldsymbol\beta}_b\|_2^2$.
The standardized summary is
$\mathbf z_b=\mathbf U(\widehat{\boldsymbol\beta}_b
-\widehat{\boldsymbol\beta}_{\mathcal H})$.
This transformation is fixed within each monitoring segment and
standardizes coefficient coordinates without filtering temporal
dependence.

The Gaussian detector receives $n_b\mathbf z_b$ with weight $n_b$;
the Gamma detector receives $S_b$ with weight $n_b-k$ and base-shape
parameter $1/2$. Both use the unknown-baseline formulation. Reference
summaries enter the detectors, but alarm evaluation begins only after
the reference period.

\paragraph{fxmonitor implementation.}
Let $e_t=y_t-\mathbf x_t^\top\widehat{\boldsymbol\beta}_{\mathcal H}$
and $\widehat v=m^{-1}\sum_{t\in\mathcal H}e_t^2$, with the reference
coefficient estimate held fixed throughout the monitoring segment.
The supplied score vector and covariance estimate are
\begin{equation}
\label{eq:fxmonitor_scores}
\begin{aligned}
\boldsymbol\psi_t
&=\begin{pmatrix}
\mathbf x_t e_t\\
e_t^2-\widehat v
\end{pmatrix},\\
\mathbf J_0
&=\begin{pmatrix}
\displaystyle\frac{1}{m}\sum_{t\in\mathcal H}
e_t^2\mathbf x_t\mathbf x_t^\top & \mathbf 0\\[2mm]
\mathbf 0^\top &
\displaystyle\frac{1}{m}\sum_{t\in\mathcal H}(e_t^2-\widehat v)^2
\end{pmatrix},
&
\mathbf J
&=\mathbf J_0+10^{-10}a\mathbf I_{k+1},
\end{aligned}
\end{equation}
where $a$ is the mean of the positive diagonal entries of
$\mathbf J_0$. The block-diagonal covariance preserves covariance
among regression score coordinates without mixing them with the
variance score. No lagged score products are included.

Numbering observations from the beginning of the reference sample,
the cumulative process and package boundary are
\begin{equation}
\label{eq:fxmonitor_process}
\begin{aligned}
\mathbf W_j
&=m^{-1/2}\mathbf J^{-1/2}\sum_{\ell=1}^{j}\boldsymbol\psi_\ell,
&
g_j&=c\,\frac{j}{m},
\end{aligned}
\end{equation}
where $\mathbf J^{-1/2}$ is the symmetric inverse square root and
$c$ is the package critical value at nominal level $0.05$.
The regression and variance monitoring statistics are
$\max_{1\leq i\leq k}|W_{j,i}|/g_j$ and $|W_{j,k+1}|/g_j$,
respectively, permitting channel-specific calibration while preserving
the time dependence of the package boundary.

Both frequentist procedures initially use 500 observations as
reference. After an alarm, both channels suspend monitoring while
four new complete 250-observation batches are collected to form the
next reference sample. For FOCuS, these batches enter separately with
their own regression estimates and residual sums of squares, while
their pooled observations determine the transformation in
\eqref{eq:focus_standardization}. Alarms are evaluated at complete
batch endpoints over the common horizons of $25{,}000$ and
$50{,}000$ observations.

\paragraph{Calibration details.}
Each of the four sample-size--AR strata uses 400 no-change streams
for calibration and 100 independent streams for validation, generated
from the corresponding initial parameter state. All methods use the
same null datasets and batch grid. Each cutoff pair is shared across
the signal strengths and changepoint counts within its stratum.
For FOCuS and fxmonitor, channelwise maxima are computed along
uninterrupted initial-reference monitoring paths, so calibration
concerns the first alarm rather than alarms after a restart.
SR-SVB uses the maximum pre-pruning reset probability in each channel.

For each $d=1,\ldots,20$, candidate cutoffs are based on the
$(d+1)$st largest calibration maximum in each channel.
FOCuS and fxmonitor use strict exceedance. SR-SVB implements
the same tie exclusion by storing the next representable value
above the order statistic and using a non-strict comparison.
For each method, we select the largest $d$ for which at most
20 of the 400 calibration streams signal in either channel,
counting each stream once. The four cutoff pairs are frozen
before validation and change-stream evaluation.

Table~\ref{tab:cp_calibration_validation} reports the selected
cutoffs and false-alarm counts. Observed calibration rates range from
$4.75\%$ to $5.00\%$, within the prescribed $5\%$ budget; independent
validation rates range from $2\%$ to $11\%$.
The numerical cutoffs are on method-specific scales.

\begin{table}[!t]
\centering
\small
\setlength{\tabcolsep}{3.5pt}
\renewcommand{\arraystretch}{1.08}
\caption{Channel-specific cutoffs and stream-level false alarms.
CAL and VAL show counts, with percentages in parentheses.}
\label{tab:cp_calibration_validation}
\begin{tabular}{llrrrr}
\toprule
Stratum & Method & $\beta$ cutoff & $\sigma^2$ cutoff & CAL & VAL \\
\midrule
\multirow[t]{3}{*}{25K--Large}
& fxmonitor & 3.75244 & 3.95107 & 20/400 (5.00) & 7/100 (7.00) \\
& FOCuS & 451.386 & 64.2814 & 20/400 (5.00) & 11/100 (11.00) \\
& SR-SVB & $4.71\times 10^{-7}$ & 0.241800
& 20/400 (5.00) & 9/100 (9.00) \\
\midrule
\multirow[t]{3}{*}{50K--Large}
& fxmonitor & 3.78887 & 3.95107 & 20/400 (5.00) & 7/100 (7.00) \\
& FOCuS & 499.119 & 70.5484 & 20/400 (5.00) & 9/100 (9.00) \\
& SR-SVB & 0.00141726 & 0.386276 & 20/400 (5.00) & 4/100 (4.00) \\
\midrule
\multirow[t]{3}{*}{25K--Small}
& fxmonitor & 1.56491 & 1.71475 & 20/400 (5.00) & 5/100 (5.00) \\
& FOCuS & 111.836 & 12.7386 & 19/400 (4.75) & 4/100 (4.00) \\
& SR-SVB & $2.77\times 10^{-7}$ & 0.238682
& 20/400 (5.00) & 10/100 (10.00) \\
\midrule
\multirow[t]{3}{*}{50K--Small}
& fxmonitor & 1.56491 & 1.73006 & 20/400 (5.00) & 6/100 (6.00) \\
& FOCuS & 117.086 & 14.0410 & 19/400 (4.75) & 2/100 (2.00) \\
& SR-SVB & 0.00185596 & 0.378505 & 20/400 (5.00) & 5/100 (5.00) \\
\bottomrule
\end{tabular}
\end{table}

\paragraph{Alarm episodes and reported changepoints.}
For SR-SVB, each uninterrupted run of batches with at least one
active reset channel forms an episode. Its first batch determines
the alarm batch and type: regression, variance, or joint according
to the channels active in that batch.
FOCuS and fxmonitor report the first threshold crossing in each
monitoring segment and restart as described above.

SR-SVB and fxmonitor report the boundary immediately preceding the
alarm batch. FOCuS maps its retrospective changepoint estimate to
an input-batch boundary using cumulative batch sizes for the regression
channel and cumulative residual degrees of freedom for the variance
channel. Locations are left undefined when this mapping is not unique
within numerical tolerance.

For exact-type one-to-one matching, true changes are processed
chronologically. Let $b_j$ denote the last batch before change $j$
and $b_{j+1}$ the last batch before the next change of either type.
For the final change, $b_{j+1}$ is the final evaluation batch.
Each change is matched to the earliest unused episode of the same
type whose alarm batch $a$ and reported boundary $\hat b$ satisfy
$b_j<a\leq b_{j+1}$ and $b_j\leq\hat b<b_{j+1}$.
Unmatched episodes, including those with undefined locations, count
as false positives; unmatched true changes count as false negatives.

As a sensitivity check across S01--S16, we retained the original
exact-type SR-SVB matches and counted the 55 additional within-episode alarmed
batches as false positives. Overall precision decreased from
$98.89\%$ to $98.65\%$ and $F_1$ from $98.58\%$ to $98.46\%$.

\section{Real Data Application Details}
\label{app:btc_details}

This section describes preprocessing and reporting for the Bitcoin
realized-volatility application in Section~\ref{sec:sec5} of the main text.

\paragraph{Data source and return construction.}
We use the Bitstamp BTC/USD data described in the main text,
from January~1, 2014 through November~26, 2025.
After sorting timestamps in UTC and retaining the last record
at duplicate timestamps, we obtain five-minute closes from the
last available one-minute close in each interval.
Five-minute log returns are computed between adjacent closes
on the regular grid, without interpolating missing prices.
For each non-overlapping 30-minute interval, realized variance is
$RV_t=\sum_{i=1}^{6}r_{t,i}^2$, requiring all six returns.
Incomplete intervals at the sample boundaries are removed,
while internal missing intervals are retained for the treatment
described below.
\paragraph{Warm-up adjustment and missing values.}
The warm-up sample runs from January~1, 2014 at 00:30 UTC to
January~1, 2015 at 00:30 UTC, excluding the latter endpoint.
The offset $c$ is $1\%$ of the median warm-up realized variance,
and $m_{s(t)}$ is the warm-up median of $\log(RV+c)$ for the
corresponding weekday and half-hour slot. These quantities remain
fixed during fitting, and warm-up observations initialize the lagged
predictors.
Internal missing intervals remain on the regular grid. Their log
realized variances are replaced by $m_{s(t)}$, giving $y_t=0$.
These completed responses enter the working likelihood and weekly
response averages, while missing realized variances remain missing
in the return-composition predictors. The single internal gap
contains 39 half-hour intervals on March~15, 2025.

\paragraph{Trading activity and return composition.}
Let $N_t$ be the sum of one-minute close prices times BTC volume
within interval $t$, using available contributions. Relative log
notional volume $v_t$ is $\log(1+N_t)$ minus its trailing median
at the same weekday and half-hour over the preceding eight weeks,
requiring at least four valid values.
Define $R_t=\sum_{i=1}^{6}r_{t,i}$,
$RV_t^-=\sum_{i=1}^{6}r_{t,i}^2\mathbf{1}\{r_{t,i}<0\}$, and
$BV_t=(\pi/2)\sum_{i=1}^{6}|r_{t,i}r_{t,i-1}|$, where
$r_{t,0}=r_{t-1,6}$.
Bipower products therefore include the half-hour boundary and require
at least five valid products per interval. The 24-hour summaries
use the preceding 48 intervals, exclude the current interval,
and require complete values for each series entering the ratio.
Table~\ref{tab:btc_predictors} summarizes the predictor definitions,
where $\epsilon>0$ is a small constant used in ratio denominators.

\paragraph{Macro variables and the trading-hours indicator.}
The macro controls represent daily return signals for E-mini Nasdaq-100
futures and the U.S.\ dollar index, and the sign-reversed daily return signal
for Treasury-note futures as a yield-change proxy rather than a
directly observed yield change. Each 30-minute observation uses
the value for the most recent retained date strictly preceding
its UTC date, with values held constant within each UTC day.

Daily return signals are obtained by median-centering standardized
return increments and summing them by UTC date.
For each series, dates with numerically negligible
sums are excluded; thus, ``previous-day'' need not refer to the
preceding calendar day. The centering medians and numerical tolerances
are computed from the full series, whereas final predictor centering
and scaling use the warm-up sample.
The U.S.-hours indicator equals one for 30-minute intervals
overlapping regular U.S.\ equity trading hours, 9:30 a.m.--4:00 p.m.
New York time on Monday--Friday, and zero otherwise.
Daylight saving time is accounted for, but exchange holidays are
not separately excluded. Missing indicator values are set to zero.

\begin{table}[!t]
\centering
\small
\setlength{\tabcolsep}{4pt}
\renewcommand{\arraystretch}{1.12}
\caption{Response and predictor definitions and interpretations before
warm-up standardization.}
\label{tab:btc_predictors}
\begin{tabular}{
@{}
>{\raggedright\arraybackslash}p{0.30\linewidth}
>{\raggedright\arraybackslash}p{\dimexpr0.70\linewidth-2\tabcolsep\relax}
@{}
}
\toprule
Variable & Definition and interpretation \\
\midrule

\texttt{y}
& $y_t=\log(RV_t+c)-m_{s(t)}$, the seasonally adjusted log
realized variance. \\

\texttt{har\_week}
& $\frac{1}{336}\sum_{j=1}^{336}y_{t-j}$, the weekly average
of completed responses, representing volatility persistence. \\

\texttt{rel\_log\_notional\_lag1}
& $v_{t-1}$, lagged log notional volume relative to its trailing
same-time-of-week baseline. \\

\texttt{ret\_30m\_lag1}
& $R_{t-1}$, the direction and magnitude of the preceding
half-hour price return. \\

\texttt{downside\_share\_24h\_lag1}
& $\frac{\sum_{j=1}^{48}RV_{t-j}^-}
{\sum_{j=1}^{48}RV_{t-j}+\epsilon}-\frac12$, the share of
negative-return variation over the preceding 24 hours, minus one half. \\

\texttt{jump\_share\_24h\_lag1}
& $\frac{\max\{\sum_{j=1}^{48}(RV_{t-j}-BV_{t-j}),0\}}
{\sum_{j=1}^{48}RV_{t-j}+\epsilon}$, clipped to $[0,1]$.
Excess realized variance over bipower variation serves as a proxy
for jump-like movements over the preceding 24 hours. \\

\texttt{trend\_ratio\_lag1}
& $\frac{R_{t-1}^2}{6RV_{t-1}+\epsilon}$, clipped to $[0,1]$,
measuring net directional movement relative to within-interval variation. \\

\texttt{NQ\_prev\_day}
& Daily return signal for E-mini Nasdaq-100 futures for the previous
retained date. \\

\texttt{DXY\_prev\_day}
& Daily return signal for the U.S.\ dollar index for the previous
retained date. \\

\texttt{Yield\_prev\_day}
& Sign-reversed daily return signal for Treasury-note futures for the previous
retained date; a yield-change proxy rather than
a directly observed yield change. \\

\texttt{us\_open}
& Indicator for the weekday U.S.\ trading-hours window. \\

\texttt{rel\_volume\_x\_us\_open}
& Product of $v_{t-1}$ and the current U.S.-hours indicator,
formed before standardization, allowing the activity--volatility
association to differ during those hours. \\

\bottomrule
\end{tabular}
\end{table}

\paragraph{Standardization and interaction.}
Continuous predictors, including the weekly HAR term and the
interaction, are centered at their warm-up medians and scaled by
$\max\{1.4826\,\mathrm{MAD},\mathrm{SD},10^{-8}\}$ to avoid excessively small scaling factors.
Standardized values are clipped to $[-12,12]$, with missing values
set to zero. The intercept and \texttt{us\_open} remain unscaled.
The interaction is formed from $v_{t-1}$ and \texttt{us\_open}
before standardization, so the U.S.-hours coefficient represents
the fitted contrast at $v_{t-1}=0$, the trailing activity baseline.

\paragraph{Reset probabilities and retrospective summaries.}
For descriptive reporting, marginal and joint reset probabilities
are obtained by summing posterior mixture weights over the
corresponding reset configurations. A channel alarm is declared when its marginal
probability is at least $0.50$. Alarms are evaluated at batch ends,
whereas displayed boundaries mark batch starts. For retrospective
regime summaries, consecutive exceedances within each channel are
grouped and represented by the batch with the largest reset
probability. Coefficient and variance summaries use the posterior
at the final batch of each displayed regime.

\end{document}